\documentclass[9.5pt,journal,compsoc]{IEEEtran}
\usepackage{amsmath, amsfonts, amsthm, amssymb} 
\usepackage{mathrsfs}
\usepackage{cite}
\usepackage{algorithm}
\usepackage{algorithmic}
\usepackage{array}
\newtheorem{thm}{Theorem}

\newtheorem{lem}{Lemma}
\newtheorem{defi}{Definition}
\DeclareMathOperator*{\argmax}{argmax}
\usepackage{graphicx}
\usepackage{subfigure}
\usepackage{cases}
\usepackage{xcolor}
\usepackage{enumerate}
\usepackage{multirow}
\usepackage{multicol}
\usepackage{setspace}
\usepackage{enumitem}
\usepackage{setspace}
\usepackage[square, comma, sort&compress, numbers]{natbib}
\usepackage{url}
\allowdisplaybreaks[3]

\begin{document}

\title{Neighborhood Matters: Influence Maximization in Social Networks with Limited Access}

\author{Chen Feng, 
        Luoyi Fu,  
        Bo Jiang,
        Haisong Zhang,
        Xinbing Wang, 
        Feilong Tang
        and Guihai Chen
\IEEEcompsocitemizethanks{\IEEEcompsocthanksitem This work was supported by National Key R\&D Program of China 2018YFB2100302, NSF China under Grant (No. 61822206, 61960206002, 61832013, 62041205, 61532012), Tencent AI Lab Rhino-Bird Focused Research Program JR202034. The authors are with the School of Electronic Information and Electrical Engineering, Shanghai Jiao Tong University, Shanghai, 20040, China.\protect\\
Email: \{fengchen, yiluofu, bjiang, xwang8\}@sjtu.edu.cn, \{tang-fl, gchen\}@cs.sjtu.edu.cn, hansonzhang@tencent.com.
\protect\\
 We would like to express our special thanks to Chengyang Wu and Xudong Wu for their helpful discussion and assistance.
}
}


\markboth{IEEE Transactions on Knowledge and Data Engineering, ~Vol.~XX, No.~X, August~20XX}%
{Shell \MakeLowercase{\textit{et al.}}: Bare Demo of IEEEtran.cls for Computer Society Journals}

\IEEEtitleabstractindextext{
\begin{abstract}

  Influence maximization (IM) aims at maximizing the spread of influence by offering discounts to influential users (called seeding). In many applications, due to user's privacy concern, overwhelming network scale etc., it is hard to target any user in the network as one wishes. Instead, only a small subset of users is initially accessible. Such access limitation would significantly impair the influence spread, since IM often relies on seeding high degree users, which are particularly rare in such a small subset due to the power-law structure of social networks.

  \quad In this paper, we attempt to solve the limited IM in real-world scenarios by the adaptive approach with seeding and diffusion uncertainty considered. Specifically, we consider fine-grained discounts and assume users accept the discount probabilistically. The diffusion process is depicted by the independent cascade model. To overcome the access limitation, we prove the set-wise friendship paradox (FP) phenomenon that neighbors have higher degree in expectation, and propose a two-stage seeding model with the FP embedded, where neighbors are seeded. On this basis, for comparison we formulate the non-adaptive case and adaptive case, both proven to be NP-hard. In the non-adaptive case, discounts are allocated to users all at once. We show the monotonicity of influence spread w.r.t. discount allocation and design a two-stage coordinate descent framework to decide the discount allocation. In the adaptive case, users are sequentially seeded based on observations of existing seeding and diffusion results. We prove the adaptive submodularity and submodularity of the influence spread function in two stages. Then, a series of adaptive greedy algorithms are proposed with constant approximation ratio.  Extensive experiments on real-world datasets show that our adaptive algorithms achieve larger influence spread than non-adaptive and other adaptive algorithms (up to a maximum of 116\%).
\end{abstract}

}

\maketitle

\IEEEraisesectionheading{\section{Introduction}\label{sec:introduction}}
  \IEEEPARstart{T}{he} last two decades have witnessed the dramatic development of social networks (e.g., Facebook, Twitter), which have become an important platform for the promotion of ideas, behaviors and products. For example, viral marketing is a widely adopted strategy in the promotion of new products. The company selects some users and provides them with some discounts within a predefined budget, hoping that the product will be known by more users via the ``word-of-mouth'' effect. This demand naturally raises the influence maximization problem, which aims at triggering the largest cascade of influence by allocating discounts to users (called seeding). Since the seminal work of Kempe \textit{et al.} \cite{seminal_IM}, numerous efforts have been made to advance the research of the influence maximization problem  \cite{competitive_IM}\cite{scalable_LT} \cite{mining} \cite{data_based_IM} \cite{Budgetd_IM} \cite{a_martingale_approach} \cite{continuous_IM} \cite{mobihoc2017}.

  In most cases, the problem is solved under the implicit assumption that the company can allocate discounts to any user in the network. However, in the real-world commercial campaign, the company often only has access to a small subset of users. For example, an online merchant wishes to promote a new product by providing samples to influential users. In practice, the merchant could only mail samples to customers who have left address information before in ways such as buying products, applying for membership. Similarly, in many other applications, due to privacy concern, network scale etc., the seeding process is limited to a small sample of the network (like most work, the network structure and the diffusion probabilities are assumed pre-known and are not of concern). Due to the power-law degree distribution of social networks, high degree users are particularly rare in the small subset. Since influence maximization often relies on seeding many high degree users (not necessarily the highest ones), the access limitation would evidently impair the influence spread. Regarding this concern, initial attempts \cite{locally_two_stage} \cite{WWW_Singer} \cite{knapsack_Singer} \cite{FOCS_adaptive_seeding} seed the subset of users to reach neighbors who are voluntary to join in the campaign automatically with the intuition of the friendship paradox (FP) phenomenon \cite{Feld_FP} which reveals that the degree of your neighbor is greater than yours in expectation. As pioneers in the access limitation problem, \cite{locally_two_stage} \cite{WWW_Singer} \cite{knapsack_Singer} \cite{FOCS_adaptive_seeding} have largely expanded the influence spread but still fall short of dealing with the uncertainties in real promotions. Such uncertainties include two aspects.
  (1) \textit{Seeding uncertainty}: a targeted user will not necessarily become the seed if the discount is not satisfactory. (2) \textit{Diffusion uncertainty}: due to the strength of social relationships or characteristics of users, the influence propagation between two users is not assured to be successful. Existing works distribute discounts to users all at once, without considering whether the actual seeding and diffusion is successful. Such fixed strategy (referred to as ``non-adaptive'' method) is vulnerable to the uncertainty in the seeding and diffusion process, resulting in unsatisfactory influence spread. Thus, we are motivated to study an adaptive \footnote{We mean ``adaptive'' in the sense that users are sequentially seeded based on previous seeding and diffusion results, while the concept in \cite{locally_two_stage} \cite{WWW_Singer} \cite{knapsack_Singer} \cite{FOCS_adaptive_seeding} means that the allocation in neighbors is related to the set of seeds in initially accessible users.} approach, where users are sequentially seeded based on previous seeding and diffusion results.


  Specifically, the problem is investigated under the following settings. Suppose a company wants to promote a new product through a social network by providing discounts for users. Due to the difficulty in collecting user's information, only a small subset of users is initially accessible, denoted as $X$. We consider a fine-grained discount setting, i.e., discounts take value from $[0,1]$ instead of only 0 or 1 in previous works on this problem. Accordingly, a user probabilistically accepts the discount and becomes a seed, from which the diffusion starts. The diffusion process is depicted by the widely acknowledged independent cascade model. Since social networks leave traces of behavioral data which allow observing and tracking, the spread of influence could be easily observed. For example, from one's social account, we can see whether the user adopts the product. On this basis, we make the first attempt to solve the limited influence maximization problem by the adaptive approach. Moreover, we investigate the non-adaptive method under this setting for comparison. When undertaking this study, we find it is challenging in the following three aspects.

  \textbf{Access Limitation:} Under our setting, the influence spread may suffer more from the access limitation due to user's uncertain nature. Thus, there is an increasing demand on a new seeding model to address the access limitation with user's probabilistic behavior considered. We attempt to design an effective model with natural intuition and theoretical support.

  \textbf{Fine-grained Discount:} With a larger discount space, the scale of possible discount allocations accordingly becomes greater in order sense. Thus, it is much harder to find an effective discount allocation. Meanwhile, the corresponding seeding uncertainty imposes additional difficulties on influence analysis and algorithm design.

  \textbf{Algorithm Design:} In our seeding model, initially accessible users are seeded to reach their neighbors for further discount allocation. It is easy to see that the two seeding processes are inter-related. This interdependency requires our algorithm to collectively consider the two seeding processes. Not only seeding results of initial users should be considered but also possible seeding results of neighbors.

  To overcome the access limitation, we intend to leverage the FP phenomenon that neighbors have higher degree in expectation. A new seeding model is proposed with the FP embedded. We first seed users in $X$ (stage 1) to reach their neighbors for further discount allocation (stage 2). Whether a user $u$ accepts the discount $c_u$ is depicted by the seed probability function $p_u(c_u)$. Accordingly, we formulate non-adaptive and adaptive cases, both proven to be NP-hard. (1) \textit{Non-adaptive case}: discounts are allocated to users in $X$ all at once and then neighbors of those who accept the discount. (2) \textit{Adaptive case}: we sequentially seed users in $X$ by adopting actions, defined as user-discount pairs, based on previous seeding and diffusion results. Each time, if the user accepts the discount, we further seed his/her neighbors.
  The main contributions of this paper are highlighted as follows.


  \textbullet\quad We first formulate the limited IM problem under fine-grained discounts and uncertain user nature. Then, we look into the FP phenomenon in the set of users $X$ and prove it to hold set-wisely in any network. With this theoretical support, we are inspired to design a two-stage seeding model where neighbors of $X$ are seeded, and thus we get access to more influential users and the influence spread is also expanded.

  \textbullet\quad In the non-adaptive case, we first show the monotonicity of influence spread w.r.t. the discount allocation. Then, a two-stage coordinate descent framework is designed to decide the fine-grained discount in two stages. We collectively consider the seeding process in stage 2 when designing the discount allocation of stage 1.

  \textbullet\quad In each round of the adaptive case, we calculate the benefit of each action (user-discount pair) by estimating the increase of the influence spread when neighbors of the user are seeded. Then, we adopt the action with the largest benefit-to-discount ratio. With this idea, two algorithms $\pi^{\text{greedy}}$ and $\pi^{\text{greedy}}_{\text{discrete}}$ are devised with performance guarantee, where discounts of neighbors are determined by the coordinate descent and greedy algorithm respectively. Furthermore, $\pi^{\text{greedy}}_{\text{enum}}$ is proposed with the enumeration idea, achieving an approximation ratio of $(1-e^{-\frac{B_1-1}{B_1}(1-\frac{1}{e})})\approx 0.469$ better than previous work \footnote{Although \cite{WWW_Singer} achieves an approximation ratio of $1-\frac{1}{e}$, its diffusion model, called the voter model, is much simpler where the influence of each user is simply additive and could be anticipated in advance.} \cite{locally_two_stage} \cite{knapsack_Singer} \cite{FOCS_adaptive_seeding}.

  \textbullet\quad  Besides theoretical guarantees, our algorithms also exhibit favorable performances in experiments. Specifically, extensive experiments are conducted on real-world social networks to evaluate the proposed algorithms. The results show that our adaptive algorithms achieve much larger influence spread than non-adaptive and other adaptive algorithms (up to a maximum of \textbf{116\%}), and meanwhile are scalable to large networks.

  The rest of the paper is organized as follows. Important milestones are reviewed in Section \ref{RW}. We describe the  model in Section \ref{model}. The non-adaptive and adaptive cases are formulated in Section \ref{PR}. In Section \ref{NIM}, we analyze the non-adaptive case. And the adaptive case is studied in Section \ref{AIM}. Numerous experiments are conducted in Section \ref{experiments}. And we conclude the paper in Section \ref{conclusion}. Due to space limitation, some technical proofs and experimental results are deferred to the supplemental file.

\section{Related Work}\label{RW}
    Domingos and Richardson \cite{mining} took the head in exploring the peer influence among customers by modeling markets as social networks. They attempted to maximize the cascade of purchasing behavior by targeting a subset of users. Kempe \textit{et al.} \cite{seminal_IM} further formulated the well-known influence maximization problem. The greedy algorithm is proposed with proven performance guarantee $(1-e^{-1})$ \cite{An_analysis}. No polynomial time algorithm could have a better performance, unless $P=NP$ \cite{Feige}. Since then, extensive works have been done in different perspectives \cite{negative_IM} \cite{scalable_LT} \cite{data_based_IM} \cite{IRIE} \cite{cost_effective}.

    Bharathi \textit{et al.} analyzed the game of competing information diffusions in one social network \cite{competitive_IM}. Li \textit{et al.} considered location information in influence maximization \cite{location_aware}. Tang \textit{et al.} presented the near-optimal time algorithm to solve the influence maximization problem for triggering models without hurting the approximation ratio \cite{time_complexity_practical_efficiency}. Estimation techniques are applied in \cite{a_martingale_approach} to improve the empirical efficiency. Yang \textit{et al.} \cite{continuous_IM} assumed that the discount can be fractional instead of only 0 or 1. Chen \textit{et al.} \cite{efficient_IM} improved the influence maximization in both running time and cascade size. Budgeted influence maximization was studied in \cite{Budgetd_IM} where each user is associated with a cost for selection. A synthetic survey is provided in \cite{survey}.

    However, the access limitation is ignored by all the above works. Seeman \textit{et al.} \cite{FOCS_adaptive_seeding} studied the seed selection based on the intuition of FP, which is first discovered by Feld \cite{Feld_FP} and further investigated point-wisely by numerous works \cite{General_FP} \cite{multistep_FP} \cite{lattanzi_FP}. Badanidiyuru \textit{et al.} considered monotone submodular objective functions and achieve a $(1-e^{-1})^2$ approximation ratio \cite{locally_two_stage}. In \cite{WWW_Singer}, by relaxing the diffusion model and discount setting, Horel \textit{et al.} made impressive progress on algorithm design and experimental validation.

    Adaptive seeding is an emerging topic in influence maximization. Users are seeded one after another based on the existing seeding and diffusion results. Golovin \textit{et al.} \cite{adaptive_submodularity_golovin} studied the adaptive submodularity and showed that a greedy policy obtains a $(1-e^{-1})$ approximation ratio. Yuan and Tang proposed the adaptive algorithm based on the greedy policy \cite{mobihoc2017}. In \cite{no_time_observe}, only partial feedback is observed before seeding the next user. 

    Our work is distinct from existing works \cite{locally_two_stage} \cite{WWW_Singer} \cite{knapsack_Singer} \cite{FOCS_adaptive_seeding} mainly in two aspects. First, the problem is comprehensively studied under practical settings, where users accept the fine-grained discounts probabilistically and the diffusion process is depicted by the well-received independent cascade model. On this basis, algorithms are proposed with theoretical guarantee and evaluated on real-world datasets. Second, all existing solutions to the access limitation can be classified as non-adaptive category. However, not only non-adaptive solution but also adaptive solution is presented in our paper to maximize the influence spread under access limitation. 


\section{Model}\label{model}
    A social network is denoted as the graph $G(V,E)$, where $V$ is the set of users and $E$ records the relationships between users. Initially accessible users are denoted as $X \subseteq V$. For any user $u \in V$, let $N(u)$ denote the neighborhood of $u$. For any subset of users $T \subseteq V$, $N(T)$ represents the neighborhood of $T$, defined as $N(T) \equiv \bigcup\limits_{u \in T} N(u)\setminus T$.

    We start the influence diffusion by allocating discounts to users, which is described by the two-stage seeding model. If a user accepts the discount allocated, we say it becomes a \textit{seed}. A set of seeds forms a \textit{seed set}. Especially, we call the seed in stage 1 an \textit{agent}. The diffusion process is described by the independent cascade model.

   \textbf{The two-stage seeding model} contains \textit{the recruitment stage} (stage 1) and \textit{the trigger stage} (stage 2). The predefined budget is $B_1$ in stage 1 and $B_2$ in stage 2
   \footnote{It is tempting to study our problem under variable $B_1$ and $B_2$, which however makes our problem even more challenging. It is also worth noting that our current solution could provide valuable hints to the case of unfixed $B_1$ and $B_2$. Specifically, we allow an error of $c$ when optimizing $B_1$ ($B_2=B-B_1$) by requiring $B_1$ to take discrete values $\{c, 2c, 3c, \dots, [B/c]\!\cdot c\}$. Then, the problem is reduced to the current one where $B_1$ and $B_2$ are pre-known. Thus we could derive the optimal combination of $B_1$ and $B_2$ by invoking our solutions $[B/c]$ times. Besides, we would also like to provide a heuristic method with user's degree as the proxy of its influence like \cite{WWW_Singer}. To explain, we figure out the average degree of $X$ and $N(X)$, and then determine $B_1$, $B_2$ according to the proportion of their average degrees.\label{footnote_fix_budget}}
   , with the total budget being $B=B_1+B_2$. The seeding process in the two-stage model is shown as follows.
   \begin{enumerate}
     \item[$\bullet$] \textit{the recruitment stage}: We seed users in $X$ with given budget $B_{1}$, to recruit agents whose friends are possibly more influential. Some users become the agent and bring their friends (newly reachable users) into this campaign in the meantime by forwarding the promotion link, providing address information for mail, even helping hand over the sample, etc..
     \item[$\bullet$] \textit{the trigger stage}: We seed newly reachable users with budget $B_{2}$ to trigger the largest influence cascade.
   \end{enumerate}
\begin{figure}[h]
\centering
\vspace{-0.3cm}
\includegraphics[width=0.33\textwidth]{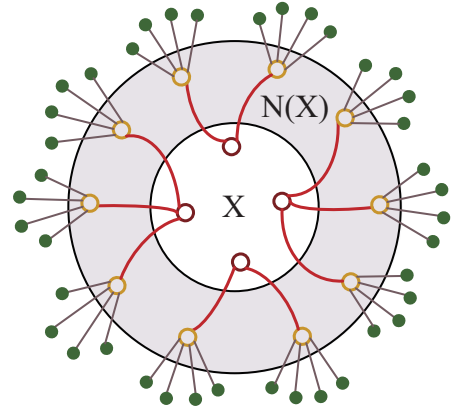}\vspace{-0.2cm}
\caption{An example of the two-stage seeding model. We first seed users in $X$ and then reach their neighbors $N(X)$ for further seeding.}\vspace{-0.3cm}\label{fig1}
\end{figure}


   The discount allocated to user $u$ is denoted as $c_u$. Whether user $u$ accepts the discount is depicted by the seed probability function $p_{u}(c_{u})$. We assume $p_{u}(\cdot)$ satisfies the four properties: (1) $p_{u}(0)=0$; (2) $p_{u}(1)=1$; (3) $p_{u}(\cdot)$ is monotonically nondecreasing; (4) $p_{u}(\cdot)$ is continuously differentiable.

   In both stages, we allocate the budget to seed users, but the intention is different. In \textit{the recruitment stage}, users are seeded to recruit agents such that we can access their influential neighbors. In \textit{the trigger stage}, where the influence diffusion really begins, we seed the newly reached neighbors to maximize the influence diffusion. In classical influence maximization, we are restricted to seed users only in $X$ without the privilege to reach more influential users, resulting in limited cascade size.

   \textbf{The diffusion model} we adopt is the \textit{independent cascade model}, which has been one of the most popular models since formally proposed in \cite{seminal_IM}. In the independent cascade model, each edge $(u,v)$ in the graph $G(V,E)$ is associated with a propagation probability $p_{uv}$, indicating the probability that $u$ influences $v$ after $u$ gets influenced first. Note that, no matter $u$ succeeds or not, it cannot attempt to influence $v$ again. The influence of different edges is independent of each other. Once user $v$ gets influenced, it will not change its state. On this basis, starting from a seed set $T$, the diffusion expands in discrete steps as follows. In each step, the newly influenced user, $u$, tries to influence its neighbors, e.g. $v$, along the edge $(u,v)$, and succeeds with corresponding probability $p_{uv}$. The process goes on in this way, until no user is further influenced. The influence of the seed set $T$ is denoted as $I(T)$, which is the expected number of users totally influenced, where $I(\cdot)$ is the influence function $I: 2^V \rightarrow \mathbf{R}$.
   As can be seen, the influence spread is closely related to the graph and propagation probabilities along edges. Thus, in computation of $I(S)$, we need full knowledge of the influence graph, i.e., $G(V,E)$ and $p_{uv}$ ($\forall (u,v)\in E$).

   To provide theoretical support for our model, by Lemma \ref{FP} we prove that the neighbors $N(X)$ have higher degree than the small set of users $X$ in expectation. To the best of our knowledge, this is the first time that the FP phenomenon is proven to hold set-wisely in any large network. It is easy to see that the traditional point-wise FP is a special case of Lemma \ref{FP} by letting $X$ contain only one user.
   %
   \begin{lem}\label{FP}
       Given any connected network $G(V,E)$, consider a set of users $X$ where each user is selected from $V$ with probability $p\rightarrow 0$, then the friendship paradox phenomenon exists between $X$ and $N(X)$, i.e., the average degree of $X$ is no larger than that of $N(X)$ in expectation.
   \end{lem}
   \begin{proof}
   Please refer to Section 2.1 of the supplemental file.
\end{proof}

\begin{table}[t]
\caption{Frequently Used Notations}\label{notations}
\vspace{-0.3cm}
\begin{center}
\renewcommand\arraystretch{1.1}
\footnotesize
\begin{tabular}{|p{1.6cm}|p{6.4cm}|}
\hline
\multicolumn{2}{|l|}{Model}\\
\hline
$G(V,E)$ & A directed graph, where $V$ and $E$ represent users and the relationships respectively\\
\hline
$N(T)$   & The neighborhood of $T$\\
\hline
$X$      & Initially reachable users\\
\hline
$p_{u}(c_{u})$  &   The probability that user $u$ accepts discount $c_u$\\
\hline
$I(T)$   & The expected number of users influenced by seed set $T$\\
\hline
\multicolumn{2}{|l|}{Non-adaptive Case} \\
\hline
$S$      & The set of users who becomes agents in X\\
\hline
$C_1$ (or $C_2$)  &    The discount allocation in stage 1 (or stage 2)\\
\hline
$P_{r}(S;C_1,X)$ & The probability that users in $S$ become agents if $C_1$ is carried out in $X$\\
\hline
$Q(C_2; N(S))$  &  The expected number of users influenced if $C_2$ is carried out in given $N(S)$\\
\hline
$f(C_1;X)$    &    The expected number of users influenced if $C_1$ is carried out in $X$\\
\hline
\multicolumn{2}{|l|}{Adaptive Case}\\
\hline
$D$      &  The set of discounts that can be adopted\\
\hline
$Y$      &  The action space in stage 1\\
\hline
$\psi$   &  The action adopted with sequence in stage 1 ($dom(\psi)$ denotes the set of actions)\\
\hline
$\lambda$ & The seeding realization in stage 2\\
\hline
$\phi$   &  The diffusion realization\\
\hline
$\pi$    &  The adaptive seeding policy\\
\hline
$\hat{\sigma}(\psi ,(\lambda_p,\phi_p))$ &  The expected number of users influenced by $\pi$ under realizations $\lambda$, $\phi$\\
\hline
$\Delta(y|\psi)$  & The marginal benefit brought by $y$\\
\hline
$R$  &  Newly reachable users in stage 2 in each round\\
\hline
$Z$  &  The action space in stage 2 in each round\\
\hline
\end{tabular}
\end{center}
\vspace{-0.3cm}
\end{table}
\normalsize
  \section{Problem Formulation}\label{PR}

    From the model, we can see that the influence maximization in the two-stage model consists of three sequential processes: seeding in stage 1, seeding in stage 2 and influence diffusion from stage 2. In this section, we will study two ways of seeding: the non-adaptive case and the adaptive case. In the first case, a process goes on after the previous process is finished. In the second case, the three processes iteratively go on in a circle. In each round, only part of the process is done and users are seeded based on the result of previous seeding and diffusion processes.

    \subsection{Non-adaptive Case}
    We first seed users in $X$ to recruit agents. The discount allocation is denoted as the $m$-dimensional vector $C_1=(c_1, c_2, \cdots, c_m)$, where $m=|X|$. Let $S$ denote users who become agents in $X$. Then, $N(S)$ denotes friends newly reached in $N(X)$. We next allocate discounts to newly reachable users $N(S)$. Similarly, the discount allocation in $N(S)$ is denoted as $C_2=(c_1, c_2, \cdots, c_k)$, where $k=|N(S)|$. $\forall c_i$ in $C_1$ or $C_2$, $c_i \in [0,1]$, i.e. the discounts are fractional.
    It is easy to see that the seed set in each stage is probabilistic, since users accept discounts with probability function $p_u(c_u)$.

    Let us consider stage 1, given the allocation $C_1$ in $X$, the probability that the subset of users $S \subseteq X$ accepts the discounts is
\begin{equation}\label{eq1}
P_{r}(S;C_1,X)=\prod_{u\in S}p_{u}(c_u)\prod_{v\in X \setminus S}(1-p_{v}(c_v)).
\end{equation}
    Following the same technique, given allocation $C_2$ in $N(S)$, the probability that $T \subseteq N(S)$ becomes the seed set in stage 2 is
\begin{equation}\label{eq2}
P_{r}(T;C_2,N(S))=\prod_{u\in T}p_{u}(c_u)\prod_{v\in N(S) \setminus T}(1-p_{v}(c_v)).
\end{equation}
    Then, the influence diffusion starts from the seed set $T$ and $I(T)$ users get influenced in expectation. In stage 2, given newly reachable users $N(S)$ and the discount allocation $C_2$, we can obtain the expected number of users influenced, formally expressed in Equation (\ref{eq3}),
\begin{equation}\label{eq3}
Q(C_2; N(S))=\sum_{T \subseteq N(S)}P_{r}(T;C_2,N(S))I(T).
\end{equation}
    Under a fixed $N(S)$, we only need to find the optimal allocation $C_2$ to maximize $Q(C_2; N(S))$ with budget constraint $B_2$. Note that $N(S)$ is probabilistically determined by the discount allocation $C_1$ in stage 1. To maximize the influence spread, we have to first optimize the influence spread over $C_1$ with stage 2 collectively considered. Given initially accessible users $X$, the expected influence spread with regard to $C_1$ is
\begin{equation}\label{eq4}
f(C_1;X)=\sum_{S \subseteq X}P_{r}(S;C_1,X)\max Q(C_2; N(S)).
\end{equation}
      In summary, we optimize the influence spread over $C_1$ with possible seeding optimization in stage 2 considered. With the allocation $C_1$, we obtain newly reachable users $N(S)$ and further derive the optimal allocation $C_2$ to maximize the influence spread.



Under budget constraints $B_1$ and $B_2$, the non-adaptive influence maximization problem (NIM) can be formally formulated as follows.
\begin{equation}\label{NIM<}
\begin{aligned}
\text{NIM}:\\
&\text{max}\!\!  & & f(C_1;X)                   \!\!   &&\text{max}\!\!&&Q(C_2; N(S))\\
&\text{s.t.}\!\! & & \forall u\in X, c_u\!\in\! [0,1]\!\!  &&\text{s.t.}\!\!&&\forall u\in N(S), c_u\!\in\! [0,1]\\
&            & & \!\!\sum_{u\in X}c_u\leq B_1    \!\!      &&           && \!\!\sum_{u\in N(S)}\!\!\!\!c_u\leq B_2
\end{aligned}
\end{equation}

  To help the reader comprehend the idea of the non-adaptive case, we provide an example in Section 1.1 of the supplemental file.

    \subsection{Adaptive Case}
    In the adaptive case, users in stage 1 are seeded sequentially, instead of computing a discount allocation of $X$ all at once. The seeding process is defined on an action space $Y := X \times D$, where $D=\{d_{1}, d_{2}, \cdots, d_{l}\}$ is the set of $l$ discrete discount rates that can be adopted. $\forall d_i \in D, d_i\in [0,1]$ and $\max\{d_i\in D\}=1$. Selecting the action $y=(v(y), d(y)) \in Y$ means seeding user $v(y)$ with discount $d(y)$. Once $v(y)$ takes the discount, we reach its neighbors and carry out an allocation therein. Note that, the seeding processes in two stages are both based on previous diffusion results. We next introduce three basic concepts of our study.


\begin{defi}\label{def1}
\textbf{(Seeding Process $\boldsymbol{\psi}$).} $\psi$ sequentially records the actions adopted in stage 1 which aim at reaching influential users in $N(X)$. Let $dom(\psi)$ denote the set of actions without sequence.
\end{defi}
\begin{defi}\label{def2}
\textbf{(Seeding Realization $\boldsymbol{\lambda}$).} In stage 2, for each user $v$ in $N(X)$, $v$ will either accept the given discount $c_v\in [0,1]$ (denoted as ``1'') with probability $p_{v}(c_v)$, or reject it (denoted as ``0'') with probability $1-p_{v}(c_v)$. The decisions of users after being seeded in stage 2 are denoted by the function $\lambda: (N(X),[0,1]^{|N(X)|}) \rightarrow \{0,1\}$.

\end{defi}
\begin{defi}\label{def3}
\textbf{(Diffusion Realization $\boldsymbol{\phi}$).} For each edge $(u,v)\in E$, it is either in ``live'' state (denoted as ``1'') or in ``dead'' state (denoted as ``0''), indicating the influence through $(u,v)$ is successful or not. The states of edges are denoted by the function $\phi: E \rightarrow \{0,1\}$.
\end{defi}

    With the above preliminaries, we are ready to elaborate the sequential seeding process. Each time, we start from stage 1 and push the partial seeding process $\psi_p$ one step forward by adopting an action $y^*=(v(y^*),d(y^*))$ from $Y$ (i.e., probing user $v(y^*)$ with discount $d(y^*)$). If the user refuses the discount $d(y^*)$, we delete $y^*$ from action space $Y$ and move to the next round. Note that the budget is not wasted in this case, ``refuse'' means ``reuse''. For example, we instantiate discounts as vouchers. The discount will not be used (i.e., ``reuse'') if the user does not accept and apply the voucher to become a seed (i.e., ``refuse''). For the case of providing samples, we can forward information of the sample to inquire users' will. If the user is not satisfied (i.e., ``refuse''), we will not mail the sample to him/her and the budget is not wasted (i.e., ``reuse'').

    On the contrary, if the user $v(y^*)$ accepts $d(y^*)$ and becomes a seed, we can reach its neighbors in stage 2. Meanwhile, actions about $v(y^*)$ are abandoned and $d(y^*)$ is subtracted from $B_1$. Next we allocate discounts in newly reachable users $R$ with some budgets from $B_2$. Specifically, the budget drawn from $B_2$ only depends on the intrinsic property of $v(y^*)$, say the discount the user expects or the number of its neighbors. It is assumed that neighbors brought by each user in $X$ are non-overlapping, since users in $X$ are usually random. Otherwise, for a common neighbor, we can designate it as one user's neighbor. Some neighbors will accept the discounts and become seeds, forming the seed set. Then, a partial diffusion starts from the seed set and explores the state of edges exiting influenced users.

\begin{defi}\label{def4}
\textbf{(Adaptive Seeding Policy $\boldsymbol{\pi}$).} The adaptive policy is the function $\pi:\sigma(\psi_p ,(\lambda_p,\phi_p))\rightarrow Y$. Given observation $\sigma(\psi_p ,(\lambda_p,\phi_p))$, the policy $\pi$ will select an action from $Y$.
\end{defi}
    After each round of seeding, $\psi_p$ observes a partial realization $(\lambda_p,\phi_p)$. The set of users influenced under $\psi_p$ is denoted as $\sigma(\psi_p, (\lambda_p,\phi_p))$, and the number of users influenced is $\hat{\sigma}(\psi_p ,(\lambda_p,\phi_p))$. Due to the probability in users and edges, realizations $\lambda$ and $\phi$ are both probabilistic. The prior joint probability distribution of seeding realization and diffusion realization is assumed to be  $p((\lambda,\phi)):=P((\Lambda,\Phi)=(\lambda,\phi))$, where $\Lambda$ is a random seeding realization and $\Phi$ is a random diffusion realization. Under realizations $\lambda$ and $\phi$, the seeding process determined by $\pi$ is denoted as $\psi(\pi|(\lambda,\phi))$. The expected number of users influenced by $\pi$ is $\hat{\sigma}(\pi)=E[\hat{\sigma}(\psi(\pi |(\Lambda,\Phi)),(\Lambda,\Phi))]$, where the expectation is calculated with respect to $p((\lambda,\phi))$. We denote the discount allocated to user $u \in X$ under $\pi$ is $c(u|\pi)$. The influence maximization in the adaptive case (AIM) is a constrained optimization problem formulated as follows.
\begin{equation}\label{eq6}
\begin{aligned}
\text{AIM}:\;\;\;\\
& \text{max}& & \hat{\sigma}(\pi) \\
& \text{s.t.}& & \forall u \in X, c(u|\pi)\in D, \forall (\lambda,\phi)\\
&  & & \forall v \in N(X), c_v\in [0,1], \forall (\lambda,\phi)\\
&  & & \!\!\sum_{u \in X}c(u|\pi)\leq B_1, \forall (\lambda,\phi)\\
&  & & \!\!\!\!\!\sum_{v\in N(X)}\!\!\!\!c_v\leq B_2, \forall (\lambda,\phi)
\end{aligned}
\end{equation}

  For the reader's comprehension, we provide an illustrative example in Section 1.2 of the supplemental file.

\section{Non-adaptive Influence Maximization}\label{NIM}
Based on the non-adaptive problem formulated, in this section we look into the problem and present its solution. We first analyze the properties of NIM and transform the inequality constraints in Formula \eqref{NIM<} into equality constraints. Then, we present the two-stage coordinate descent algorithm to solve the problem.

  \subsection{Properties of NIM}

  We find that the optimization problem in NIM is NP-hard by Lemma \ref{le1}. Therefore, NIM can not be solved in polynomial time unless $P=NP$.
\begin{lem}\label{le1}
Finding the optimal discount allocation in two stages in NIM is NP-hard.
\end{lem}
  Lemma \ref{le2} shows that under the same budget $B_2$, if we reach more users in stage 2, the maximal influence spread will be larger, since we have more seeding options.
\begin{lem}\label{le2}
With the same budget $B_2$, if $T_1\subseteq T_2$, then $\max Q(C_2; T_1)\leq \max Q(C_2; T_2)$.
\end{lem}
  The proofs of Lemma \ref{le1} and \ref{le2} are provided in Sections 2.2 and 2.3 of the supplemental file respectively. Recall that $C_1=(c_1, c_2, \cdots, c_m)$ and $C_2=(c_1, c_2, \cdots, c_k)$. Let us write $C_1^\prime=(c_1^\prime, \cdots, c_m^\prime)$ and $C_2^\prime=(c_1^\prime, \cdots, c_k^\prime)$. For $C_1$ and $C_1^\prime$ (resp. $C_2$ and $C_2^\prime$), if $\forall i$, $c_i\geq c_i^\prime$, we denote $C_1\geq C_1^\prime$ (resp. $C_2\geq C_2^\prime$). On this basis, we have the following theorem.
\begin{thm}\label{th1}
Monotonicity property holds in both stages, i.e.,
\begin{enumerate}
\item[(1)] If $C_2\geq C_2^\prime$, then $Q(C_2; N(S))\geq Q(C_2^\prime; N(S));$
\item[(2)] If $C_1\geq C_1^\prime$, then $f(C_1; X)\geq f(C_1^\prime; X).$
\end{enumerate}
\end{thm}
\begin{proof}
Please refer to Section 2.4 in the supplemental file.
\end{proof}

  Theorem \ref{th1} indicates that the more discount we allocate to users in each stage, the larger the influence spread is. Thus, we can draw the conclusion that the budget allocated to both stages will be used up. By contradiction, if there is remaining budget in stage 1 (resp. stage 2) while $f$ (resp. $Q$) is maximized, we can add it to current discount allocation $C_1$ (resp. $C_2$). Then, $f$ (resp. $Q$) is further increased, a contradiction.
  Thus, the NIM is equivalent to the following problem.
\begin{equation}\label{NIM=}
\begin{aligned}
&\text{max}\!\!  & & f(C_1;X)                   \!\!   &&\text{max}\!\!&&Q(C_2; N(S))\\
&\text{s.t.}\!\! & & \forall u\in X, c_u\!\in\! [0,1]\!\!  &&\text{s.t.}\!\!&&\forall u\in N(S), c_u\!\in\! [0,1]\\
&            & & \!\!\sum_{u\in X}c_u= B_1    \!\!      &&           && \!\!\sum_{u\in N(S)}\!\!\!\!c_u= B_2
\end{aligned}
\end{equation}
  \subsection{Coordinate Descent Allocation}
  We attempt to decide the discount allocation in each stage by the coordinate descent algorithm. Since the design of $C_1$ should collectively consider the allocation in stage 2, for readability, we first explain how to decide the allocation in stage 2. Following a similar idea, we design the allocation in stage 1.
  \subsubsection{Coordinate Descent in Stage 2}

  Given the seed set $S$ in stage 1 and the budget constraint $B_2$, the coordinate descent algorithm iteratively optimizes $Q(C_2; N(S))$ from an initial allocation in $N(S)$, e.g., uniform allocation. In each iteration, we randomly pick two users $u$ and $v$, whose discounts are $c_u$ and $c_v$ respectively. Then, we adjust the discounts between $u$ and $v$ to optimize $Q(C_2; N(S))$, with other users' discounts fixed.

  Let $B_2^\prime=c_u + c_v$, $B_2^\prime$ is a constant during the rearrangement. Similar to the expansion of the objective function in \cite{continuous_IM}, we could rewrite $Q(C_2; N(S))$ w.r.t. $c_u$ as follows
\begin{equation}\label{eq8}
\begin{split}
Q(C_2; N(S))= & \sum_{T \subseteq N(S)\setminus \{u,v\}}P_{r}(T;C_2,N(S)\setminus\{u,v\})\cdot
 \\           & \Big \{[1-p_{u}(c_u)][1-p_{v}(B_2^\prime-c_u)]I(T)
 \\           & +[1-p_{u}(c_u)]p_{v}(B_2^\prime-c_u)I(T\cup \{v\})\\
              & +p_{u}(c_u)[1-p_{v}(B_2^\prime-c_u)]I(T\cup \{u\})\\
              & +p_{u}(c_u)p_{v}(B_2^\prime-c_u)I(T\cup \{u,v\})\Big \}.
\end{split}
\end{equation}
    Thus, $Q(C_2; N(S))$ is a function w.r.t. $c_u$. Then, we can write it as $Q(c_u)$. Due to $0\leq c_u, c_v \leq 1$ and $c_u+c_v=B_2^\prime$, we have the constraint $\max(0, B_2^\prime-1)\leq c_u\leq \min(B_2^\prime, 1)$.

    In each iteration, we obtain the new discount of $u$ by solving the following optimization problem. In the mean time, $c_v$ is determined by $c_v=B_2^\prime-c_u$.
\begin{equation}\label{eq9}
\begin{aligned}
& \text{max}
& & Q(c_u) \\
& \text{s.t.}
& & \max(0, B_2^\prime-1)\leq c_i\leq \min(B_2^\prime, 1)
\end{aligned}
\end{equation}
    It is an optimization over a single variable in a closed interval. Since $p_{u}(\cdot)$ and $p_{v}(\cdot)$ are continuously differentiable, $Q(c_u)$ is differentiable as well. Therefore, the discount $c_u$ that maximizes $Q(c_u)$ must be in one of the three cases: (1) the stationary points of $Q(c_u)$ in interval $(\max(0, B_2^\prime-1),\min(B_2^\prime, 1))$; (2) $\max(0, B_2^\prime-1)$; (3) $\min(B_2^\prime, 1)$. We check the value of $Q(c_u)$ in the three cases and choose the optimal one as the new discount allocated to user $u$.

\begin{algorithm}[h]
\begin{spacing}{1}
\small
\caption{The Coordinate Descent Algorithm CD(Q,B,T)} \label{alg1}
\begin{algorithmic}[1]
\REQUIRE ~~\!\!\!\!Objective function $Q(\cdot)$, budget $B$, accessible users $T$
\ENSURE ~~\!\!\!\!Allocation $C$
\STATE Initialize $C$
\WHILE{not converge}
	 \STATE Randomly pick users $u \;\text{and}\; v\in T$
	 \STATE $B^\prime\leftarrow c_u + c_v$
     \STATE Find all stationary points $x$ of $Q(c_u)$ in $(\max(0, B^\prime-1), \min(B^\prime, 1))$
     \STATE $c_u\leftarrow \argmax_{\{x, \max(0, B^\prime-1), \min(B^\prime, 1)\}}Q(c_u)$
     \STATE $c_v\leftarrow B^\prime-c_u$
\ENDWHILE
\STATE Return $C$
\normalsize
\end{algorithmic}
\end{spacing}
\end{algorithm}

      The coordinate descent algorithm is described in Algorithm \ref{alg1}. Its convergence is guaranteed, since $Q(c_u)$ is upper bounded by the size of the network, and $Q(c_u)$ is monotonically increasing as the iteration goes on. Thus, the algorithm will finally converge to a limit allocation. However, we can not ensure the limit to be locally optimal, which depends on the property of $Q(c_u)$. Even if the limit is locally optimal, it may not be the global optimum, where different initial allocations could be applied for improvement.


  \subsubsection{Coordinate Descent in Stage 1}
  Considering the hardness of computing the global optimum of $Q(C_2;\cdot)$, which is needed in $f(C_1;X)$ by Eq. \ref{eq4}, we would like to define a proxy function $\hat{Q}(C_2;\cdot)$ which is maximized by Alg. \ref{alg1}. Accordingly, the objective function becomes $\hat{f}(C_1;X)$.
  The coordinate descent allocation in stage 1 follows the similar idea. We start from an arbitrary allocation in $X$ and optimize $\hat{f}(C_1;X)$ iteratively. In each iteration,  we randomly pick two users $i$, $j$ and adjust the discounts between them with other discounts fixed. Let $B_1^\prime =c_i + c_j$, $B_1^\prime$ is a constant during the rearrangement. Then, $\hat{f}(C_1;X)$ can be written as
\begin{equation}\label{f(C_1)}
\begin{split}
\hat{f}(C_1;X)=\!\!& \sum_{S \subseteq X\setminus \{i,j\}}P_{r}(S;C_1,X\setminus \{i,j\})\cdot
 \\           &\!\! \Big \{[1\!-\!p_{i}(c_i)][1\!-\!p_{j}(B_1^\prime\!-\!c_i)]\max \hat{Q}(C_2; N(S))
 \\           &\!\! +[1\!-\!p_{i}(c_i)]p_{j}(B_1^\prime\!-\!c_i)\max \hat{Q}(C_2; N(S\cup \{j\}))\\
              &\!\! +p_{i}(c_i)[1\!-\!p_{j}(B_1^\prime\!-\!c_i)]\max \hat{Q}(C_2; N(S\cup \{i\}))\\
              &\!\! +p_{i}(c_i)p_{j}(B_1^\prime\!-\!c_i)\max \hat{Q}(C_2; N(S\cup \{i,j\}))\Big \}.
\end{split}
\end{equation}


      To handle $\hat{f}(C_1;X)$, we need to optimize $\hat{Q}(C_2; \cdot)$ and estimate its influence spread. Note that the optimal allocation of $\hat{Q}(C_2; \cdot)$ could be obtained by Alg. \ref{alg1}. To estimate the resultant influence spread, we apply a polling based method in Section 3 of the supplementary file.
  Analogous to the analysis in stage 2, $\hat{f}(C_1;X)$ can be written as $\hat{f}(c_i)$. In each iteration, we obtain the new discount of $i$ by solving the following optimization problem in the same way as stage 2. Meanwhile, $c_j$ is determined according to $c_j=B_1^\prime-c_i$.
\begin{equation}\label{eq11}
\begin{aligned}
& \text{max}
& & \hat{f}(c_i) \\
& \text{s.t.}
& & \max(0, B_1^\prime-1)\leq c_i\leq \min(B_1^\prime, 1)
\end{aligned}
\end{equation}

  A framework of deciding discounts with coordinate descent algorithm is established in Algorithm \ref{alg2}. We first design the allocation in stage 1 with stage 2 collectively considered. Then, based on the seeding result in stage 1, we determine the allocation in stage 2.


\begin{algorithm}[h]
\begin{spacing}{1}
\small
\caption{The Two-Stage Coordinate Descent Framework} \label{alg2}
\begin{algorithmic}[1]
\REQUIRE ~~\!\!\!\!Budget $B_1$, $B_2$, initially reachable users $X$
\ENSURE ~~\!\!\!\!Allocation $C_1$, $C_2$
     \STATE $C_1\leftarrow CD(\hat{f}(C_1;X), B_1, X)$
     \STATE Seed users in X with $C_1$
     \STATE $S\leftarrow$ users in $X$ who accept discounts
     \STATE $C_2\leftarrow CD(\hat{Q}(C_2;N(S)), B_2, N(S))$
\STATE Return $C_1$, $C_2$
\end{algorithmic}
\end{spacing}
\normalsize
\end{algorithm}


\section{Adaptive Influence Maximization}\label{AIM}
  In the adaptive case, the seeding and diffusion processes iteratively go on in a circle. In each round, one user is seeded in stage 1. Discounts are allocated to the newly reached neighbors in stage 2, and then the influence spread expands. The adaptive case is studied in two discount settings: discrete-continuous setting and discrete-discrete setting.

  \subsection{Discrete-Continuous Setting}
   \textbf{Discrete-Continuous Setting:} Users in stage 1 are probed with actions from $Y=X\times D$, where $D=\{d_{1}, d_{2}, \cdots, d_{l}\}$ is the set of optional discount rates (\textit{discrete}). In stage 2, discounts of newly reachable users take value in interval $[0,1]$ (\textit{continuous}).
   \subsubsection{Seeding Strategy}

    In this subsection, we first specify the selection of actions in stage 1. Given the previous seeding process $\psi_p$ in stage 1, recall that $\sigma(\psi_p,(\lambda_p, \phi_p))$ denotes the set of influenced users under $\psi_p$. Without causing ambiguity, we will write $\sigma(\psi_p)=\sigma(\psi_p,(\lambda_p, \phi_p))$. Then, the induced graph of uninfluenced users $V\setminus \sigma(\psi_p)$ can be denoted as $G(V\setminus \sigma(\psi_p))$. Let $\Delta(y|\psi_p)$ denote the expected number of users influenced in $G(V\setminus \sigma(\psi_p))$, if $v(y)$ becomes the seed and discounts are allocated to $v(y)$'s newly reached neighbors. $\Delta(y|\psi)$ is the marginal benefit brought by $y$, expressed as
\begin{equation}\label{eq12}
\begin{aligned}
\Delta(y|\psi_p):=& E[\hat{\sigma}(\psi_p\cup \{y\},(\Lambda, \Phi))-\\
                & \hat{\sigma}(\psi_p,(\Lambda, \Phi))|(\Lambda, \Phi)\sim \psi_p],
\end{aligned}
\end{equation}
where $(\Lambda, \Phi)\sim \psi_p$ denotes random realizations that contain the existing realization observed by $\psi_p$, and the expectation is taken with respect to $p((\lambda, \phi)):=P((\Lambda,\Phi)=(\lambda,\phi))$.
In each round, we select the action that maximizes the benefit-to-cost ratio with the remaining budget, i.e.
\begin{equation}\label{eq13}
y^*=\argmax_{y\in Y}\frac{\Delta(y|\psi_p)}{d(y)}.
\end{equation}
Supposing the targeted user $v(y^*)$ accepts the discount, we reach its neighbors and come to stage 2, where we further allocate discounts to newly reached neighbors $R$ according to the coordinate descent algorithm $CD(\cdot)$. However, if the user refuses the discount, we remove $y^*$ from $Y$ and proceed to the next round. Based on the above description, the pseudo-code is presented in Algorithm \ref{alg3} $\pi^{\text{greedy}}$. The $DCA(\cdot)$ function in Algorithm \ref{alg3} is presented in Algorithm \ref{alg4}, which describes the actions triggered after action $y^*$ is accepted. The framework of $\pi^{\text{greedy}}$ will be inherited in the subsequent adaptive greedy algorithms, except some modifications.

\begin{algorithm}[h]
\begin{spacing}{0.9}
\small
\caption{The Adaptive Greedy Algorithm $\pi^{\text{greedy}}$} \label{alg3}
\begin{algorithmic}[1]
\REQUIRE ~~\!\!\!\!Budget $B_1$, $B_2$, action space $Y=X\times D$
\ENSURE ~~\!\!\!\!Accepted actions $P_1$, allocation $C_2$
\STATE Initialize $P_1\leftarrow \emptyset$, $C_2\leftarrow \mathbf{0}$
\WHILE{$B_1\geq 0$}
\IF{$\exists y\in Y$ s.t. $d(y)\leq B_1$}
    \STATE Select $y^*=\argmax_{y\in Y}\frac{\Delta(y|\psi_p)}{d(y)}$ s.t. $d(y)\leq B_1$,
    \STATE Probe $v(y^*)$ with discount $d(y^*)$
\IF{$v(y^*)$ accepts $d(y^*)$}
    \STATE $R\leftarrow$ newly reachable users
    \STATE $C_2^p\leftarrow DCA(R, y^*, P_1, B_1, Y)$
    \STATE Update $C_2$ with $C_2^p$
\ELSE
    \STATE $Y\leftarrow Y\setminus y^*$
\ENDIF
\ENDIF
\ENDWHILE
\STATE Return $P_1$, $C_2$
\end{algorithmic}
\end{spacing}
\normalsize
\end{algorithm}


\begin{algorithm}[h]
\begin{spacing}{0.9}
\small
\caption{D-C Allocation $DCA(R, y^*, P_1, B_1, Y)$} \label{alg4}
\begin{algorithmic}[1]
\REQUIRE ~~\!\!\!\!Newly reachable users $R$, action newly accepted $y^*$, accepted actions $P_1$, remaining budget $B_1$, current action space $Y$
\ENSURE ~~\!\!\!\!Allocation $C_2^p$
    \STATE Initialize $C_2^p$
    \STATE $P_1\leftarrow P_1 \cup y^*$; $B_1\leftarrow B_1-d(y^*)$
    \STATE $Y\leftarrow Y\setminus \{y|v(y)=v(y^*)\}$
    \STATE $C_2^p\leftarrow CD(\hat{Q}(C_2^p;R)$,$\frac{B_2}{B-B_2}\times d(y^*),R)$
    \STATE Return $C_2^p$
\end{algorithmic}
\end{spacing}
\normalsize
\end{algorithm}

\begin{defi}\label{def5}
\textbf{(Adaptive Submodularity).} The function $\hat{\sigma}$ is adaptive submodular with respect to $p((\lambda, \phi))$, if for all $\psi \subseteq \psi^\prime$(i.e.$\psi$ is a subprocess of $\psi^\prime$), and $\forall y \in Y\setminus dom(\psi)$, we have
\begin{center}
$\Delta(y|\psi)\geq \Delta(y|\psi^\prime)$.
\end{center}
Furthermore, if for all $\psi$ and all $y\in Y$, $\Delta(y|\psi)\geq 0$ holds, then $\hat{\sigma}$ is adaptive monotone.
\end{defi}

  By Lemma \ref{adasub1}, we find that $\hat{\sigma}(\cdot|(\lambda,\phi))$ is adaptive submodular and adaptive monotone, where the proof is deferred to Section 2.5 of the supplemental file. Since a non-negative linear combination of monotone adaptive submodular functions is still monotone adaptive submodular, we have $\hat{\sigma}(\cdot)$ is monotone and adaptive submodular.

\begin{lem}\label{adasub1}
$\hat{\sigma}(\cdot|(\lambda,\phi))$ is adaptive monotone and adaptive submodular, under any realization $(\lambda,\phi)$.
\end{lem}
  \subsubsection{Relaxation Analysis}
  Following a similar idea of analyzing the adaptive algorithm in \cite{mobihoc2017}, we relax the seeding process in stage 1 by assuming that the minimum discount rate $d_{\min}(u)\in D$ desired by each user $u$ is pre-known. In the seeding process, if $u$ is probed with a discount no smaller than $d_{\min}(u)$, then $u$ will accept it and become a seed, and vice versa. To maximize the benefit-to-cost ratio, each user will be probed from small discounts. By the definition of $d_{\min}(u)$, discounts smaller than $d_{\min}(u)$ will be rejected. When the discount becomes $d_{\min}(u)$, the user will accept it and become the seed, and polices with higher discount are abandoned. Thus, it becomes meaningless to probe user $u$ with discounts higher or lower than $d_{\min}(u)$. And, the action space is reduced to $Y^{\text{relaxed}}=\{(u, d_{\min}(u)), u\in X\}$. We denote the adaptive greedy algorithm under the relaxed setting as $\pi_{\text{relaxed}}^{\text{greedy}}$. In each round, we select an action $y^*$ from $Y^{\text{relaxed}}$ that maximizes the benefit-to-cost ratio $\frac{\Delta(y|\psi_p)}{d_{\min}(v(y))}$. Once an action in $Y^{\text{relaxed}}$ is adopted, the user will become a seed definitely, since $d_{\min}(u)$ is the desired discount of user $u$. The subsequent seeding strategy in newly reachable users remains unchanged.

With the relaxation analysis, we find that the seed sets of the relaxed setting and the original setting are the same in stage 1 by Lemma \ref{relax}, which is similar to \cite{mobihoc2017}. This lemma is proven by mathematical induction in Section 2.6 of the supplemental file.
\begin{lem}\label{relax}
Under any realization $(\lambda, \phi)$, $\pi_{\text{relaxed}}^{\text{greedy}}$ yields the same seed set in stage 1 as $\pi^{\text{greedy}}$.
\end{lem}
  Let $\pi^{\text{OPT}}$ denote the optimal policy under the discrete-continuous setting. With the adaptive submodularity of $\hat{\sigma}(\cdot)$ and Lemma \ref{relax}, we obtain the performance guarantee of the greedy algorithm $\pi^{\text{greedy}}$ in Theorem \ref{th2}.
\begin{thm}\label{th2}
If global optimality is obtained in stage 2 in each round, then the adaptive greedy policy $\pi^{\text{greedy}}$ obtains at least $(1-e^{-\frac{B_1-1}{B_1}})$ of the value of the optimal policy $\pi^{\text{OPT}}$,\\
$\hat{\sigma}(\pi^{\text{greedy}})\geq (1-e^{-\frac{B_1-1}{B_1}})\hat{\sigma}(\pi^{\text{OPT}})$.
\end{thm}
\begin{proof}
Please refer to Section 2.7 in the supplemental file.

\end{proof}

  \subsection{Discrete-Discrete Setting}
    In the previous setting, the coordinate descent algorithm is applied to decide the continuous allocation in stage 2, which needs numerous iterative optimizations. In this subsection, we introduce a discrete solution in stage 2.

    \textbf{Discrete-Discrete Setting:} In each stage, we select actions from an action space defined by the Cartesian product of users and discount rates $D$. Thus, the discounts in both stages are discrete.
\begin{defi}\label{def6}
\textbf{(Submodularity)} For a real-valued function $h(\cdot)$ defined on subsets of a finite ground set $G$. If for all $A\subseteq B\subseteq G$, and for all $x\in G\setminus B$, we have

\begin{center}
$h(A \cup\{x\})-h(A)\geq h(B \cup\{x\})-h(B)$.
\end{center}
Then, we say $h(\cdot)$ is submodular. Furthermore, if $h(A)\leq h(B)$ holds for all $A\subseteq B\subseteq G$, $h(\cdot)$ is said to be monotone.
\end{defi}

    The seeding process in stage 1 is still sequential. In each round, we select an action $y^*$ from $Y:=X\times D$. If $v(y^*)$ refuses the discount $d(y^*)$, we delete $y^*$ from $Y$ and move to the next round. If $v(y^*)$ accepts it and becomes a seed, some users become newly reachable, denoted as $R$. An action space $Z=R\times D$ is defined in stage 2. We will select a subset of actions $L\subseteq Z$ to seed users in $R$ under some budgets drawn from $B_2$. The budget only depends on the intrinsic property of $v(y^*)$, just like the discrete-continuous setting. Different from selecting actions sequentially in stage 1, we decide actions in stage 2 all at once without observing the seeding and diffusion results of each action. Because in realistic scenario, we are not likely to have so much time to observe the diffusion of each action in stage 2.


    We next show how to select the set of actions in stage 2. Let $d(u|L)$ denote the discount allocated to user $u$ under actions $L$. Analogous to the definition in Equations \eqref{eq1} and \eqref{eq3}, we define the probability that a subset of users $T\subseteq R$ accept discounts as
\begin{equation}\label{eq14}
P_{r}(T;L,R)=\prod_{u\in R}p_{u}(d(u|L))\prod_{v\in R \setminus S}(1-p_{v}(d(v|L))).
\end{equation}
    Given the partial seeding process $\psi_p$ in stage 1 and the set of influenced users $\sigma(\psi_p)$, we denote the expected number of users newly influenced by $T$ as $I_{G(V\setminus \sigma(\psi_p))}(T)$. Then, the number of users influenced by actions $L$ is
\begin{equation}\label{eq15}
Q(L; R)=\sum_{T \subseteq R}P_{r}(T;L,R)I_{G(V\setminus \sigma(\psi_p))}(T).
\end{equation}

   We attempt to find a set of actions that maximize the influence spread. However, we find that maximizing $Q(L; R)$ is NP-hard by Lemma \ref{NP2}.
\begin{lem}\label{NP2}
Finding the optimal set of actions in stage 2 is NP-hard.
\end{lem}

   However, we can prove the monotonicity and submodularity of $Q(L; R)$ by Lemma \ref{le8}. Since $Q(L; R)$ is monotone and submodular, we are motivated to design the approximate algorithm $GS(R, y^*, P_1, B_1, Y)$ in Algorithm \ref{alg5} to determine the actions in stage 2.

\begin{lem}\label{le8}
$Q(L; R)$ is monotone and submodular w.r.t. $L$.
\end{lem}

   The selection of action $y^*$ in stage 1 is similar to the discrete-continuous setting. We select the one that maximizes the benefit-to-cost ratio $\frac{\Delta(y|\psi_p)}{d(y)}$. Note that, when calculating $\Delta(y|\psi_p)$, discounts of users in stage 2 are determined by Algorithm \ref{alg5} rather than the coordinate descent algorithm.

   We continue to examine the property of $\hat{\sigma}(\cdot |(\lambda, \phi))$ in the discrete-discrete setting and find that it is still adaptive submodular by Lemma \ref{adasub2}.
\begin{lem}\label{adasub2}
In the Discrete-Discrete Setting, $\hat{\sigma}(\cdot |(\lambda, \phi))$ is still adaptive submodular, under any realization ($\lambda$, $\phi$).
\end{lem}

 For readability, the proofs of Lemma \ref{NP2}, \ref{le8} and \ref{adasub2} in this subsection are deferred to Sections 2.8, 2.9 and 2.10 of the supplemental file respectively. Based on the above description, we obtain the adaptive greedy algorithm $\pi^{\text{greedy}}_{\text{discrete}}$ in the discrete-discrete setting. Since the only difference from $\pi^{\text{greedy}}$ lies in the allocation in stage 2, we can derive $\pi^{\text{greedy}}_{\text{discrete}}$ by replacing $DCA(R, y^*, P_1, B_1, Y)$ with $GS(R, y^*, P_1, B_1, Y)$. For the sake of the space, we omit the detailed description here. Let $\pi^{\text{OPT}}_{\text{discrete}}$ denote the optimal policy. With the adaptive submodularity in stage 1 and submodularity in stage 2, we obtain the performance guarantee of $\pi^{\text{greedy}}_{\text{discrete}}$ in Theorem \ref{th3}.
\begin{algorithm}[h]
\begin{spacing}{1}
\small
\caption{The Greedy Selection $GS(R, y^*, P_1, B_1, Y)$}\label{alg5}
\begin{algorithmic}[1]
\REQUIRE ~~\!\!\!\!Newly reachable users $R$, action newly accepted $y^*$, actions accepted $P_1$, remaining budget $B_1$, current action space $Y$
\ENSURE ~~\!\!\!\!Actions $P_2$
    \STATE Initialize $P_2\leftarrow \emptyset$
    \STATE $P_1\leftarrow P_1 \cup \{y^*\}$; $B_1\leftarrow B_1-d(y^*)$
    \STATE $Y\leftarrow Y\setminus \{y|v(y)=v(y^*)\}$
    \STATE $S_1\leftarrow \emptyset$, $Z\leftarrow R\times D$, $S_2\leftarrow \argmax_{z\in Z}\{Q(\{z\};R)|z\in Z, d(z)\leq \frac{B_2}{B-B_2}\times d(y^*)\}$
\WHILE{$\sum_{z\in S_1}d(z)\leq \frac{B_2}{B-B_2}\times d(y^*)$}
    \STATE $z^*\leftarrow \argmax_{z\in Z}\frac{Q(S_1\cup {z})-Q(S_1)}{d(z)}$ 
    \IF{$d(z^*)+\sum_{z\in S_1}d(z)\leq \frac{B_2}{B-B_2}d(y^*)$}
    \STATE $S_1\leftarrow S_1\cup \{z^*\}$, $Z\leftarrow Z\setminus \{z^*\}$
    \ENDIF
\ENDWHILE
    \STATE $P_2\leftarrow \argmax_{S\in \{S_1,S_2\}}Q(S;R)$
    \STATE Return $P_2$
\end{algorithmic}
\end{spacing}
\normalsize
\end{algorithm}
\begin{thm}\label{th3}
The greedy policy $\pi^{\text{greedy}}_{\text{discrete}}$ obtains at least $(1-e^{-\frac{B_1-1}{2B_1}(1-\frac{1}{e})})$ of the value of $\pi^{\text{OPT}}_{\text{discrete}}$,
\begin{equation*}
\hat{\sigma}(\pi^{\text{greedy}}_{\text{discrete}})\geq (1-e^{-\frac{B_1-1}{2B_1}(1-\frac{1}{e})})\hat{\sigma}(\pi^{\text{OPT}}_{\text{discrete}}).
\end{equation*}
\end{thm}
\begin{proof}

Please refer to Section 2.11 in the supplemental file.

\end{proof}
The performance guarantee of $\pi^{\text{greedy}}_{\text{discrete}}$ is not appealing enough. The enumeration method can be further applied to improve the approximation ratio. In fact, the size of newly reachable users is relatively small compared with the whole network, hence, the enumeration will not be so computationally costly. The modified greedy algorithm in stage 2 is described in Algorithm \ref{alg6} whose approximation ratio is $(1-e^{-1})$ \cite{a_note_on}. The complete algorithm $\pi^{\text{greedy}}_{\text{enum}}$ can be derived by replacing $DCA(R, y^*, P_1, B_1, Y)$ in Algorithm \ref{alg3} with the $MGS(R, y^*, P_1, B_1, Y)$. Following similar argument in the proof of Theorem \ref{th3}, we obtain the approximation ratio of $\pi^{\text{greedy}}_{\text{enum}}$ in Theorem \ref{th4}.
\begin{thm}\label{th4}
The policy $\pi^{\text{greedy}}_{\text{enum}}$ achieves an approximation ratio of $1-e^{-\frac{B_1-1}{B_1}(1-\frac{1}{e})}$.
\end{thm}

\begin{algorithm}[h]
\begin{spacing}{1}
\small
\caption{Modified Greedy $MGS(R, y^*, P_1, B_1,Y)$} \label{alg6}
\begin{algorithmic}[1]
\REQUIRE ~~\!\!\!\!Newly reachable users $R$, action newly accepted $y^*$, actions accepted $P_1$, remaining budget $B_1$, current action space $Y$
\ENSURE ~~\!\!\!\!Actions $P_2$
    \STATE $P_1\leftarrow P_1 \cup \{y^*\}$; $B_1\leftarrow B_1-d(y^*)$
    \STATE $Y\leftarrow Y\setminus \{y|v(y)=v(y^*)\}$
    \STATE $S_1\leftarrow \emptyset$, $Z\leftarrow R\times D$, $S_2\leftarrow \argmax\{Q(A;R)|A\subseteq Z, |A|<3,\sum_{z\in A}d(z)\leq \frac{B_2}{B-B_2} d(y^*)\}$
    \FOR{$A\subseteq Z$, $|A|=3$, $\sum_{z\in A}d(z)\leq \frac{B_2}{B-B_2}d(y^*)$}
        \STATE $Z^\prime\leftarrow Z\setminus A$
        \WHILE{$Z^\prime \neq \emptyset$}
            \STATE $z^*\leftarrow \argmax_{z\in Z^\prime}\frac{Q(A\cup {z};R)-Q(A;R)}{d(z)}$ 
            \IF{$d(z^*)+\sum_{z\in A}d(z)\leq \frac{B_2}{B-B_2}d(y^*)$}
                \STATE $A\leftarrow A\cup \{z^*\}$
            \ENDIF
            \STATE $Z^\prime\leftarrow Z^\prime \setminus z^*$
        \ENDWHILE
        \IF{$Q(A)>Q(S_1)$}
          \STATE $S_1\leftarrow A$
        \ENDIF
    \ENDFOR
    \STATE $P_2\leftarrow \argmax_{S\in \{S_1,S_2\}}Q(S;R)$
\STATE Return $P_2$
\end{algorithmic}
\end{spacing}
\normalsize
\end{algorithm}
\section{Experiments}\label{experiments}

In this section, we examine the performance of the proposed algorithms on four real-world datasets. The purpose lies in three parts. First, we compare the expected influence spread of our two-stage algorithms with other algorithms to show the advantage of our proposed algorithms. Second, we examine the scalability of our algorithms with respect to the total budget. Third, we test the sensitivity to different settings of seed probability function. Fourth, the impact of the friendship paradox is evaluated. All algorithms were implemented in C\# and simulated on a Linux x64 server (Intel Xeon E5-2650 v2 @2.6Ghz, 128GB RAM).

  \subsection{Experimental setup}\label{Experimental_setup}
    \textbf{Dataset Description.} We test our algorithms on four networks derived from SNAP \cite{snap}. The parameters of the four datasets are presented in Table \ref{dataset}. Undirected networks are converted to directed networks, which means that every undirected edge $(u,v)$ is replaced by two directed edges $(u,v)$ and $(v,u)$, and the number of edges is doubled. In each network, we randomly selected 100 nodes as the initially accessible users $X$.
\begin{table}[h]
\caption{Datasets}\label{dataset}
\vspace{-3mm}
\begin{center}
\renewcommand\arraystretch{1.2}
\footnotesize
\begin{tabular}{|p{2.2cm}<{\centering}|p{1.2cm}<{\centering}|p{1.3cm}<{\centering}|p{1.25cm}<{\centering}|p{0.7cm}<{\centering}|}
\hline
Datasets & Nodes & Edges & Type &$\theta$ \\
\hline
wiki-Vote & 7,115 & 103,689 & Directed& 0.25M\\
\hline
ca-CondMat & 23,133 & 186,936 & Directed& 2M\\
\hline
com-DBLP & 317,080 & 1,049,866 & Undirected& 20M\\
\hline
soc-LiveJournal1 & 4,847,571 & 68,993,773& Directed & 40M\\
\hline
\end{tabular}
\end{center}
\vspace{-3mm}
\end{table}
\normalsize

    \textbf{Propagation Probability.} The diffusion model adopted is the independent cascade model, which is widely employed in the literature of influence maximization \cite{a_martingale_approach, continuous_IM, time_complexity_practical_efficiency,scalable_IM_in_large_scale}. Each edge $(u,v)$ is associated with a propagation probability $p_{uv}$, set to be $\frac{\alpha}{\text{in-degree of}\, v}$, where $\alpha\in \{0.6, 0.8, 1\}$. This setting is quite common in existing works \cite{seminal_IM, mobihoc2017, continuous_IM, scalable_IM_in_large_scale}.

    \textbf{Seed Probability Function.} Recall that whether a user $u$ accepts the discount $c_u$ is captured by the seed probability function $p_u(c_u)$, which means the probability that $u$ accepts the discount $c_u$ and becomes a seed. User's behavior is affected by various factors, such as time and demand. The best way to estimate the probability function may be learning from data, which is out of the scope of our research. Thus, we apply synthesized seed probability functions, which satisfy the four properties mentioned in Section \ref{model}. For each network, we randomly select 5\% nodes and set $p_u(c_u)$ with $p_u(c_u)=c_u^2$, 10\% nodes with $p_u(c_u)=c_u$ and 85\% with $p_u(c_u)=2c_u-c_u^2$.

    \textbf{Discount Rate.} The action space is defined as the Cartesian product of users and discount rates $D$. In our experiment, the discount rate $D$ is set to be an arithmetic progression from 10\% to 100\% with common difference 10\%. That is, 10 candidate discounts are considered. 

    \textbf{Implementation.} Influence estimation is frequently demanded in the algorithms to determine discount allocation. To obtain an unbiased estimation, we adopt the polling based technique proposed in \cite{a_martingale_approach} \cite{time_complexity_practical_efficiency}, where $\theta$ (given in Table \ref{dataset}) reverse reachable sets are generated. For detailed description, please refer to Section 3 of the supplemental file. All the reported influence spreads are estimated by running 20K times Monte Carlo simulations.

  \subsection{Algorithms Evaluated}

  To validate the performance of our four algorithms, we include six more algorithms for comparison. All the algorithms are tested under budgets $B\in \{10, 20, 30, 40, 50\}$ with $B_1:B_2=1:4$. For one-stage algorithms, the budget $B$ is all spent in $X$.

  \begin{itemize}[leftmargin=*]
    \item \textbf{Non-adaptive Algorithms}
  \end{itemize}

    \textbf{Random Friend (RF):} We introduce RF as a basic two-stage algorithm. We uniformly and randomly select $B_1$ users in $X$ as agents $S$. Then, $B_2$ users are selected from $N(S)$ in a similar way.

    \textbf{Discrete Influence Maximization (IM):} First applied by Kempe \emph{et al.} \cite{seminal_IM}, IM is a classic algorithm  with performance guarantee $1-\frac{1}{e}$. Users in $X$ are greedily selected in a manner that, given the first $k$ seeds, the user $u$, which maximizes the expected marginal benefit $I(T\cup \{u\})-I(T)$, is selected as the $(k+1)$-th seed.

    \textbf{Coordinate Descent (CD):} This is a one-stage algorithm carried out in initially accessible users $X$. Detailed description could be found in Algorithm \ref{alg1}. As for the initial allocation, we first rank users in $X$ with respect to the degree in a non-increasing order. Then, the budget $B$ is uniformly allocated to the first $1.5B$ users. The number of iterations is set to be 50, enough for the refinement in 100 nodes.

    \textbf{Two-stage CD (2CD):} This is the two-stage algorithm described in Algorithm \ref{alg2}. The initial allocation in each stage is determined in the same way as CD. The number of iterations is 10 in both stages.

  \begin{itemize}[leftmargin=*]
    \item \textbf{Adaptive Algorithms}
  \end{itemize}

   \textbf{Adaptive Selection (Ada):} As the only one-stage algorithm in the adaptive case, adapted from \cite{mobihoc2017}, Ada sequentially seeds users in $X$ with budget $B$. Each time, we select the action $y$ from $Y$ which maximizes the benefit-to-cost ratio $\frac{\Delta(y|\psi_p)}{d(y)}$, where $\Delta(y|\psi_p)$ is the expected influence spread brought by $v(y)$ under the previous diffusion result.

   \textbf{LP-Based Approach (LP):} A two-stage algorithm proposed in \cite{WWW_Singer}, where the degree of a user is regarded as its influence. The maximization problem is formulated as an integer linear programming. The solution returns an allocation with $(1-\frac{1}{e})$-approximation ratio.

  \textbf{A-Greedy:} Proposed in \cite{A_Greedy}, A-Greedy is an adaptive one-stage algorithm. Thus, we transform it into a two-stage algorithm. In both stages, A-Greedy is applied to select the seeds. The optimal seeding pattern $A^*$ is adopted, where one user is selected each time and the next selection takes place until no user is further influenced. The activation probability $f_u$ of each user $u$ in \cite{A_Greedy} is set to be 0.6.

  $\boldsymbol{\pi^{\text{greedy}}}$ \textbf{(Ada+CD):} This two-stage algorithm is described in detail in Alg. \ref{alg3}. The coordinate descent algorithm is applied in newly reachable users. The initialization and number of iterations are the same as the CD algorithm.

  $\boldsymbol{\pi^{\text{greedy}}_{\text{discrete}}}$ \textbf{(Ada+GS):} The framework is the same as the Ada+CD algorithm. But in stage 2, the coordinate descent algorithm is replaced by the Greedy Selection in Alg. \ref{alg5}. It is worth noting that the discount rate $D$ in stage 2 becomes $\{0.5, 1\}$, because we find that a fine-grained discount rate with granularity 0.1 will lead to worse results. The explanation is that greedy selection prefers giving small discounts to many users, while the number of newly reachable users is relatively small. Then, the budget is left, making the experiment unfair.

  $\boldsymbol{\pi^{\text{greedy}}_{\text{enum}}}$ \textbf{(Ada+MGS):} The actions in the second stage are determined by the Modified Greedy Selection described in Alg. \ref{alg6}. The discount rate in stage 2 remains to be $\{0.5, 1\}$ for the same reason.

  \subsection{Experimental Results}

    The ten algorithms are mainly evaluated on four metrics: (1) the expected influence spread; (2) the scalability regarding the total budget; (3) the sensitivity of algorithms with respect to different settings of seed probability functions; (4) the impact of the FP phenomenon. Due to space limitation, the influence spread and scalability under $\alpha=0.8$ are presented in Section 4 of the supplemental file.

  \begin{figure*}[h]
    \centering
    \subfigure[Wiki-Vote]
        {
            \begin{minipage}[h]{0.24\textwidth}
            \centerline{\includegraphics[width=1\textwidth]{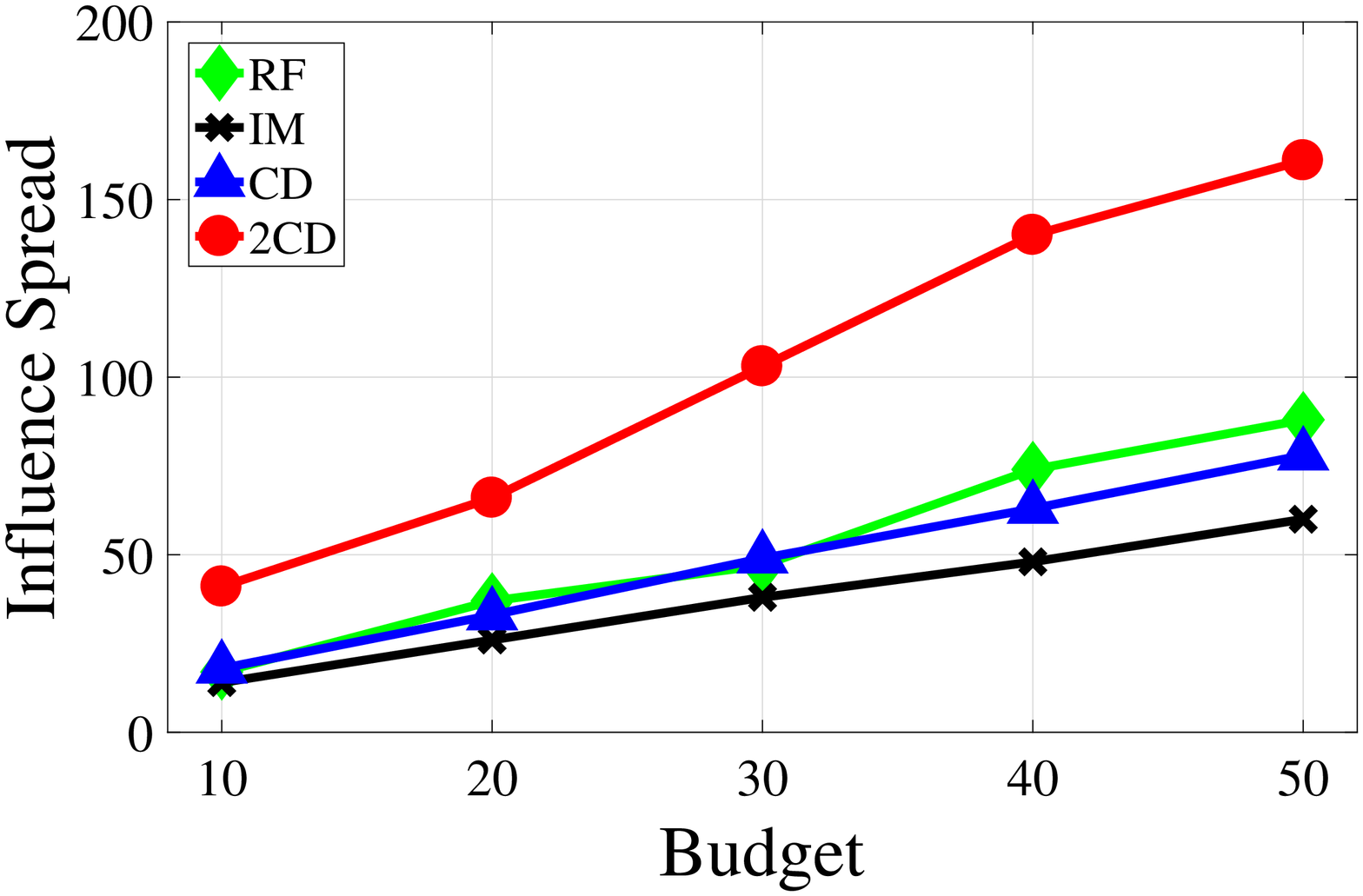}\label{wiki-vote_Non-ada_alpha=06}}
            \vspace{-2mm}
            \centerline{\quad \footnotesize{$\alpha$=0.6}}
            \vspace{1mm}
            \end{minipage}
            \hspace{-1mm}
            \begin{minipage}[h]{0.24\textwidth}
            \centerline{\includegraphics[width=1\textwidth]{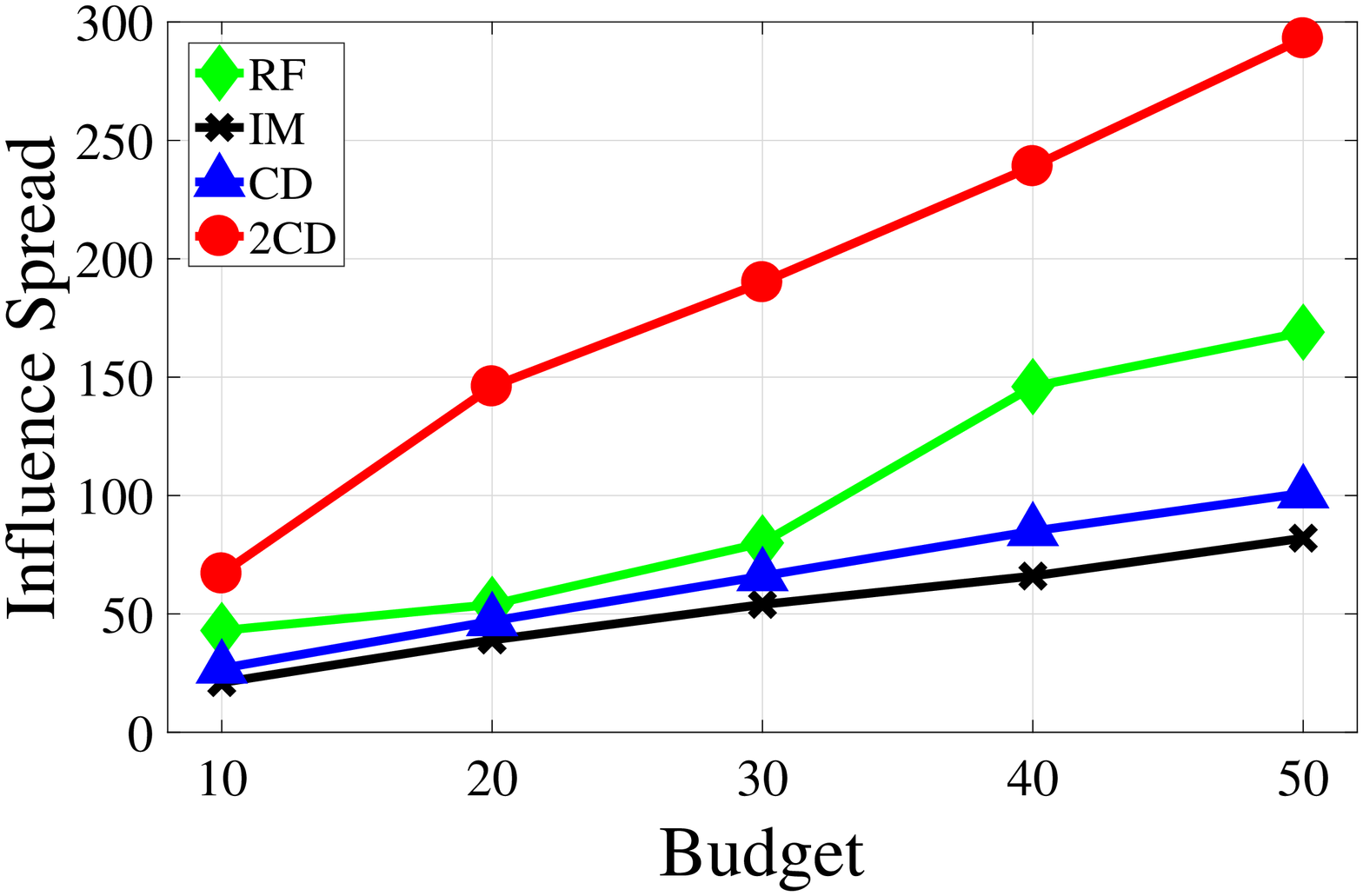}\label{wiki-vote_Non-ada_alpha=10}}
            \vspace{-2mm}
            \centerline{\quad \footnotesize{$\alpha$=1.0}}
            \vspace{1mm}
            \end{minipage}
        }
    \hspace{-3mm}\vspace{-0.5mm}
    \subfigure[Ca-CondMat]
        {
            \begin{minipage}[h]{0.24\textwidth}
            \centerline{\includegraphics[width=1\textwidth]{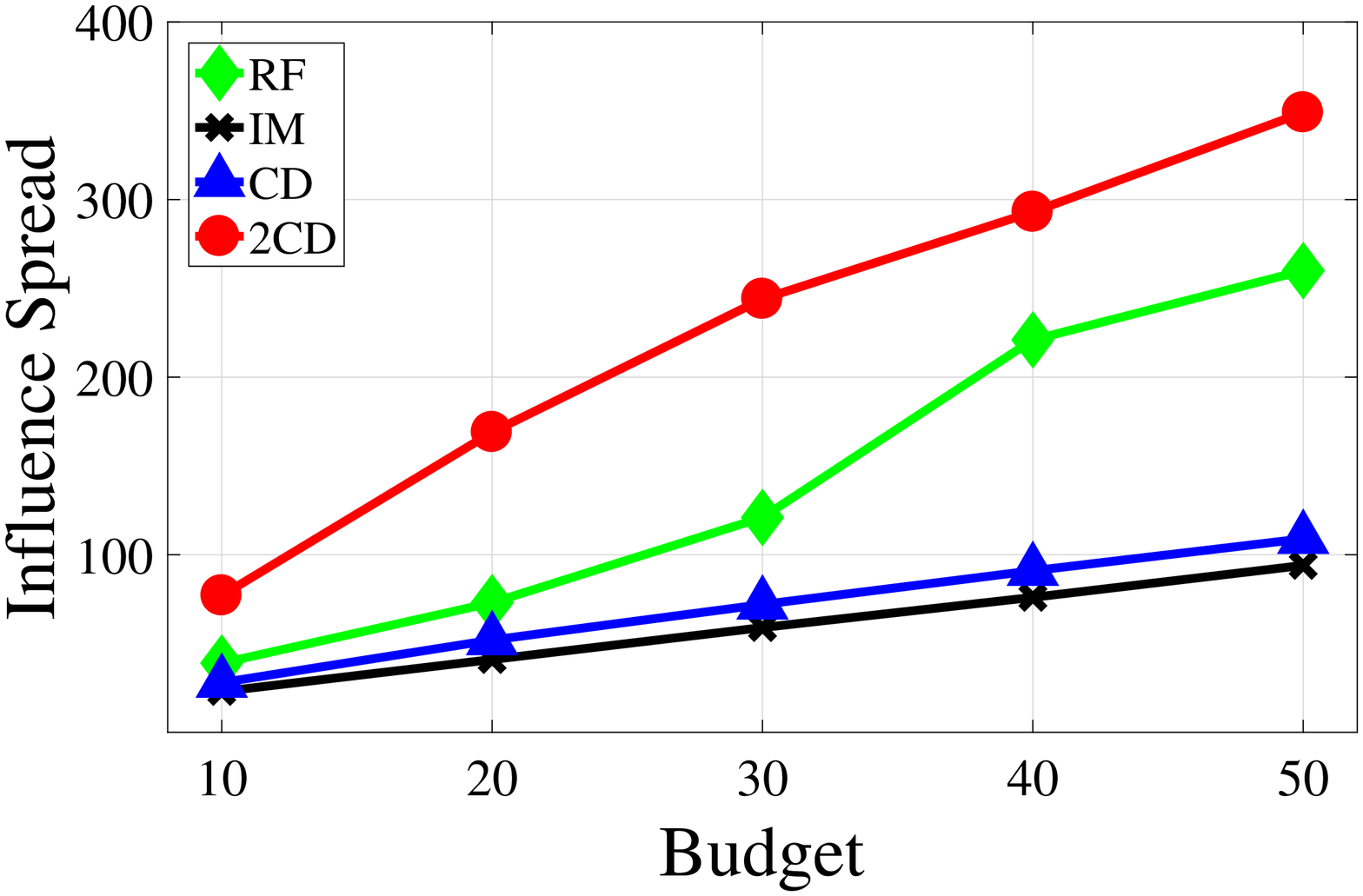}\label{Condmat_Non-ada_alpha=06}}
            \vspace{-2mm}
            \centerline{\quad \footnotesize{$\alpha$=0.6}}
            \vspace{1mm}
            \end{minipage}
            \hspace{-1mm}
            \begin{minipage}[h]{0.24\textwidth}
            \centerline{\includegraphics[width=1\textwidth]{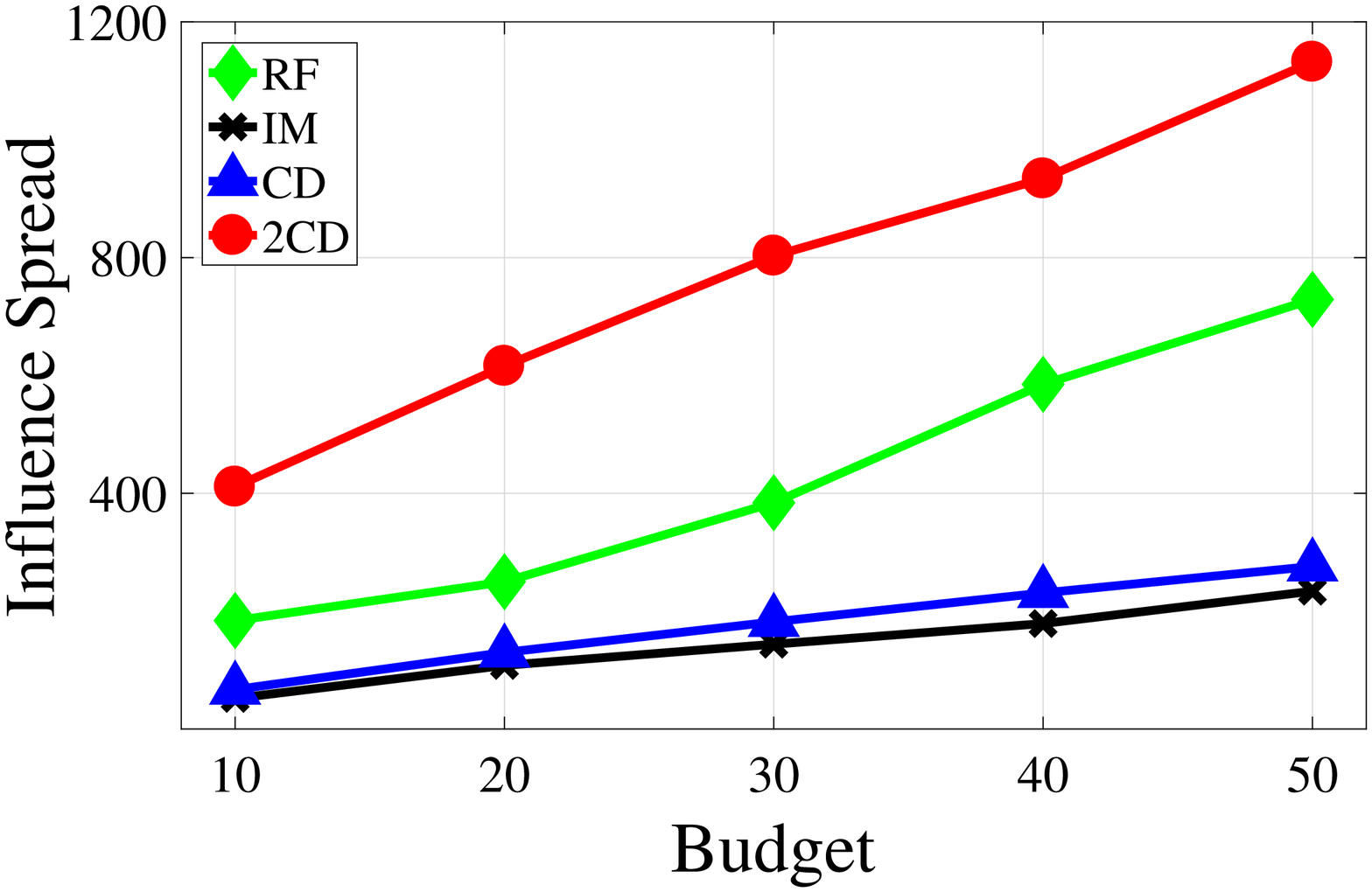}\label{Condmat_Non-ada_alpha=10}}
            \vspace{-2mm}
            \centerline{\quad \footnotesize{$\alpha$=1.0}}
            \vspace{1mm}
            \end{minipage}
        }
    \vspace{-0.5mm}
    \\
      \subfigure[com-Dblp]
        {
            \begin{minipage}[h]{0.24\textwidth}
            \centerline{\includegraphics[width=1\textwidth]{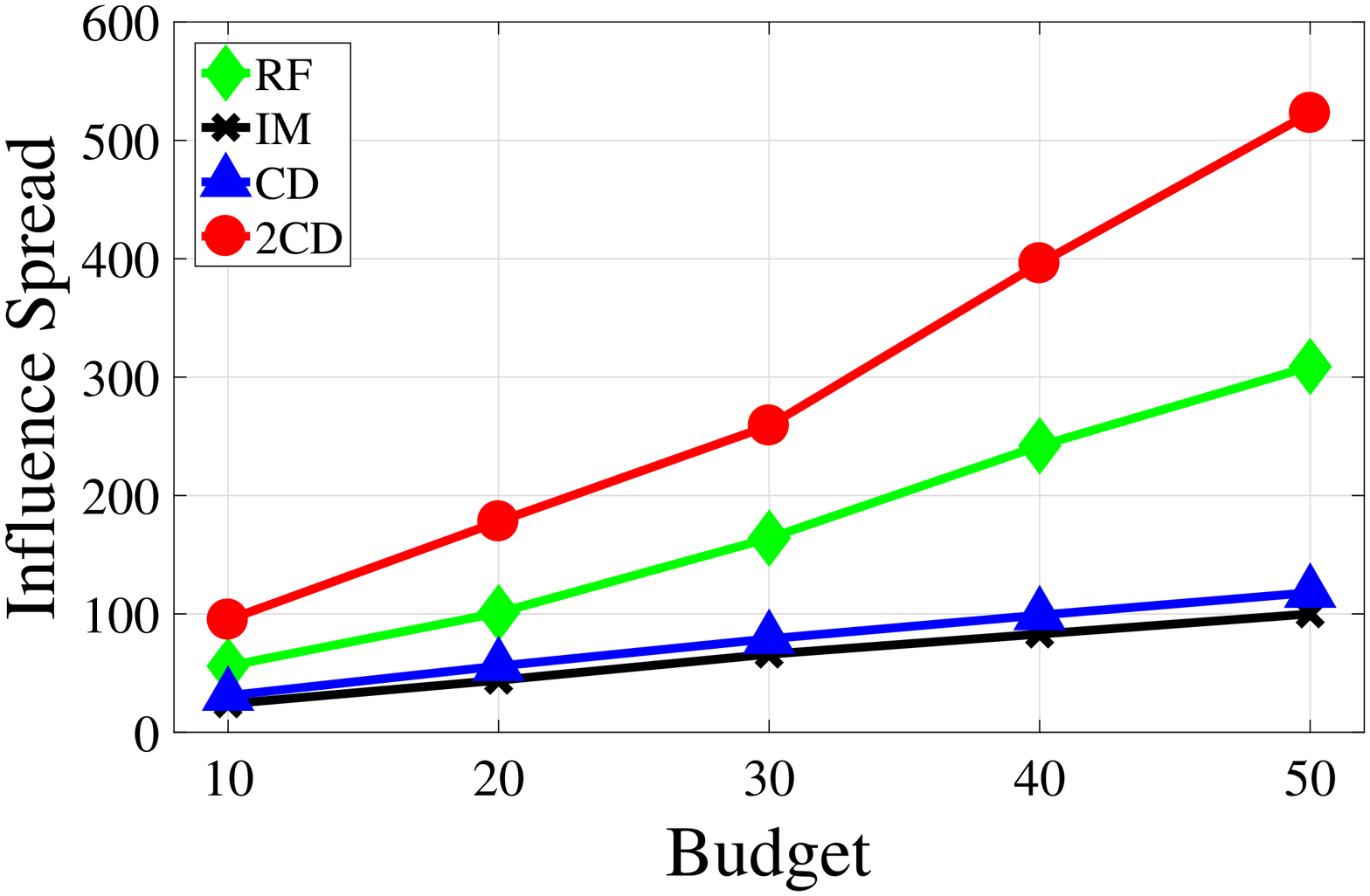}\label{Dblp_Non-ada_alpha=06}}
            \vspace{-2mm}
            \centerline{\quad \footnotesize{$\alpha$=0.6}}
            \vspace{1mm}
            \end{minipage}
            \hspace{-1mm}
            \begin{minipage}[h]{0.24\textwidth}
            \centerline{\includegraphics[width=1\textwidth]{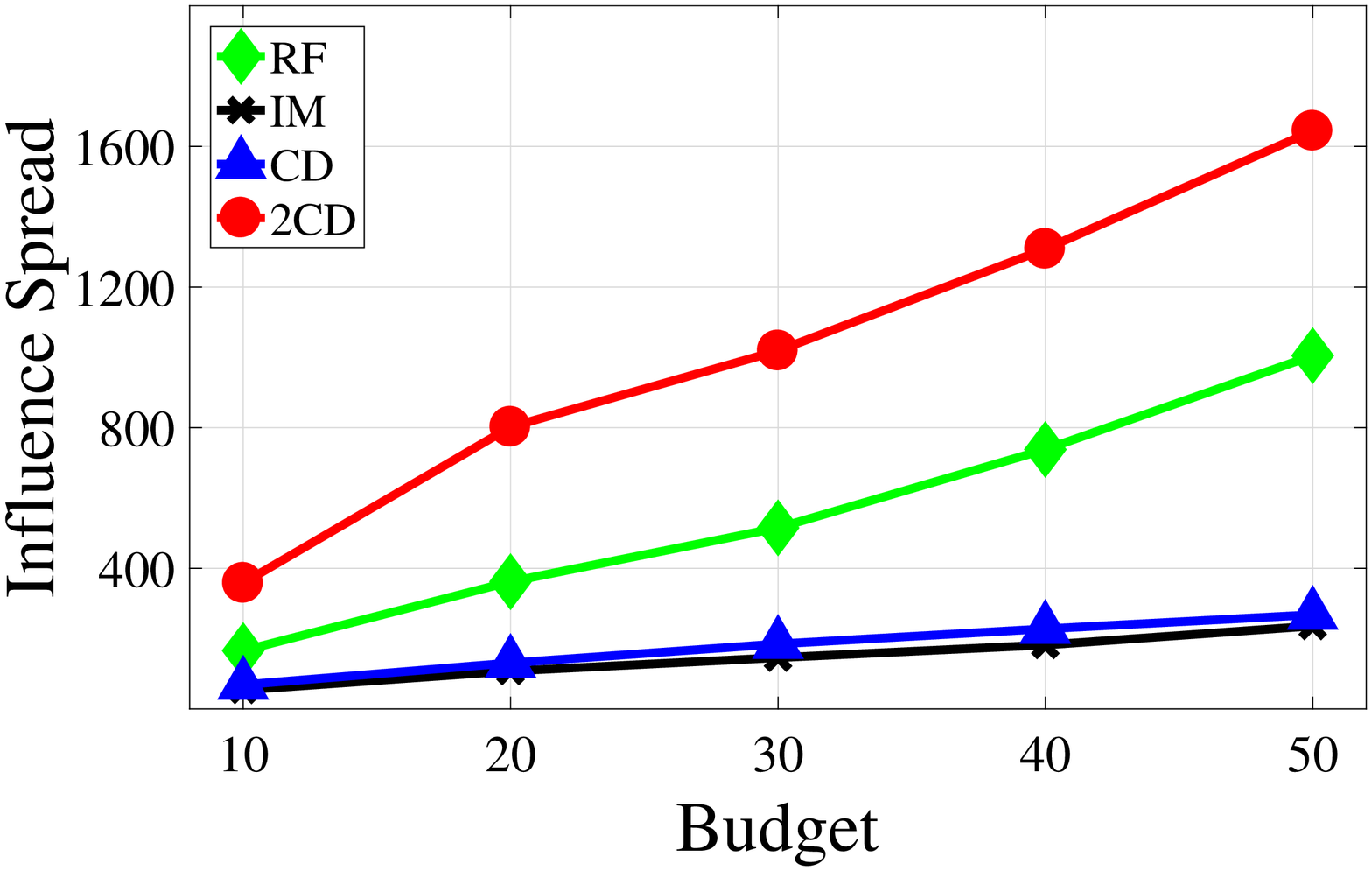}\label{Dblp_Non-ada_alpha=10}}
            \vspace{-2mm}
            \centerline{\quad \footnotesize{$\alpha$=1.0}}
            \vspace{1mm}
            \end{minipage}
        }
    \hspace{-3mm}
      \subfigure[soc-Livejournal]
        {
            \begin{minipage}[h]{0.24\textwidth}
            \centerline{\includegraphics[width=1\textwidth]{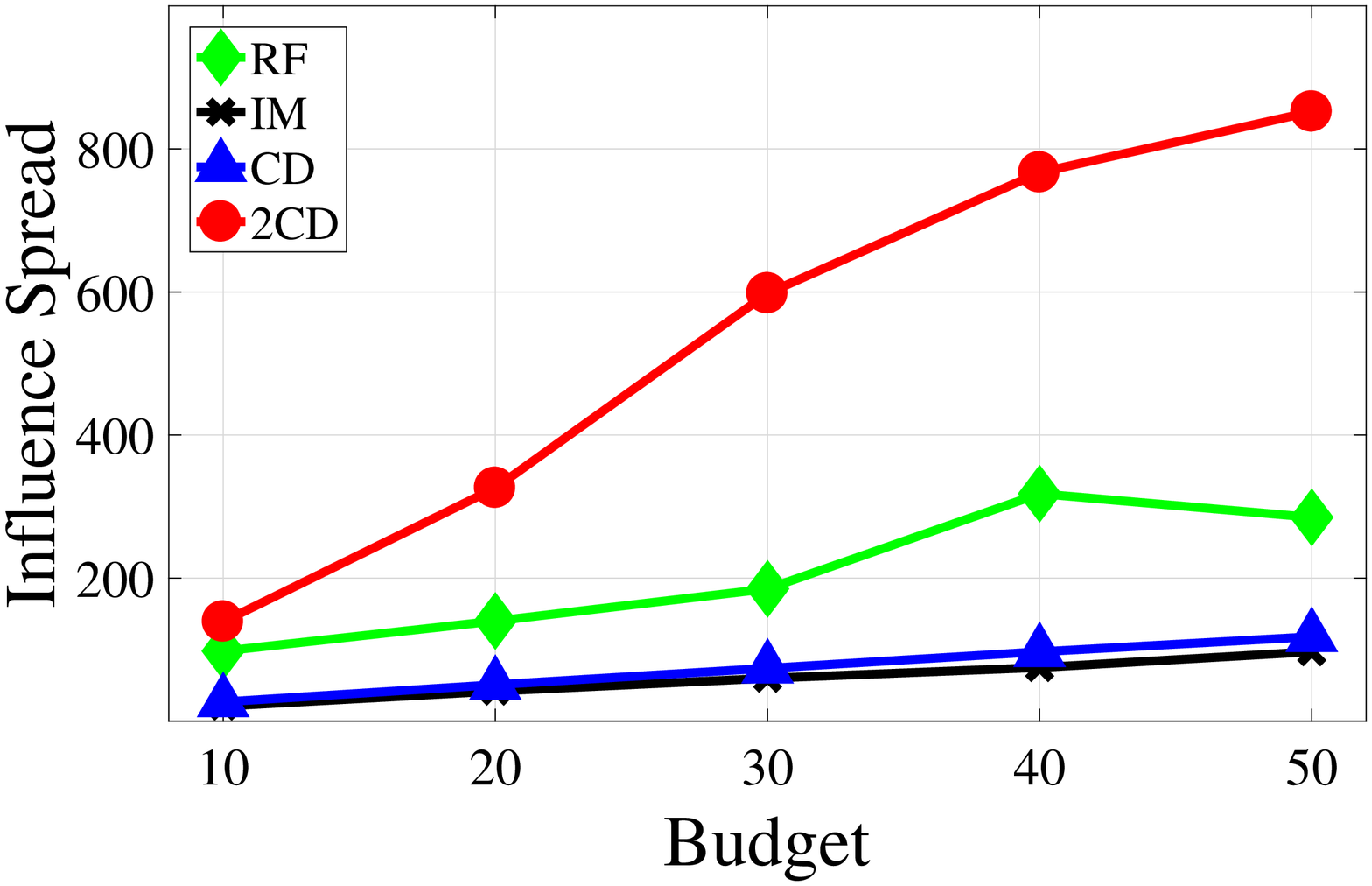}\label{livejournal_Non-ada_alpha=06}}
            \vspace{-2mm}
            \centerline{\quad \footnotesize{$\alpha$=0.6}}
            \vspace{1mm}
            \end{minipage}
            \hspace{-1mm}
            \begin{minipage}[h]{0.24\textwidth}
            \centerline{\includegraphics[width=1\textwidth]{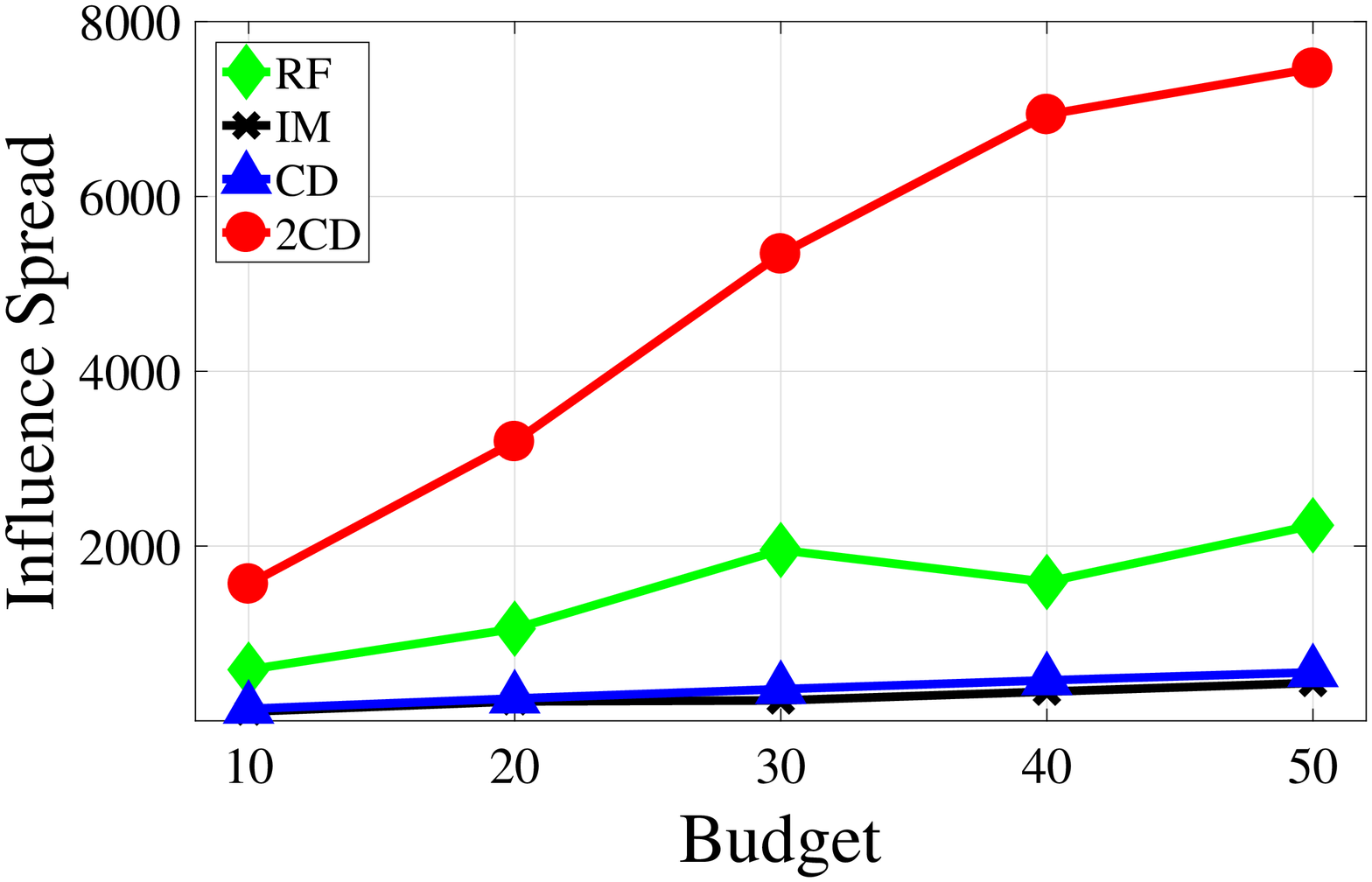}\label{livejournal_Non-ada_alpha=10}}
            \vspace{-2mm}
            \centerline{\quad \footnotesize{$\alpha$=1.0}}
            \vspace{1mm}
            \end{minipage}
        }
      \vspace{-1.5mm}
      \caption{Influence Spread in the Non-adaptive Case.}\label{Influence_spread_in_Non-ada}
      \vspace{-2mm}
  \end{figure*}

  \begin{figure*}[h]
    \centering
    \subfigure[Wiki-Vote]
        {
            \begin{minipage}[h]{0.24\textwidth}
            \centerline{\includegraphics[width=1\textwidth]{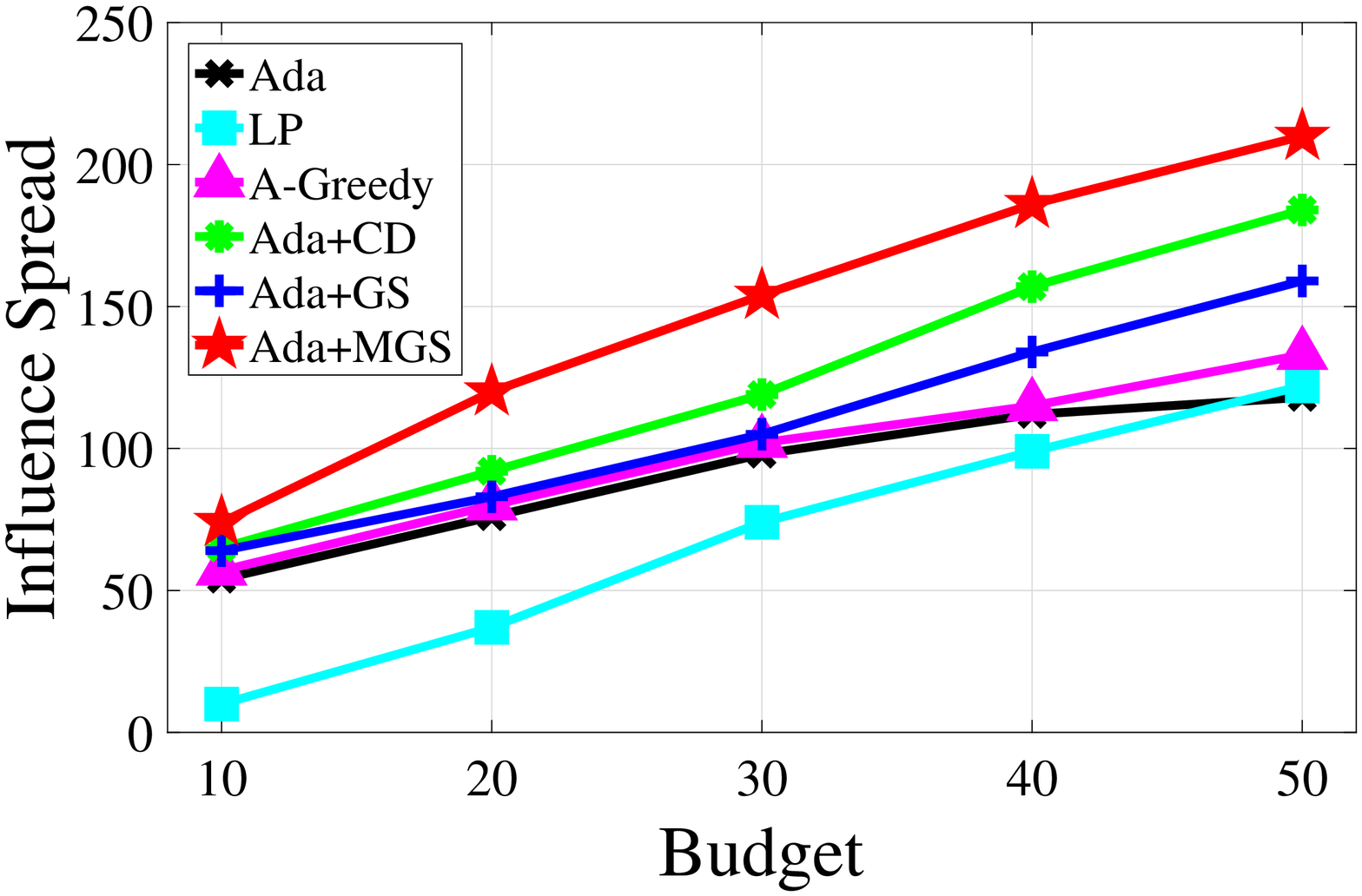}\label{wiki-vote_Ada_alpha=06}}
            \vspace{-2mm}
            \centerline{\quad \footnotesize{$\alpha$=0.6}}
            \vspace{1mm}
            \end{minipage}
            \hspace{-1mm}
            \begin{minipage}[h]{0.24\textwidth}
            \centerline{\includegraphics[width=1\textwidth]{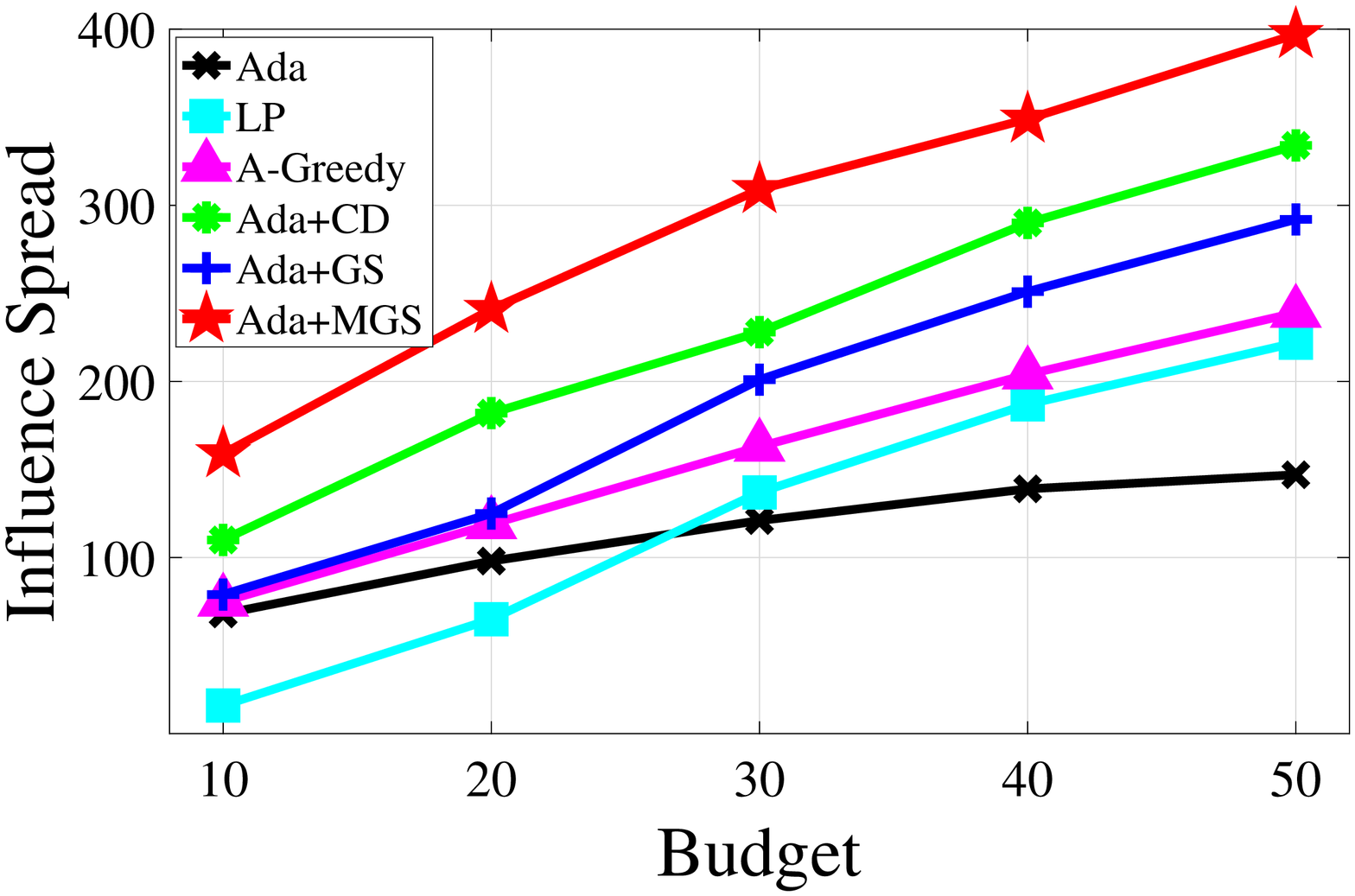}\label{wiki-vote_Ada_alpha=10}}
            \vspace{-2mm}
            \centerline{\quad \footnotesize{$\alpha$=1.0}}
            \vspace{1mm}
            \end{minipage}
        }
    \hspace{-3mm}\vspace{-0.5mm}
    \subfigure[Ca-CondMat]
        {
            \begin{minipage}[h]{0.24\textwidth}
            \centerline{\includegraphics[width=1\textwidth]{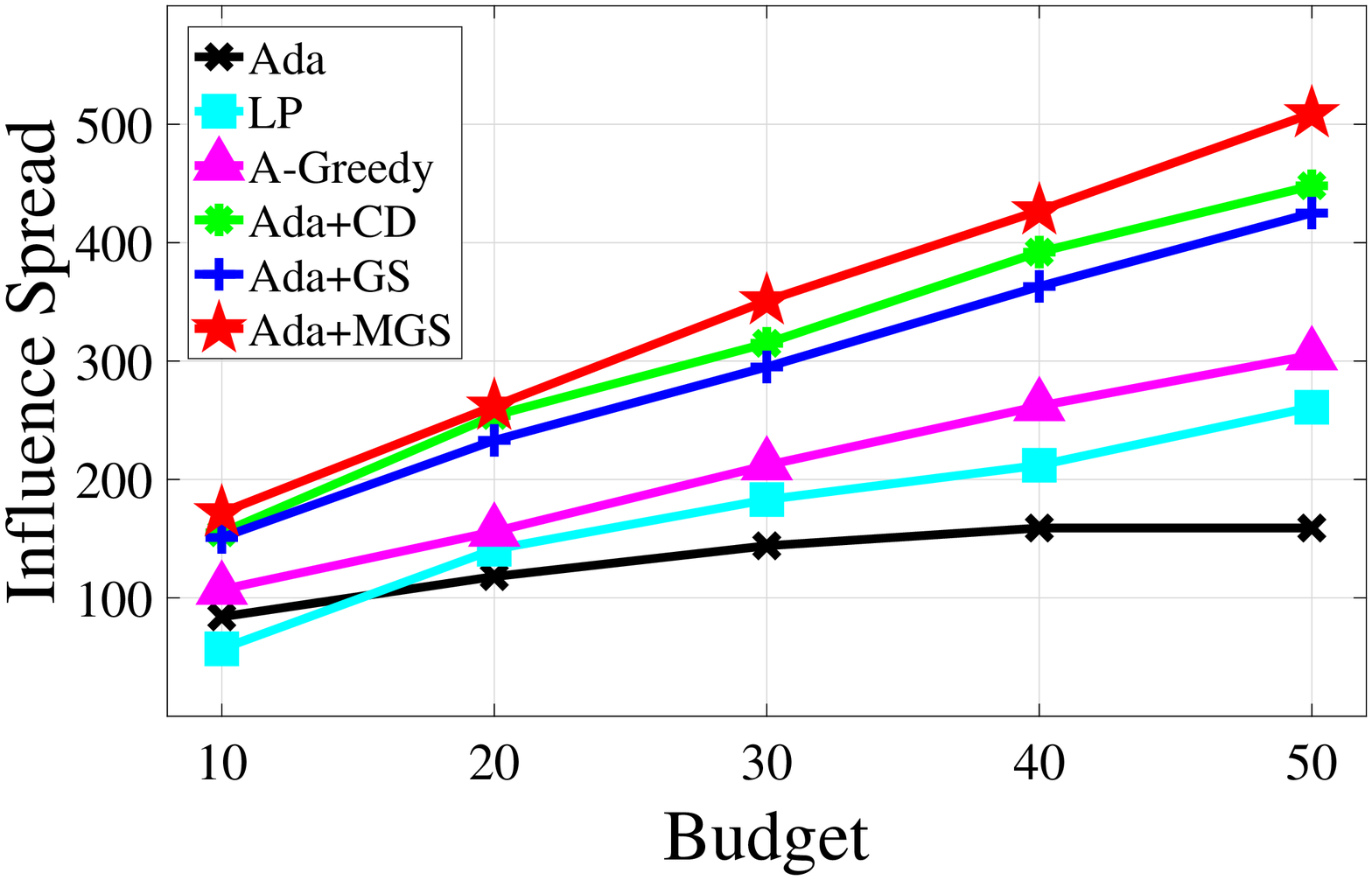}\label{Condmat_Ada_alpha=06}}
            \vspace{-2mm}
            \centerline{\quad \footnotesize{$\alpha$=0.6}}
            \vspace{1mm}
            \end{minipage}
            \hspace{-1mm}
            \begin{minipage}[h]{0.24\textwidth}
            \centerline{\includegraphics[width=1\textwidth]{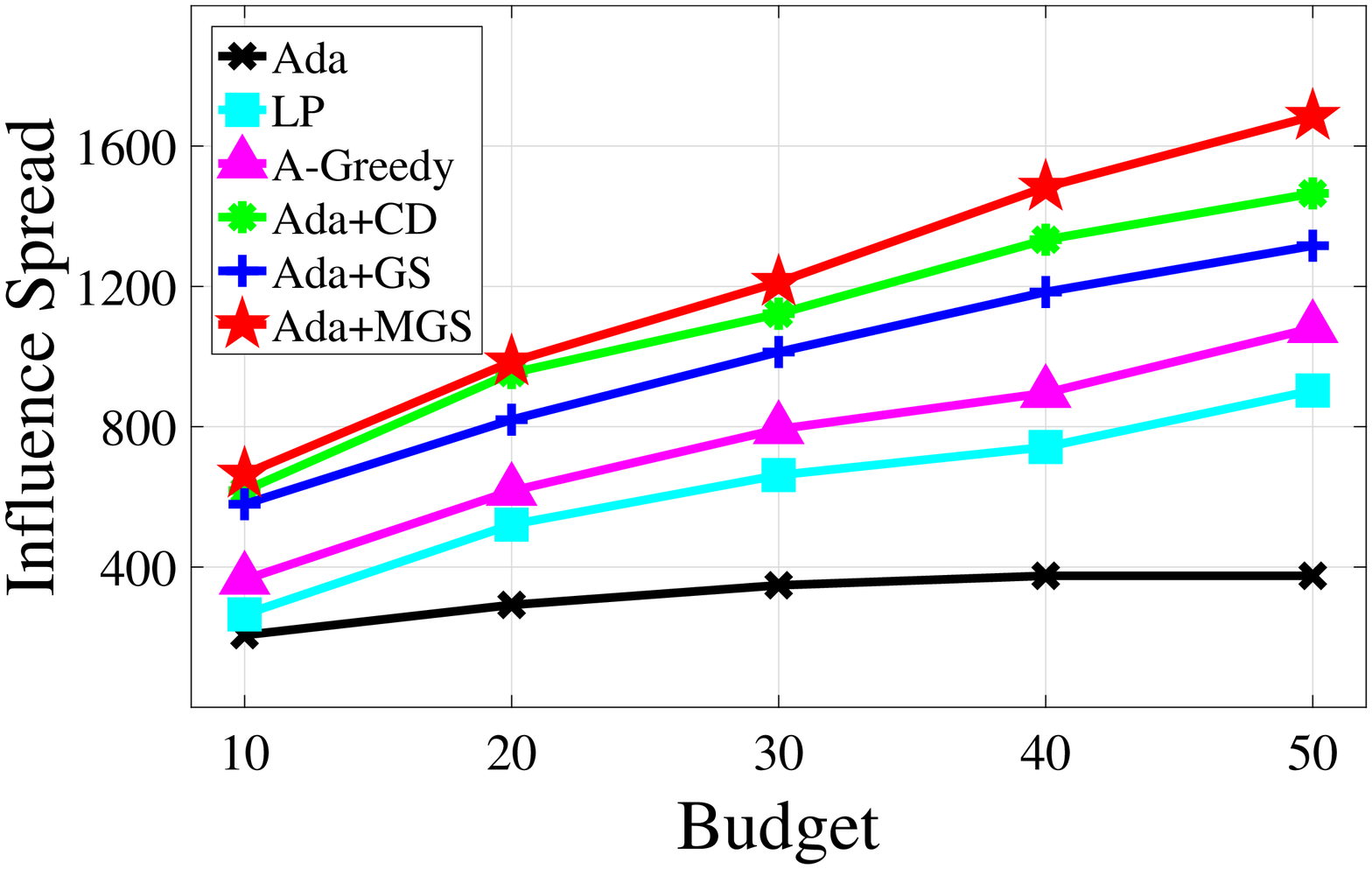}\label{Condmat_Ada_alpha=10}}
            \vspace{-2mm}
            \centerline{\quad \footnotesize{$\alpha$=1.0}}
            \vspace{1mm}
            \end{minipage}
        }
    \vspace{-0.5mm}
    \\
      \subfigure[com-Dblp]
      {
            \begin{minipage}[h]{0.24\textwidth}
            \centerline{\includegraphics[width=1\textwidth]{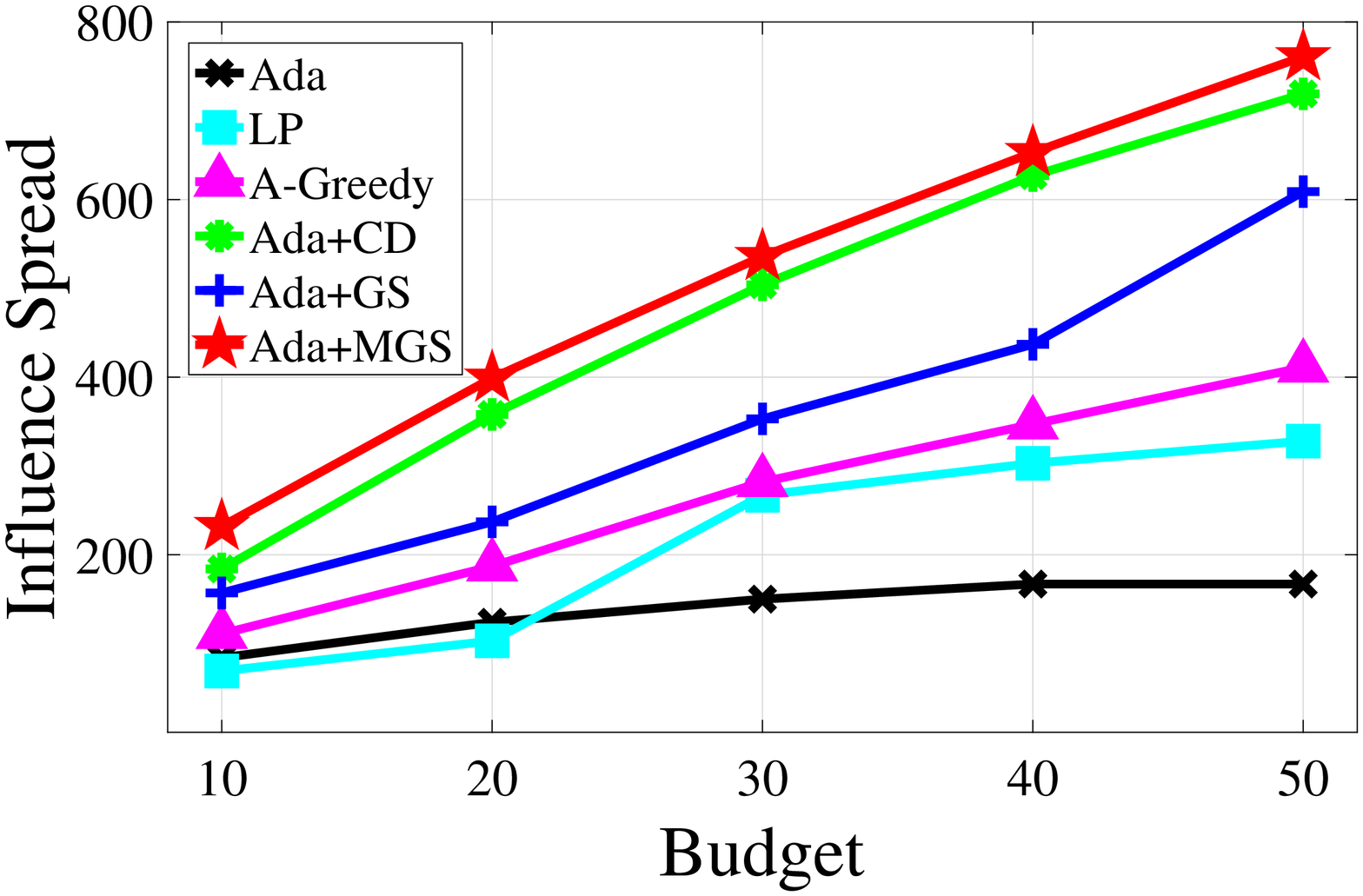}\label{Dblp_Ada_alpha=06}}
            \vspace{-2mm}
            \centerline{\quad \footnotesize{$\alpha$=0.6}}
            \vspace{1mm}
            \end{minipage}
            \hspace{-1mm}
            \begin{minipage}[h]{0.24\textwidth}
            \centerline{\includegraphics[width=1\textwidth]{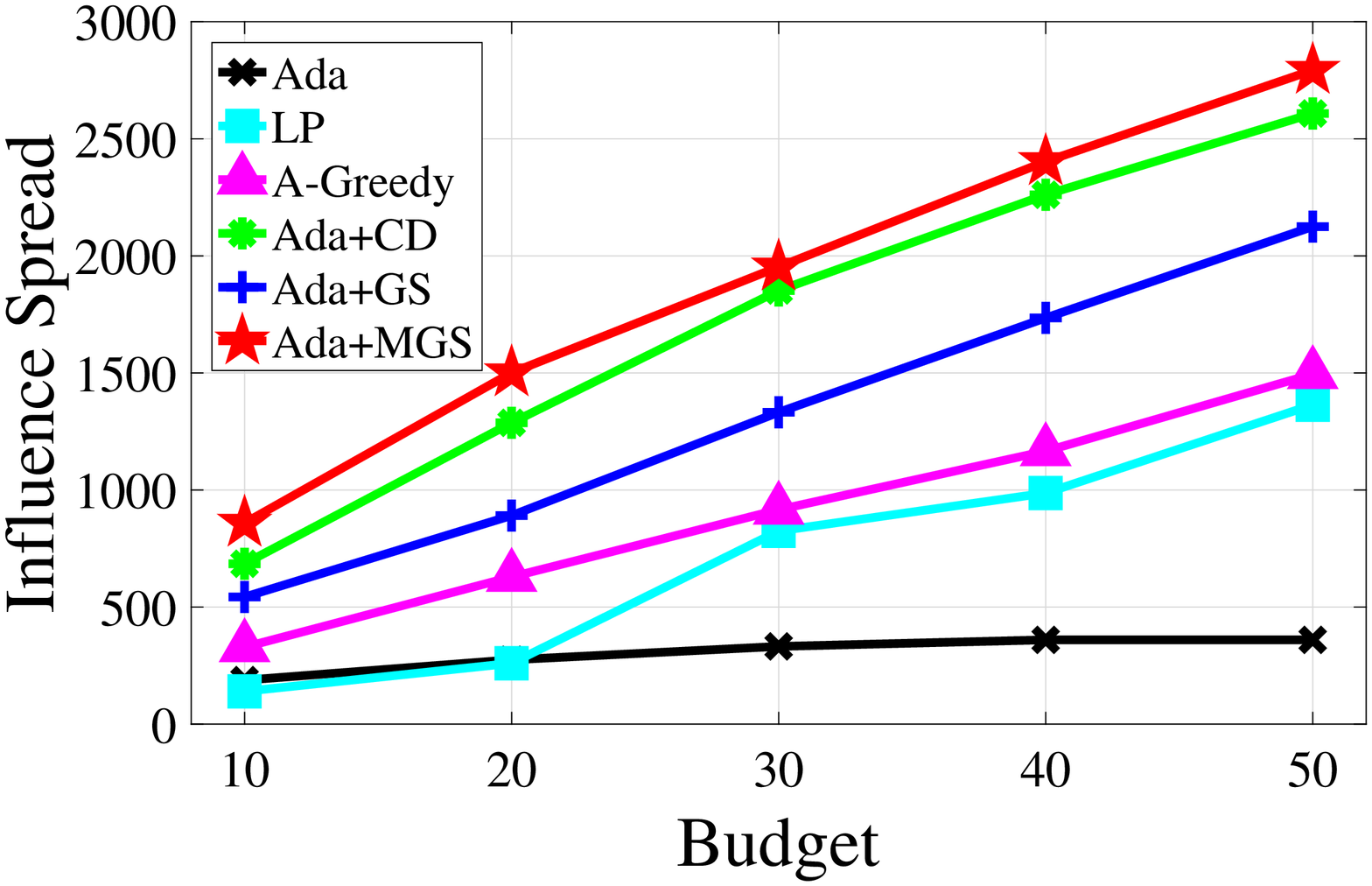}\label{Dblp_Ada_alpha=10}}
            \vspace{-2mm}
            \centerline{\quad \footnotesize{$\alpha$=1.0}}
            \vspace{1mm}
            \end{minipage}
        }
      \hspace{-3mm}
      \subfigure[soc-Livejournal]
      {
            \begin{minipage}[h]{0.24\textwidth}
            \centerline{\includegraphics[width=1\textwidth]{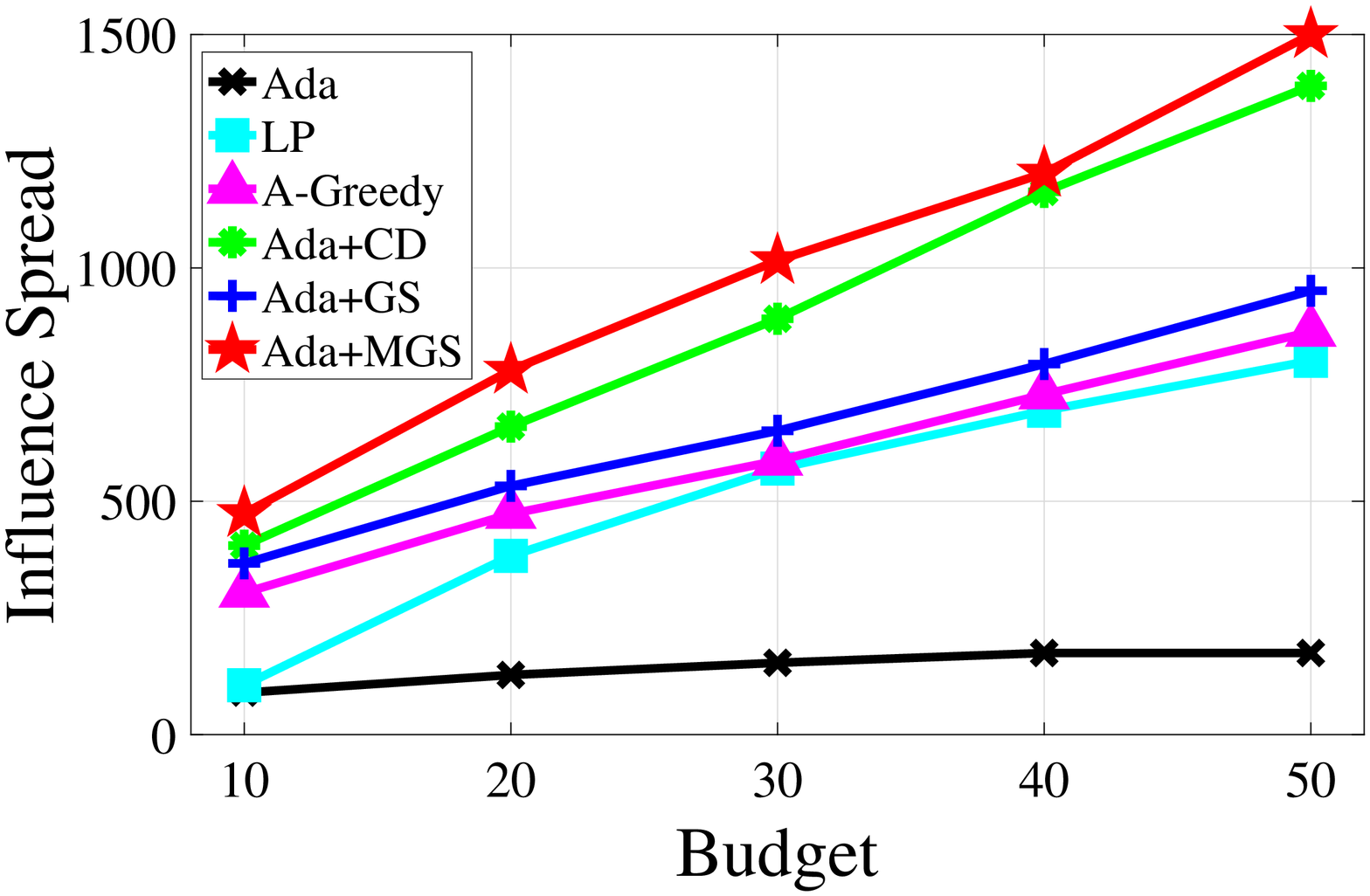}\label{livejournal_Ada_alpha=06}}
            \vspace{-2mm}
            \centerline{\quad \footnotesize{$\alpha$=0.6}}
            \vspace{1mm}
            \end{minipage}
            \hspace{-1mm}
            \begin{minipage}[h]{0.24\textwidth}
            \centerline{\includegraphics[width=1\textwidth]{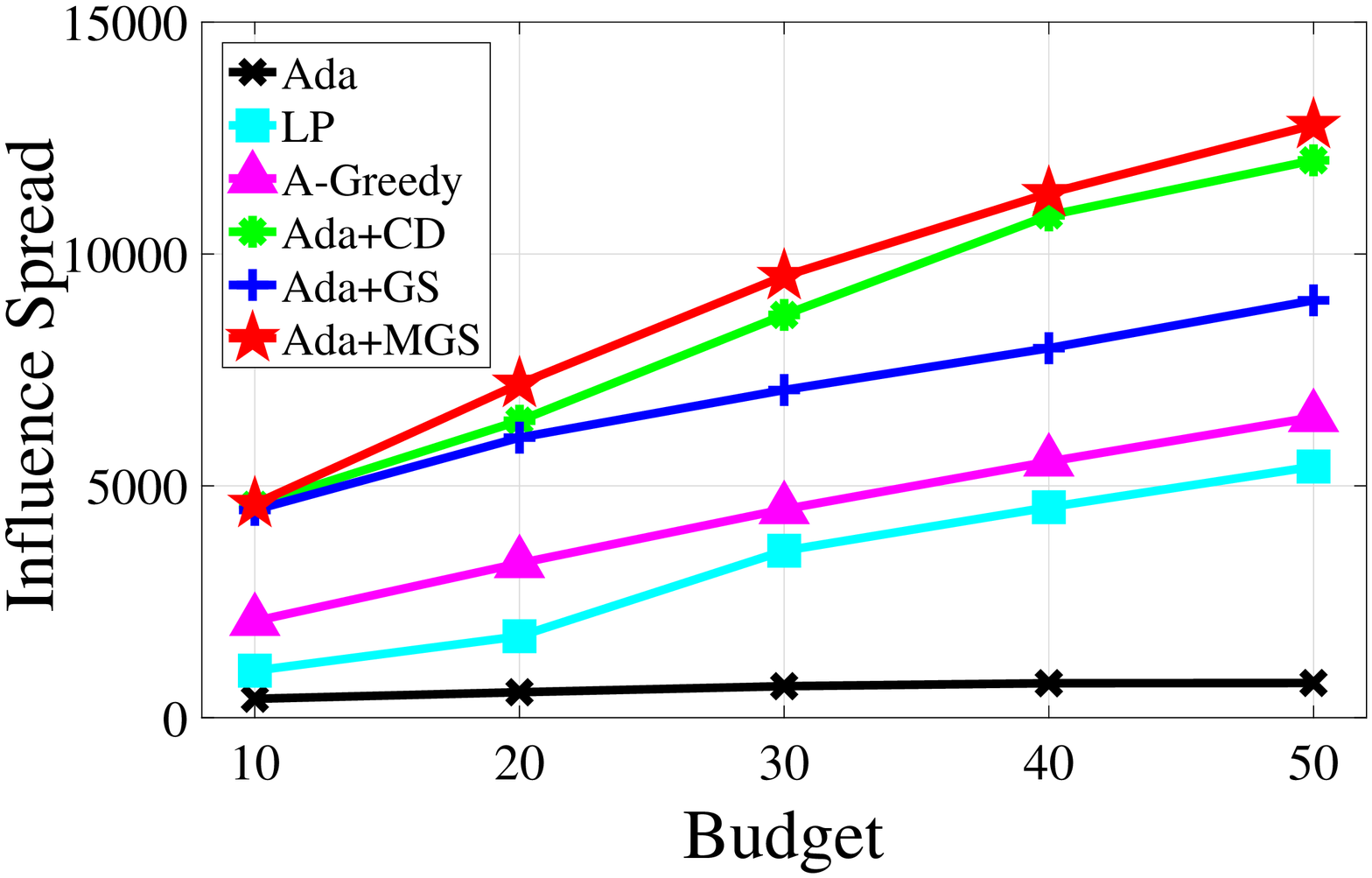}\label{livejournal_Ada_alpha=10}}
            \vspace{-2mm}
            \centerline{\quad \footnotesize{$\alpha$=1.0}}
            \vspace{1mm}
            \end{minipage}
        }
     \vspace{-1.5mm}
      \caption{Influence Spread in the Adaptive Case.}\label{Influence_spread_in_Ada}
      \vspace{-2mm}
  \end{figure*}

  \begin{figure*}[h]
    \centering
    \subfigure[Wiki-Vote]
        {
            \begin{minipage}[h]{0.24\textwidth}
            \centerline{\includegraphics[width=1\textwidth]{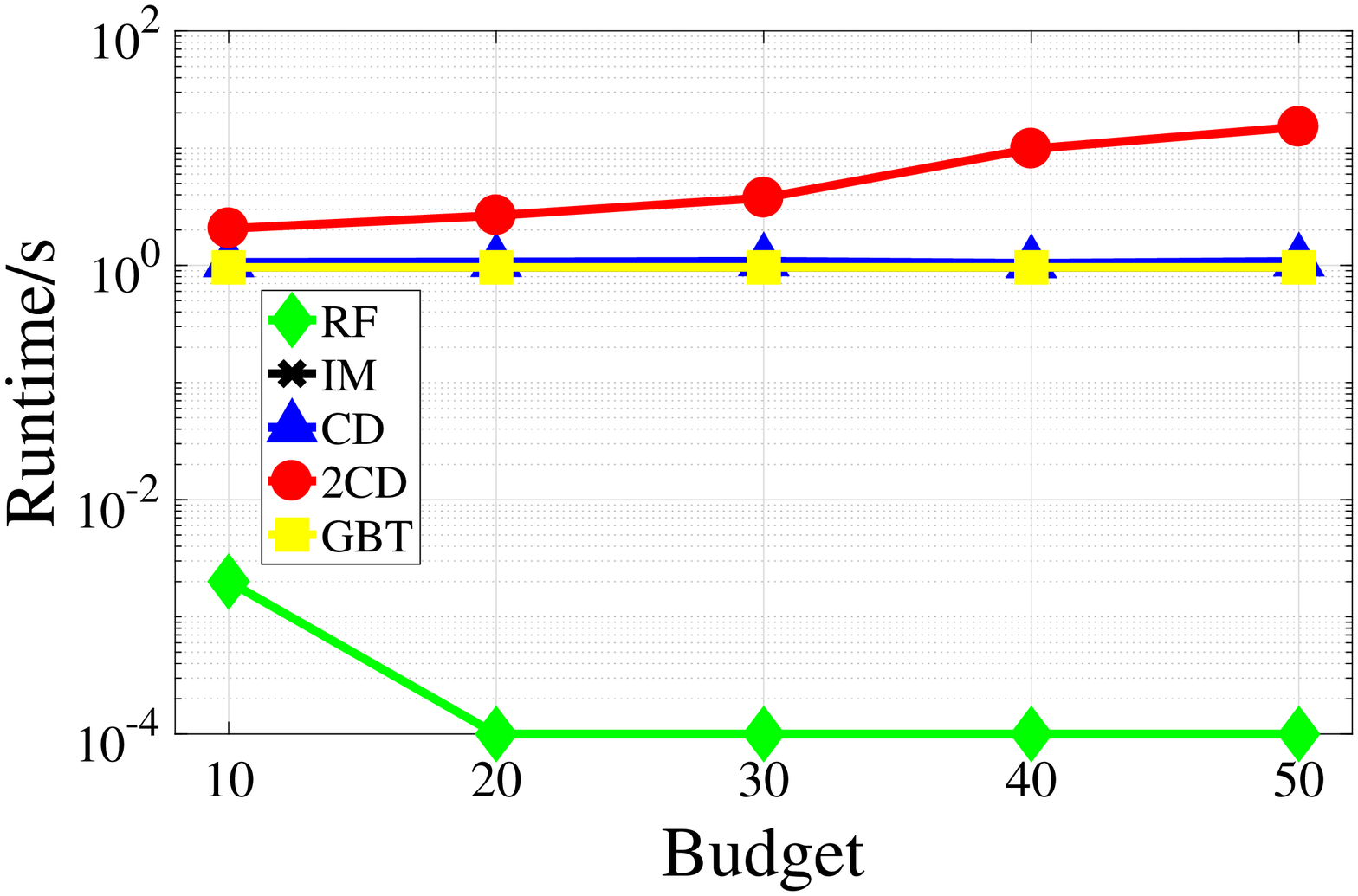}\label{wiki-vote_Non-ada_alpha-Runtime=06}}
            \vspace{-2mm}
            \centerline{\quad \footnotesize{$\alpha$=0.6}}
            \vspace{1mm}
            \end{minipage}
            \hspace{-1mm}
            \begin{minipage}[h]{0.24\textwidth}
            \centerline{\includegraphics[width=1\textwidth]{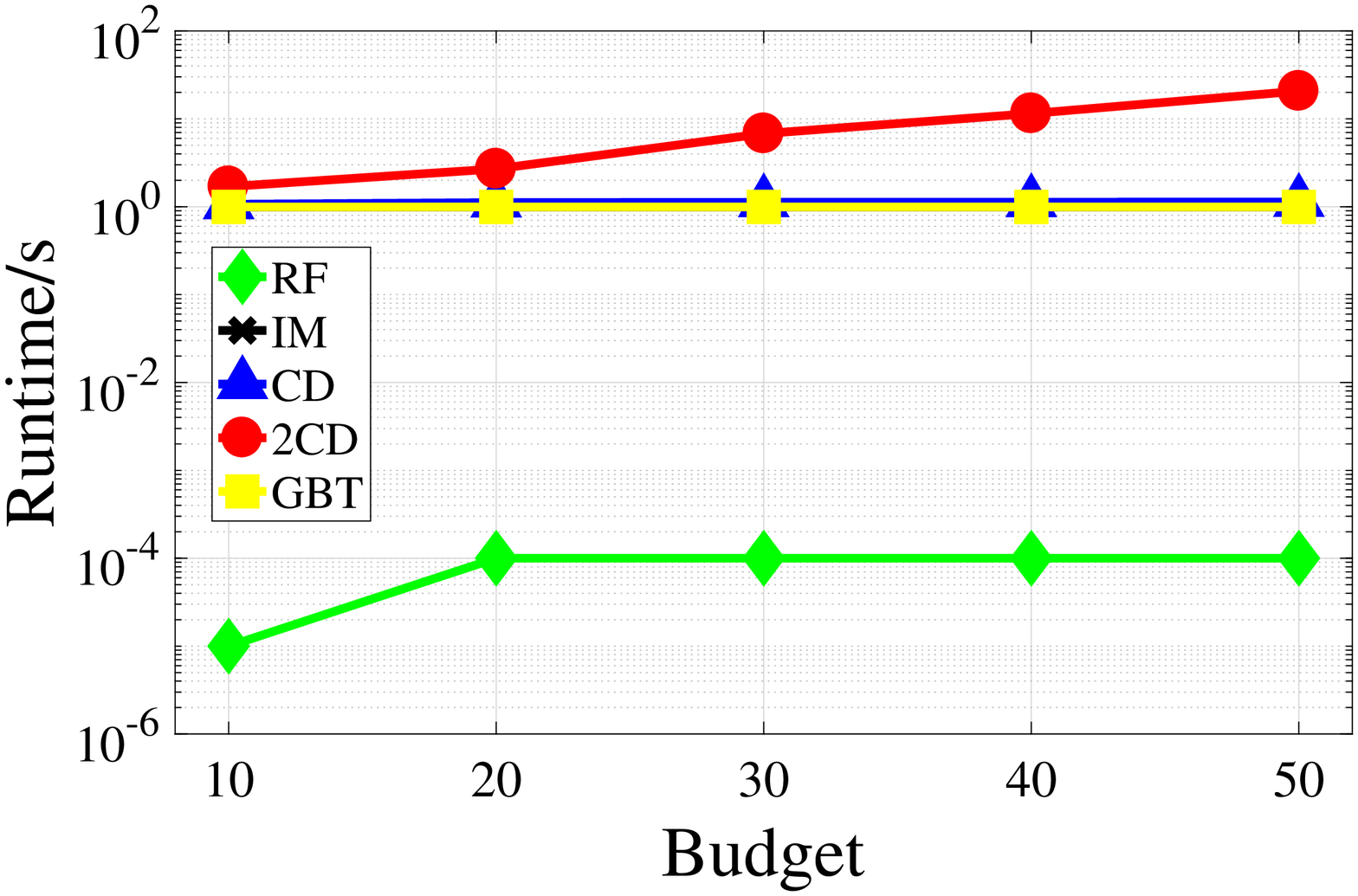}\label{wiki-vote_Non-ada_alpha-Runtime=10}}
            \vspace{-2mm}
            \centerline{\quad \footnotesize{$\alpha$=1.0}}
            \vspace{1mm}
            \end{minipage}
        }
    \hspace{-3mm}\vspace{-0.5mm}
    \subfigure[Ca-CondMat]
        {
            \begin{minipage}[h]{0.24\textwidth}
            \centerline{\includegraphics[width=1\textwidth]{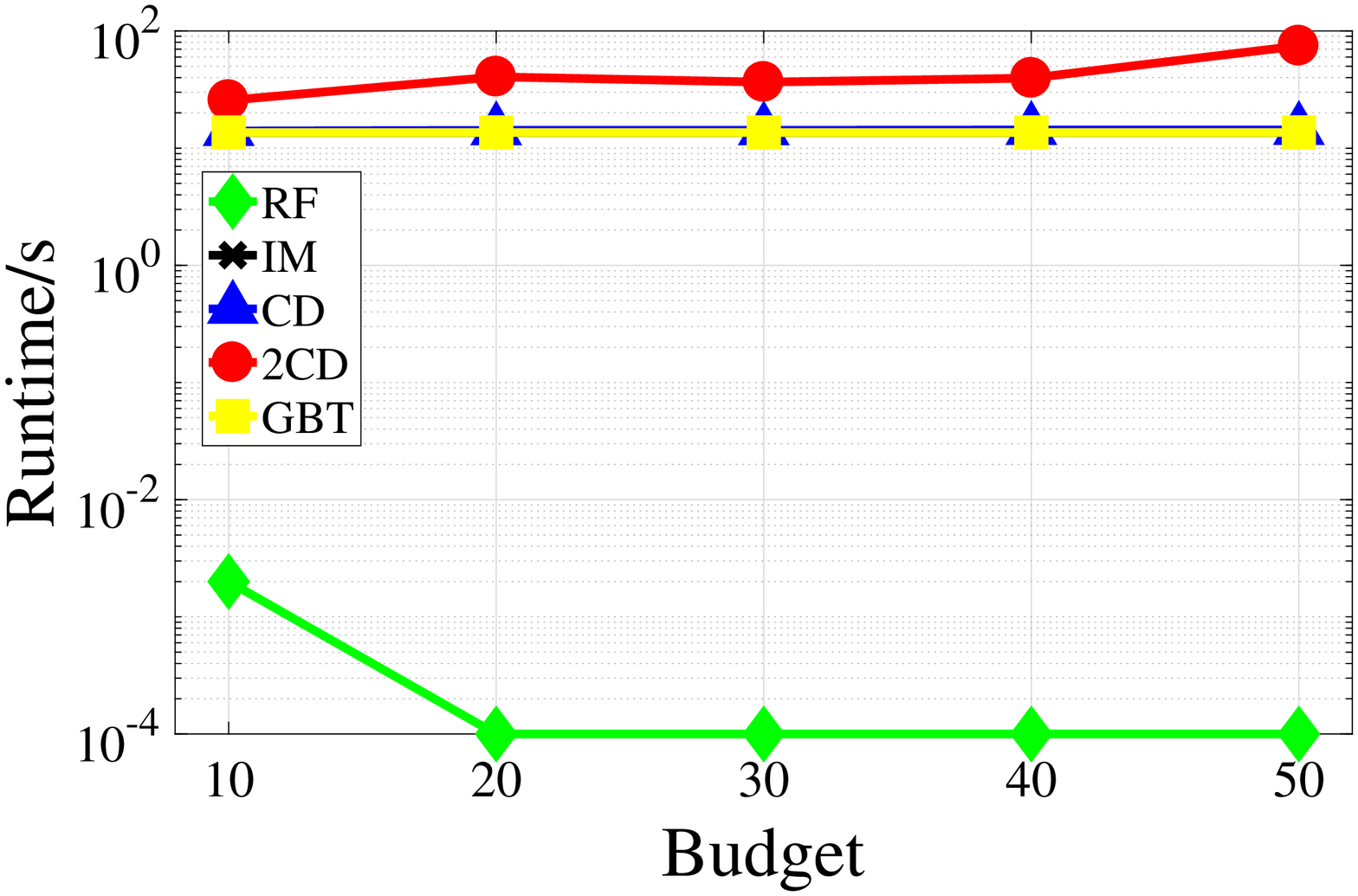}\label{Condmat_Non-ada_alpha-Runtime=06}}
            \vspace{-2mm}
            \centerline{\quad \footnotesize{$\alpha$=0.6}}
            \vspace{1mm}
            \end{minipage}
            \hspace{-1mm}
            \begin{minipage}[h]{0.24\textwidth}
            \centerline{\includegraphics[width=1\textwidth]{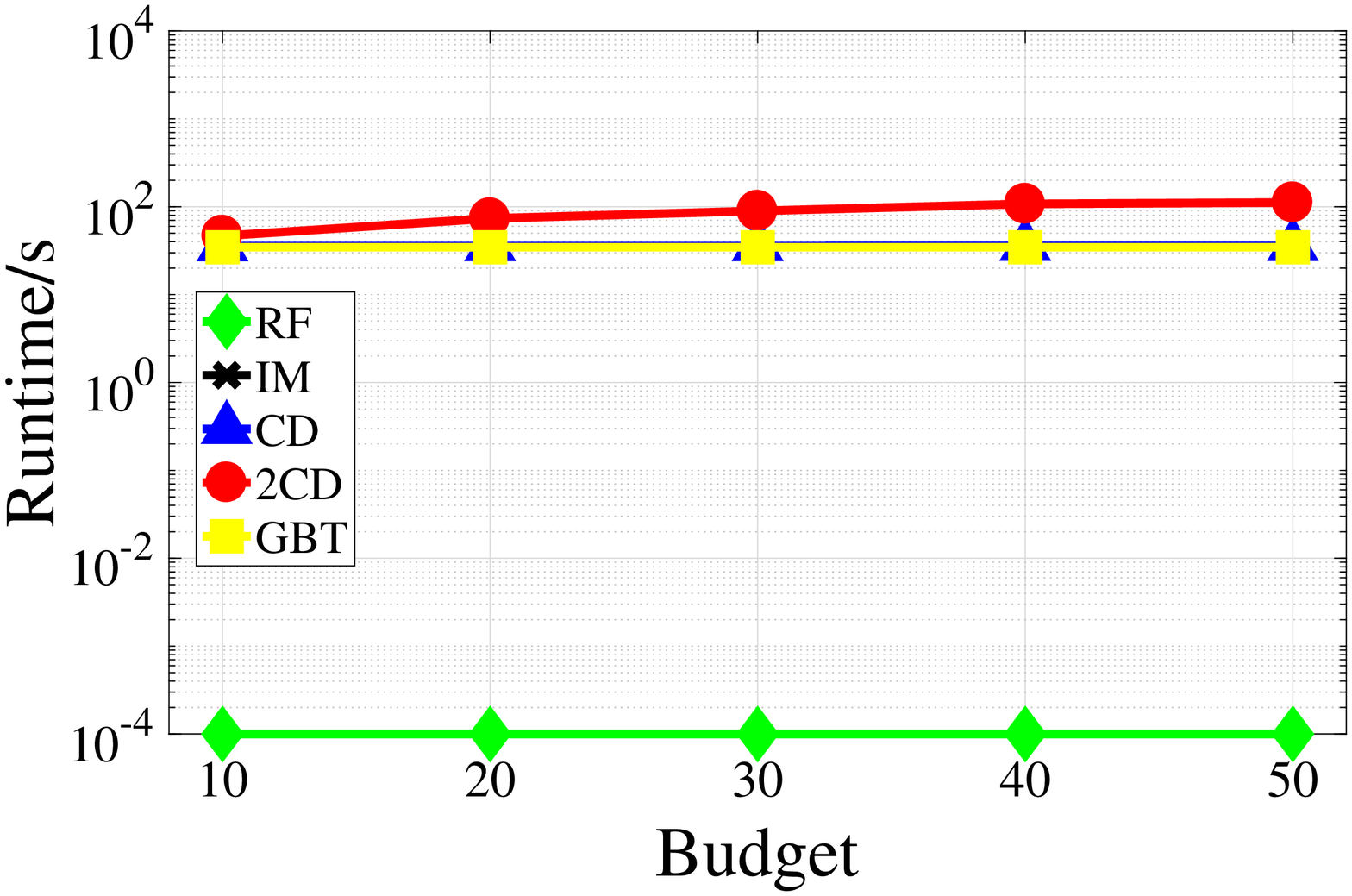}\label{Condmat_Non-ada_alpha-Runtime=10}}
            \vspace{-2mm}
            \centerline{\quad \footnotesize{$\alpha$=1.0}}
            \vspace{1mm}
            \end{minipage}
        }
    \vspace{-0.5mm}
    \\
      \subfigure[com-Dblp]
      {
            \begin{minipage}[h]{0.24\textwidth}
            \centerline{\includegraphics[width=1\textwidth]{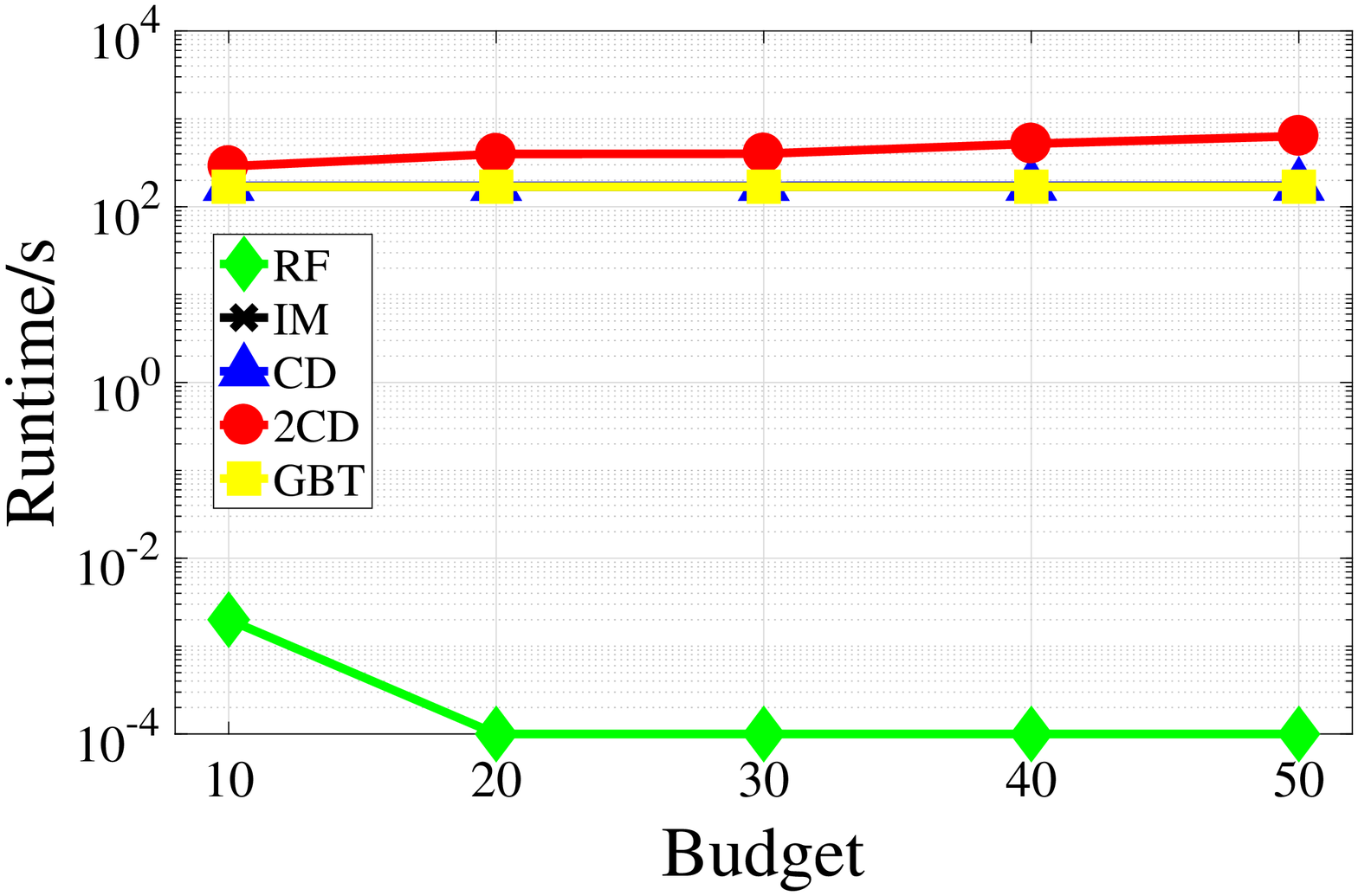}\label{Dblp_Non-ada_alpha-Runtime=06}}
            \vspace{-2mm}
            \centerline{\quad \footnotesize{$\alpha$=0.6}}
            \vspace{1mm}
            \end{minipage}
            \hspace{-1mm}
            \begin{minipage}[h]{0.24\textwidth}
            \centerline{\includegraphics[width=1\textwidth]{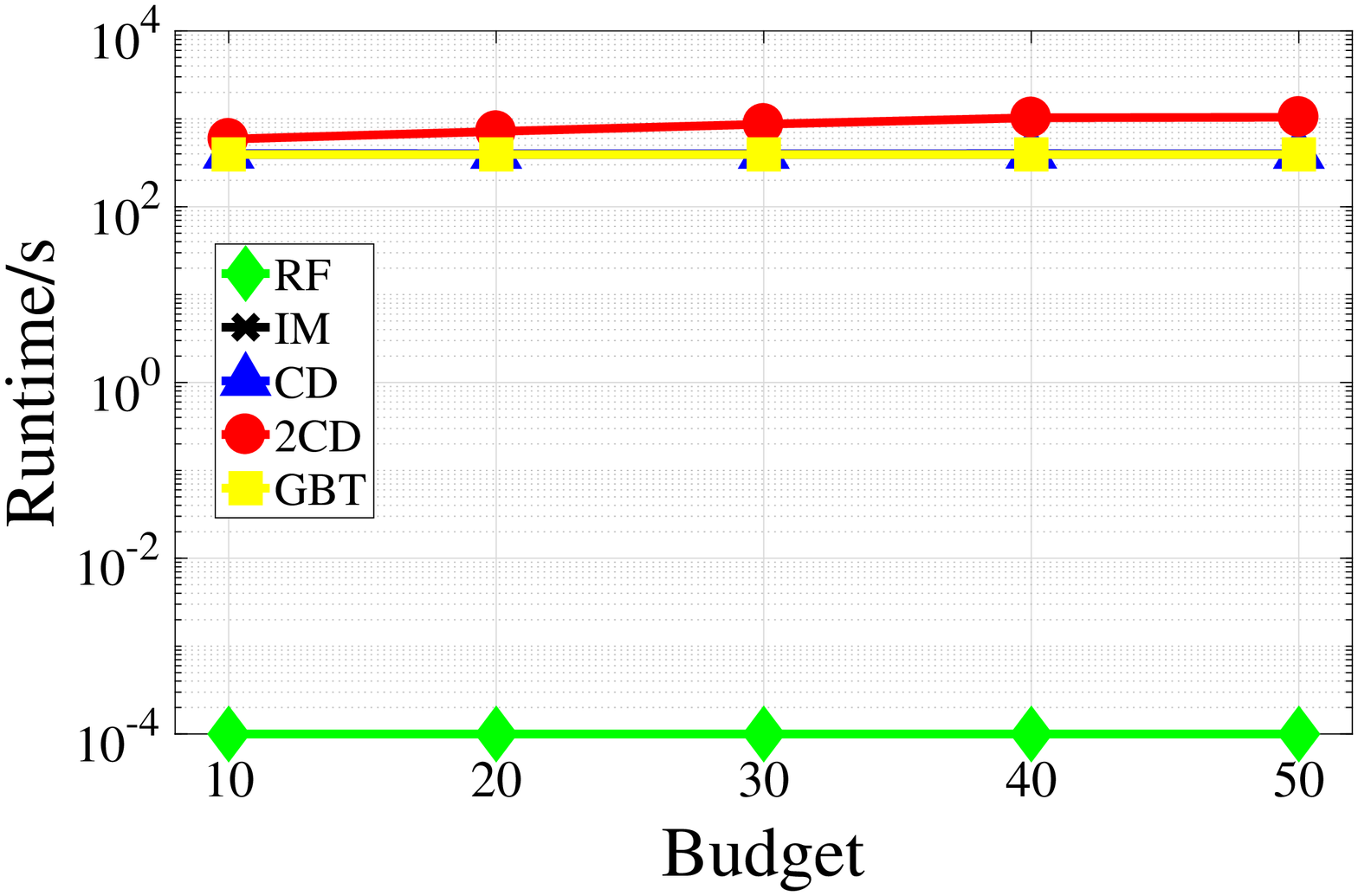}\label{Dblp_Non-ada_alpha-Runtime=10}}
            \vspace{-2mm}
            \centerline{\quad \footnotesize{$\alpha$=1.0}}
            \vspace{1mm}
            \end{minipage}
        }
      \hspace{-3mm}
      \subfigure[soc-Livejournal]
      {
            \begin{minipage}[h]{0.24\textwidth}
            \centerline{\includegraphics[width=1\textwidth]{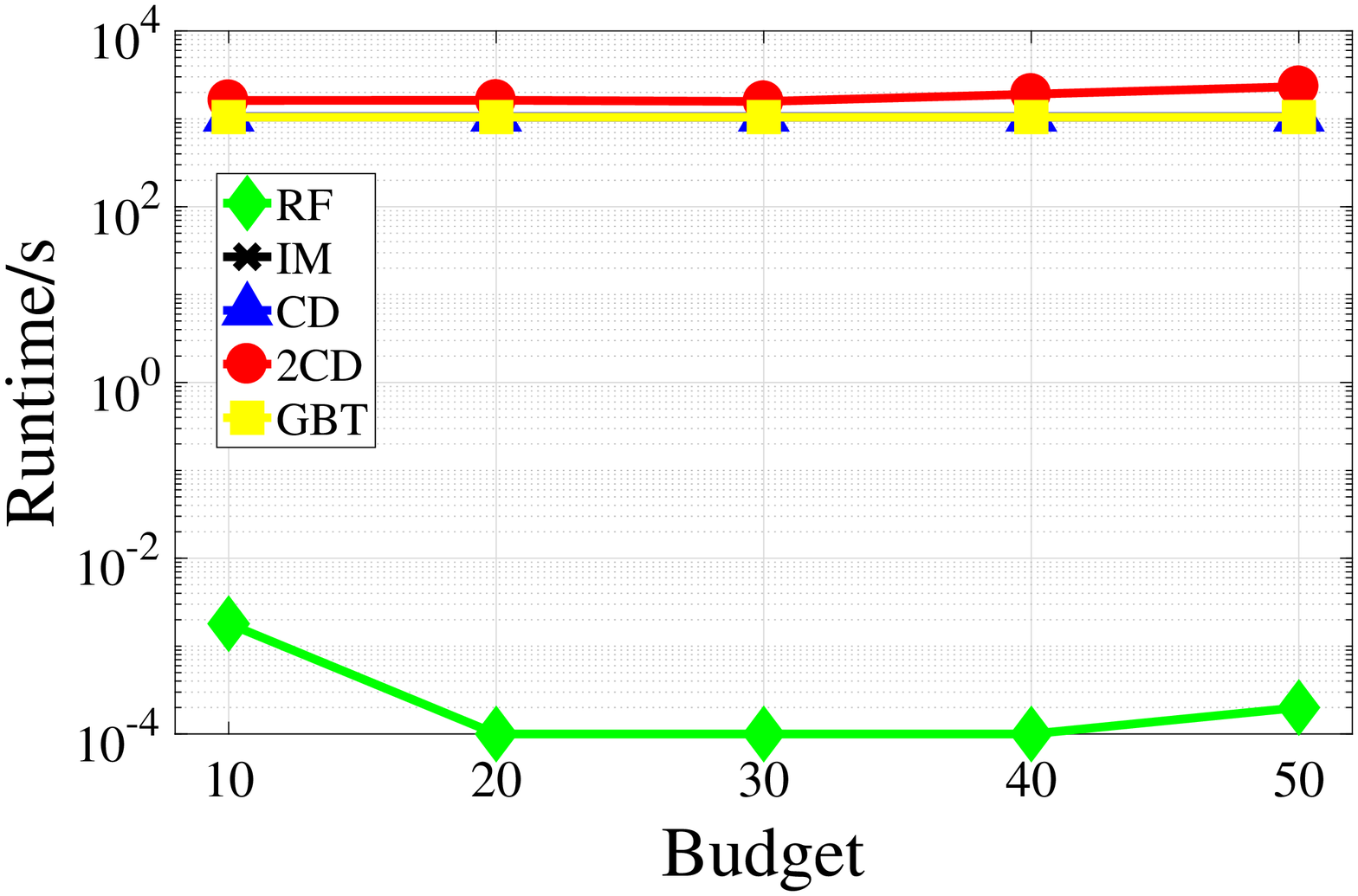}\label{livejournal_Non-ada_alpha-Runtime=06}}
            \vspace{-2mm}
            \centerline{\quad \footnotesize{$\alpha$=0.6}}
            \vspace{1mm}
            \end{minipage}
            \hspace{-1mm}
            \begin{minipage}[h]{0.24\textwidth}
            \centerline{\includegraphics[width=1\textwidth]{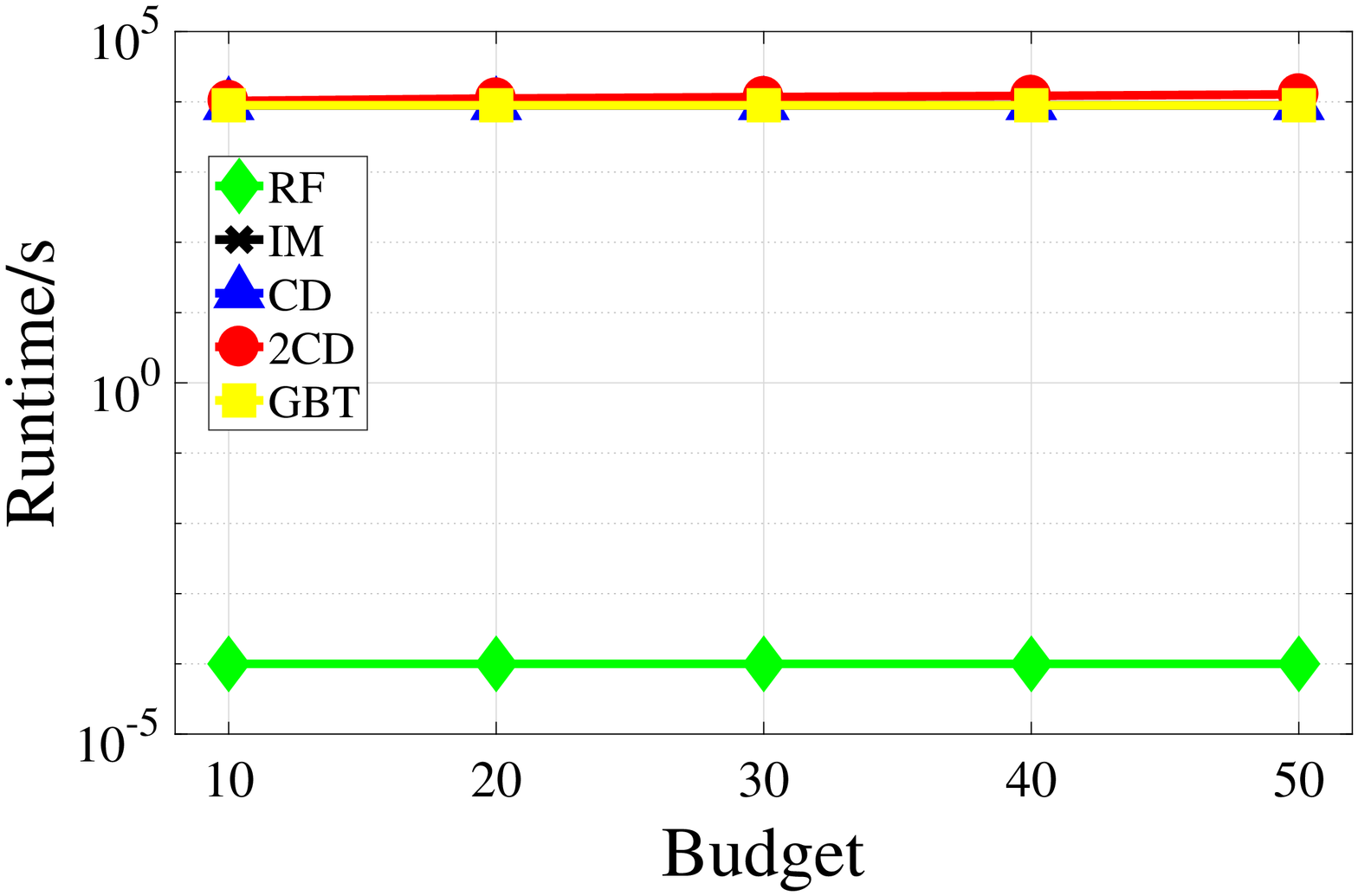}\label{livejournal_Non-ada_alpha-Runtime=10}}
            \vspace{-2mm}
            \centerline{\quad \footnotesize{$\alpha$=1.0}}
            \vspace{1mm}
            \end{minipage}
        }
     \vspace{-1.5mm}
      \caption{Running Time in the Non-adaptive Case.}\label{Running_Time_in_Non-ada}
      \vspace{-2mm}
  \end{figure*}

  \begin{figure*}[h]
    \centering
    \subfigure[Wiki-Vote]
        {
            \begin{minipage}[h]{0.24\textwidth}
            \centerline{\includegraphics[width=1\textwidth]{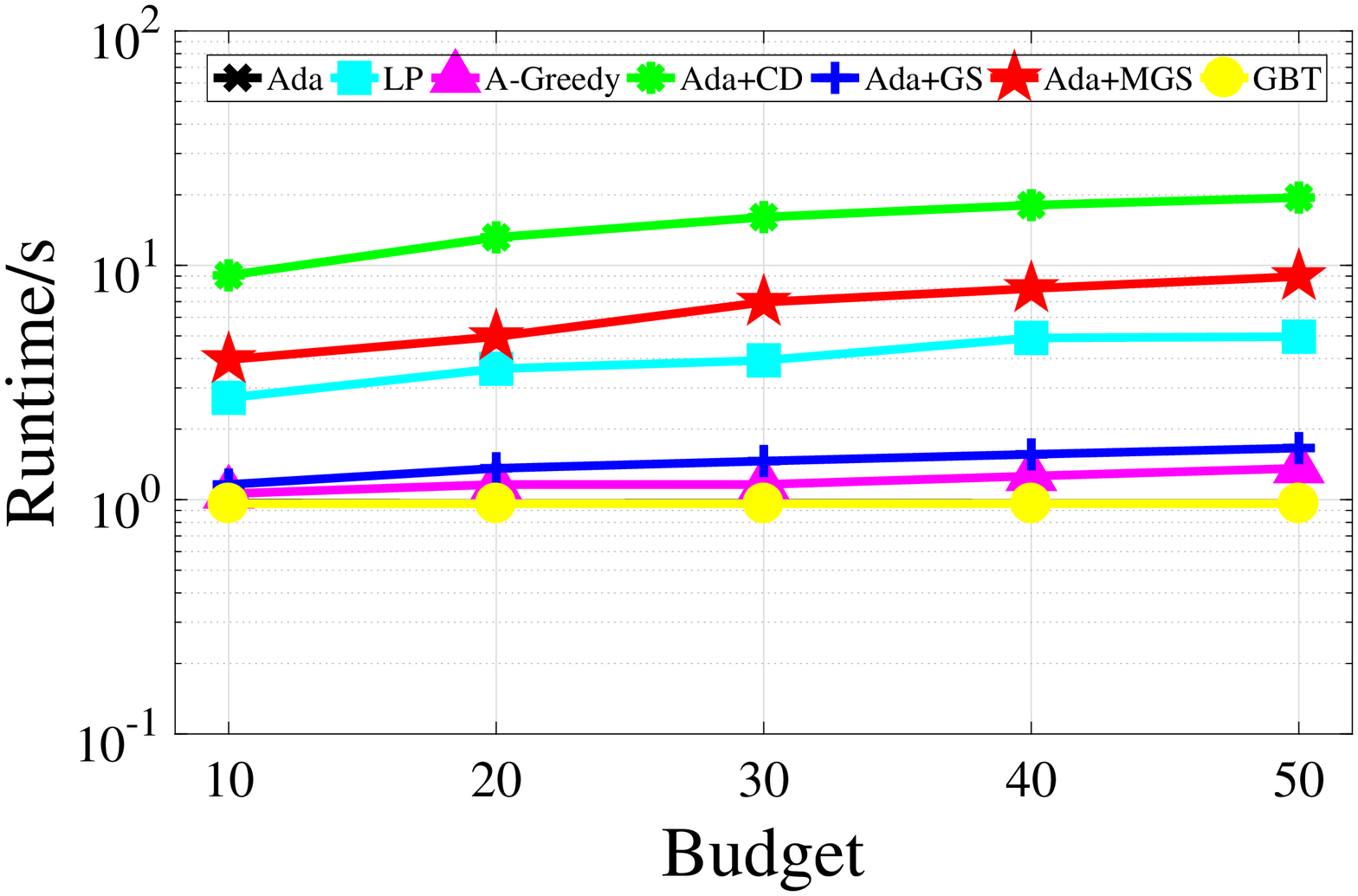}\label{wiki-vote_Ada_alpha=06-Runtime}}
            \vspace{-2mm}
            \centerline{\quad \footnotesize{$\alpha$=0.6}}
            \vspace{1mm}
            \end{minipage}
            \hspace{-1mm}
            \begin{minipage}[h]{0.24\textwidth}
            \centerline{\includegraphics[width=1\textwidth]{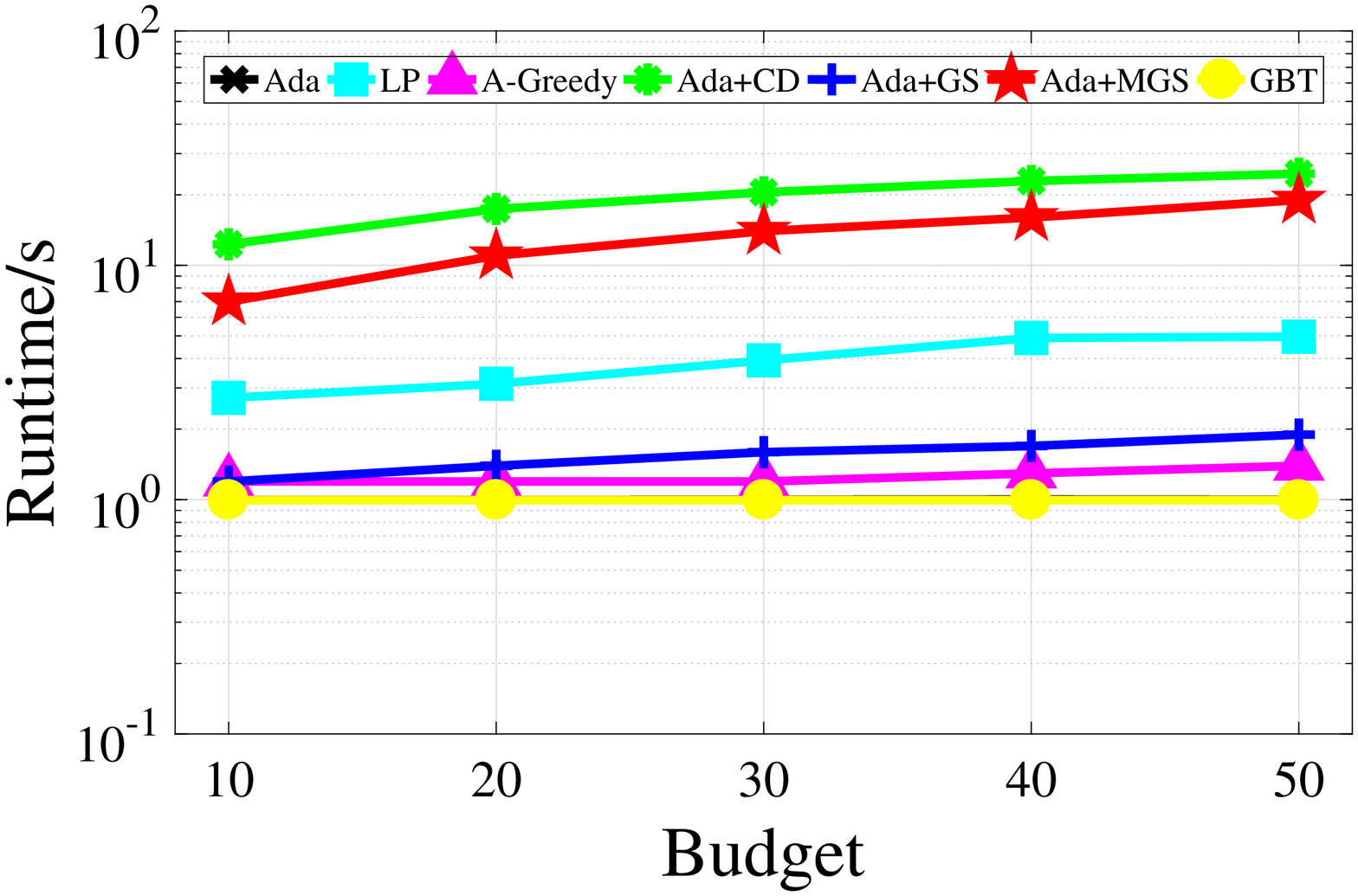}\label{wiki-vote_Ada_alpha=10-Runtime}}
            \vspace{-2mm}
            \centerline{\quad \footnotesize{$\alpha$=1.0}}
            \vspace{1mm}
            \end{minipage}
        }
    \hspace{-3mm}\vspace{-0.5mm}
    \subfigure[Ca-CondMat]
        {
            \begin{minipage}[h]{0.24\textwidth}
            \centerline{\includegraphics[width=1\textwidth]{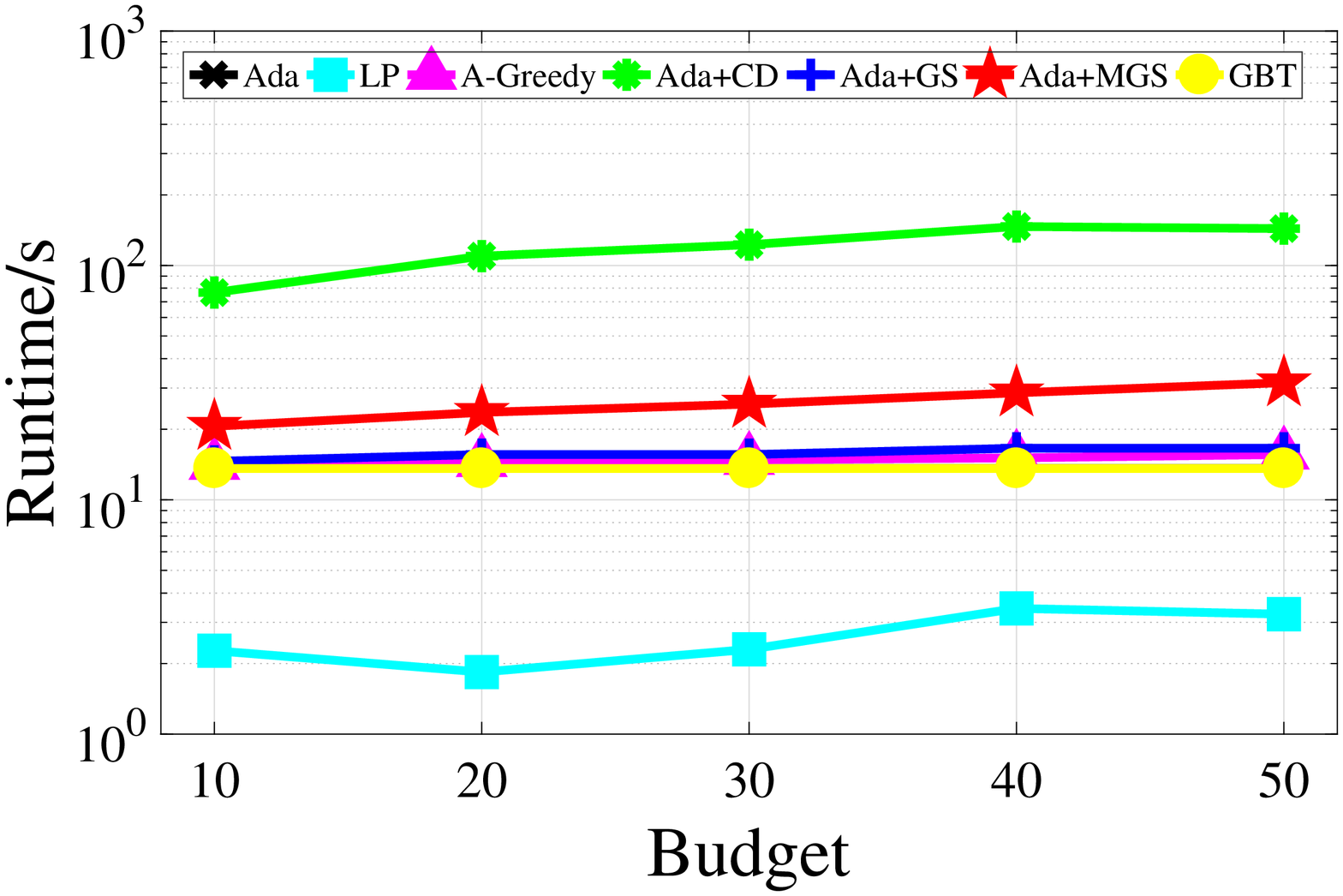}\label{Condmat_Ada_alpha=06-Runtime}}
            \vspace{-2mm}
            \centerline{\quad \footnotesize{$\alpha$=0.6}}
            \vspace{1mm}
            \end{minipage}
            \hspace{-1mm}
            \begin{minipage}[h]{0.24\textwidth}
            \centerline{\includegraphics[width=1\textwidth]{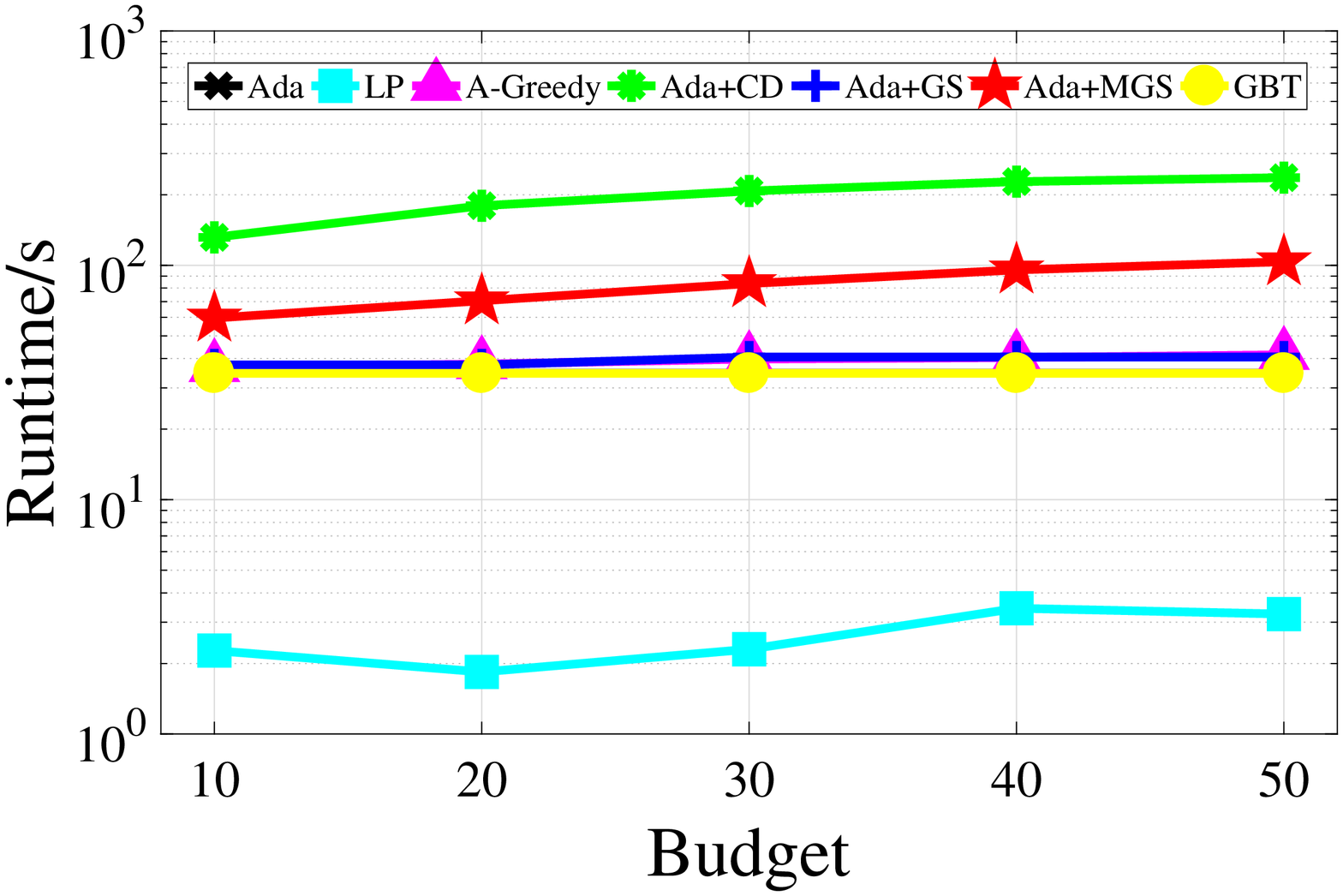}\label{Condmat_Ada_alpha=10-Runtime}}
            \vspace{-2mm}
            \centerline{\quad \footnotesize{$\alpha$=1.0}}
            \vspace{1mm}
            \end{minipage}
        }
      \vspace{-1.5mm}
      \caption{Running Time in the Adaptive Case.}\label{Run_time_in_Ada}
      \vspace{-2mm}
  \end{figure*}

\begin{table*}[h]
\caption{Sensitivity to the seed probability function}\label{table_sensitivity}
\vspace{-0.6cm}
\footnotesize
\begin{center}
\renewcommand\arraystretch{1}
\begin{tabular}{|p{1.3cm}<{\centering}|p{0.25cm}<{\centering}|p{0.65cm}<{\centering}|p{0.65cm}<{\centering}|p{1.15cm}<{\centering}|p{0.65cm}<{\centering}|p{0.65cm}<{\centering}|p{1.15cm}<{\centering}|p{0.65cm}<{\centering}|p{0.65cm}<{\centering}|p{1.15cm}<{\centering}|p{0.65cm}<{\centering}|p{0.65cm}<{\centering}|p{1.15cm}<{\centering}|}
\hline
\multirow{2}*{Dataset} & \multirow{2}*{B} & \multicolumn{3}{|c|}{2CD} & \multicolumn{3}{|c|}{Ada+CD} & \multicolumn{3}{|c|}{Ada+GS} & \multicolumn{3}{|c|}{Ada+MGS}\\
\cline{3-14}
\multirow{2}*{} & \multirow{2}*{} & S1 & S2 & Reduction & S1 & S2 & Reduction & S1 & S2 & Reduction & S1 & S2 & Reduction\\
\hline
\multirow{3}*{Wiki-Vote} & 10 & 67 & 57 & 14.9\% & 110 & 65 & 40.9\% & 79 & 70 & 11.4\% & 159 & 135 & 15.1\%\\
\cline{2-14}
\multirow{5}*{} &20& 146 & 98 & 32.9\% & 182 & 121 & 33.5\% & 125 & 97 & 22.4\% & 241 & 218 & 9.5\%\\
\cline{2-14}
\multirow{3}*{} &30& 171 & 119 & 30.4\% & 228 & 162 & 28.9\% & 201 & 170 & 15.4\% & 309 & 274 & 11.3\%\\
\cline{2-14}
\multirow{5}*{} &40& 239 & 177 & 25.9\% & 290 & 184 & 36.6\% & 251 & 206 & 17.9\% &349 & 330 & 5.4\%\\
\cline{2-14}
\multirow{3}*{} &50& 293 & 208 & 29\% & 334 & 205 & 38.6\% & 292 & 272 & 6.8\% & 397 & 379 & 4.5\% \\
\hline
\multirow{3}*{Condmat} & 10& 411 & 252 & 38.7\% & 618 & 321 & 48.1\% & 579 & 407 & 29.7\% & 665 & 566 & 14.9\%\\
\cline{2-14}
\multirow{5}*{} & 20 & 616 & 390 & 35.6\% & 953 & 558 & 41.4\% & 820 & 593 & 27.7\% & 985 & 872 & 11.5\%\\
\cline{2-14}
\multirow{3}*{} & 30 & 803 & 595 & 25.9\% & 1122 & 733 & 34.7\% & 1013 & 747 & 26.3\% & 1212 & 1099 & 9.3\%\\
\cline{2-14}
\multirow{5}*{} & 40 & 934 & 686 & 26.6\% & 1333& 814 & 38.9\% & 1184 & 881 & 25.6\% & 1483 & 1282 & 13.6\%\\
\cline{2-14}
\multirow{3}*{} & 50 & 1132 & 863 & 23.8\% & 1465 & 859 & 41.3\% & 1376 & 1029 & 25.2\% & 1685 & 1498 & 11.1\%\\
\hline
\multirow{3}*{Dblp} & 10 & 358 & 251 & 29.9\% & 685 & 266 & 46.6\% & 543 & 345 & 36.5\% & 859 & 703 & 18.2\%\\
\cline{2-14}
\multirow{5}*{} & 20 & 802 & 581 & 27.6\% & 1287 & 553 & 57\% & 891 & 621 & 30.3\% & 1501 & 1176 & 21.7\%\\
\cline{2-14}
\multirow{3}*{} & 30 & 1019 & 724 & 28.9\% & 1852 & 851 & 54\% & 1332 & 880 & 33.9\% & 1953 & 1612 & 17.5\%\\
\cline{2-14}
\multirow{5}*{} & 40 & 1308 & 954 & 27.1\% & 2261 & 983 & 56.5\% & 1735 & 1322 & 23.8\% & 2403 & 2119 & 11.8\%\\
\cline{2-14}
\multirow{3}*{} & 50 & 1644 & 1201 & 26.9\% & 2608 & 1191 & 54.3\% & 2125 & 1607 & 24.4\% & 2793 & 2440 & 12.6\%\\
\hline
\multirow{3}*{Livejournal} & 10 & 1564 & 1073 & 31.4\% & 4591 & 2524 & 45\% & 4481 & 2626 & 41.4\% & 4610 & 3632 & 21.2\%\\
\cline{2-14}
\multirow{5}*{} & 20 & 3193 & 2333 & 26.9\% & 6405 & 4325 & 32.5\% & 6040 & 4159 & 31.1\% & 7199 & 6843 & 4.9\%\\
\cline{2-14}
\multirow{3}*{} & 30 & 5340 & 3222 & 39.7\% & 8686 & 5148 & 40.7\% & 7067 & 5114 & 27.6\% & 9523 & 8785 & 7.7\%\\
\cline{2-14}
\multirow{5}*{} & 40 & 6935 & 5462 & 21.2\% & 10831 & 6590 & 39.2\% & 7969 & 6776 & 15\% & 11310 & 10793 & 4.6\%\\
\cline{2-14}
\multirow{3}*{} & 50 & 7464 & 6280 & 15.9\% & 12913 & 7246 & 56.1\% & 9000 & 7541 & 16.2\% & 12781 & 12719 & 0.5\%\\
\hline
\end{tabular}
\end{center}
\vspace{-0.2cm}
\end{table*}
\normalsize

\begin{table*}[h]
\caption{The impact of the friendship paradox phenomenon}\label{table_FP}
\vspace{-0.6cm}
\footnotesize
\begin{center}
\renewcommand\arraystretch{1}
\begin{tabular}{|p{1.3cm}<{\centering}|p{0.3cm}<{\centering}|p{0.75cm}<{\centering}|p{0.75cm}<{\centering}|p{0.9cm}<{\centering}|p{0.75cm}<{\centering}|p{0.75cm}<{\centering}|p{0.9cm}<{\centering}|p{0.75cm}<{\centering}|p{1.25cm}<{\centering}|p{0.9cm}<{\centering}|p{0.75cm}<{\centering}|p{1.25cm}<{\centering}|p{0.9cm}<{\centering}|}
\hline
\multirow{3}*{Dataset} & \multirow{3}*{B} & \multicolumn{6}{|c|}{Non-adaptive} & \multicolumn{6}{|c|}{Adaptive}\\
\cline{3-14}
\multirow{3}*{} & \multirow{3}*{} & \multicolumn{3}{|c|}{$|X|=100$} & \multicolumn{3}{|c|}{$|X|=1000$} &  \multicolumn{3}{|c|}{$|X|=100$} & \multicolumn{3}{|c|}{$|X|=1000$}\\
\cline{3-14}
\multirow{3}*{} & \multirow{3}*{} & CD & 2CD & Increase & CD & 2CD & Increase & Ada & Ada+MGS & Increase & Ada & Ada+MGS & Increase\\
\hline
\multirow{3}*{Wiki-Vote} & 10 & 27 & 67 &40 &34 & 50 & 16 & 68 & 159 & 91 & 149 & 205 & 56\\
\cline{2-14}
\multirow{5}*{} &20 & 47 & 146 & 99 & 60 & 106 & 46 & 98 & 241 & 143 & 260 & 301 & 41\\
\cline{2-14}
\multirow{3}*{} &30 & 66 & 190 & 124 & 84 & 149 & 65 & 121 & 309 & 188 & 336 & 382 & 46\\
\cline{2-14}
\multirow{5}*{} &40 & 85 & 239 & 154 & 110 & 220 & 110 & 139 & 349 & 210 & 398 & 439 & 41\\
\cline{2-14}
\multirow{3}*{} &50 & 101 & 293 & 192 & 128 & 270 & 142 & 146 & 397 & 251 & 451 & 489 & 38\\
\hline
\multirow{3}*{Condmat} & 10 &67 & 411 & 344 & 83 & 135 & 52 & 206 & 665 & 459 & 453 & 652 & 199\\
\cline{2-14}
\multirow{5}*{} & 20 & 130 & 616 & 486 & 156 & 295 & 139 & 292 & 985 & 693 & 725 & 1039 & 314\\
\cline{2-14}
\multirow{3}*{} & 30 & 182 & 803 & 621 & 234 & 558 & 324 & 348 & 1272 & 924 & 935 & 1314 & 379\\
\cline{2-14}
\multirow{5}*{} & 40 & 231 & 934 & 703 & 306 & 744 & 438 & 375 & 1483 & 1108 & 1105 & 1570 & 465\\
\cline{2-14}
\multirow{3}*{} & 50 & 276 & 1132 & 856 & 377 & 762 & 385 & 375 & 1685 & 1310 & 1258 & 1774 & 516\\
\hline
\multirow{3}*{Dblp} & 10 & 69 & 358 & 289 & 75 & 251 & 176 & 188 & 859 & 671 & 412 & 1126 & 714\\
\cline{2-14}
\multirow{5}*{} & 20 & 131 & 802 & 671 & 146 & 680 & 534 & 275 & 1501 & 1226 & 657 & 2007 & 1350\\
\cline{2-14}
\multirow{3}*{} & 30 & 185 & 1019 & 798 & 213 & 996 & 783 & 332 & 1953 & 1621 & 850 & 2763 & 1913\\
\cline{2-14}
\multirow{5}*{} & 40 & 228 & 1308 & 1080 & 281 & 1157 & 876 & 360 & 2403 & 2043 & 1011 & 3417 & 2406\\
\cline{2-14}
\multirow{3}*{} & 50 & 268 & 1644 & 1376 & 346 & 1424 & 1328 & 360 & 2793 & 2433 & 1159 & 4057 & 2898\\
\hline
\multirow{3}*{Livejournal} & 10 & 138 & 1564 & 1426 & 237 & 1508 & 1271 & 408 & 4436 & 4028 & 980 & 24863 & 23883\\
\cline{2-14}
\multirow{5}*{} & 20 & 256 & 3193 & 2937 & 425 & 2523 & 2098 & 550 & 7199 & 6649 & 1603 & 30029 & 28426\\
\cline{2-14}
\multirow{3}*{} & 30 & 363 & 5340 & 4977 & 580 & 5389 & 4809 & 680 & 9523 & 8843 & 2089 & 34059 & 31970\\
\cline{2-14}
\multirow{5}*{} & 40 & 464 & 6935 & 6471 & 713 & 6101 & 5388 & 744 & 11310 & 10566 & 2465 & 38132 & 35667\\
\cline{2-14}
\multirow{3}*{} & 50 & 556 & 7464 & 6908 & 863 & 7261 & 6398 & 749 & 12781 & 12032 & 2788 & 41280 & 38492\\
\hline
\end{tabular}
\end{center}
\vspace{-0.2cm}
\end{table*}
\normalsize

  \subsubsection{Influence Spread}

    The expected influence spread of the non-adaptive case and the adaptive case is presented in Fig. \ref{Influence_spread_in_Non-ada} and Fig. \ref{Influence_spread_in_Ada} respectively.

  \begin{itemize}[leftmargin=*]
    \item \textbf{Non-adaptive Case}
  \end{itemize}

    From Fig. \ref{Influence_spread_in_Non-ada}, we can see that CD performs better than IM, since the discount allocation in CD is allowed to be fractional and thus more fine-grained. It is a little surprising that the simple two-stage algorithm RF shows larger influence spread than the elaborate CD in most settings, except two points in Wiki-Vote ($\alpha=0.6$, $B=10 \text{ and } 30$) due to its randomness. The reason is that RF has access to influential neighbors. As can be seen, the two-stage coordinate descent algorithm outperforms the other three algorithms in all the settings. The ratio between 2CD and the second best result varies from 1.3 to nearly 3. This result is easy to understand since 2CD not only has access to the influential neighbors but also makes refinements in both stages.

    Moreover, although the scale of the four datasets is quite different, we find that the influence spread of one-stage algorithms (i.e., CD and IM) is nearly in the same scale, while the influence spread of two-stage algorithms scales as the size of networks. We can infer that simply allocating discounts to initially accessible users restricts the spread of influence. Meanwhile, exploiting the friendship paradox helps expand the influence spread.

  \begin{itemize}[leftmargin=*]
    \item \textbf{Adaptive Case}
  \end{itemize}

    As can be seen from Fig. \ref{Influence_spread_in_Ada}, in most settings the one-stage adaptive algorithm Ada has the smallest influence spread, since only initially accessible users are seeded. In some settings, the LP algorithm is even worse than Ada, especially in small datasets and small budgets, since in LP the degree is directly regarded as the influence. This treatment helps with the complexity but losses accuracy, and causes possible blindness when selecting seeds. In A-Greedy, the influence is unbiasedly estimated by the hyper-graph. The result of A-Greedy is thus better than that of LP. Our proposed three algorithms achieve larger influence spread than the above three algorithms. The reason is three-fold. First, the discount is fractional and thus more fine-grained. Second, the influence is accurately estimated. Third, the FP phenomenon is leveraged. There is an evident gap between Ada+MGS and Ada+GS. This phenomenon conforms with the theoretical result that the approximation ratio of Ada+MGS is larger than Ada+GS. In the experiment, Ada+CD shows smaller influence spread than Ada+GS and Ada+MGS. However, we can not say that Ada+CD is definitely inferior to the other two algorithms, since its performance is closely related to the initial allocation. With a better initial allocation and more iterations, Ada+CD could deliver a better performance.

    Comparing our four proposed algorithms, we can see that adaptive algorithms yield larger influence spread than 2CD, except some minor points in the Wiki-Vote dataset. The reason is that adaptive algorithms make the most of the remaining budget by seeding the next user wisely based on the observation on the previous influence spread.

  \subsubsection{Scalability}

      The scalability of algorithms in the non-adaptive case and the adaptive case are reported in Fig. \ref{Running_Time_in_Non-ada} and Fig. \ref{Run_time_in_Ada} respectively. GBT is the building time of reversely reachable sets.

  \begin{itemize}[leftmargin=*]
    \item \textbf{Non-adaptive Case}
  \end{itemize}

    According to Fig. \ref{Running_Time_in_Non-ada}, the running time of IM and CD is almost the same as the GBT, since computing the allocation in 100 users will not take too much time. With the increase of the network size, the gap between 2CD and GBT decreases. In the smallest dataset Wiki-Vote, the running time of 2CD is about 20 times that of GBT. However, the ratio becomes less than 1.3 in the largest dataset soc-Livejournal. The reason is that the computation cost of the hyper-graph is high, while the execution of the 2CD algorithm is relatively efficient. It is worth noting that, the RF algorithm has the least running time and the best scalability, since RF does not need to build the hyper-graph. Thus, the simple two-stage algorithm RF outperforms IM and CD in both influence spread and scalability.

  \begin{itemize}[leftmargin=*]
    \item \textbf{Adaptive Case}
  \end{itemize}

    It is shown in Fig. \ref{Run_time_in_Ada} that the running time of the adaptive algorithms follows the sequence: Ada+CD $>$ Ada+MGS $>$ Ada+GS $>$ A-Greedy $>$ Ada. We next explain the sequence in an increasing order. The running time of Ada is the smallest since it does not need to consider the allocation in neighbors. A-Greedy computes the allocation in both stages, so the running time is higher than Ada. The running time of Ada+GS is larger than that of A-Greedy, since fine-grained discounts incur more sophisticated computation. It is not a surprise to see that Ada+MGS takes more time than Ada+GS, since enumeration is applied in Ada+MGS. As for Ada+CD, to estimate the benefit of each action, numerous coordinate descent algorithms are carried out in stage 2, incurring tremendous iterations. Thus, Ada+CD is less efficient than Ada+MGS. Similar to the non-adaptive case, the gap between algorithms and GBT is decreasing as the network size grows, due to the same reason in the non-adaptive case. The running time of LP is almost in the same scale over the four satasets, since it is not based on the hyper-graph and only needs to solve an LP problem.

    From the results in both cases, we find that most time is spent on building the hyper-graph. The execution time (GBT not included) of the proposed algorithms is less than one hour in all the cases except one point ($B=50$ in Fig. \ref{livejournal_Non-ada_alpha-Runtime=10}). Generally, the running time of 2CD is larger than Ada+MGS while smaller than Ada+CD. Thus, we can roughly obtain the sequence of running time of the proposed algorithms: Ada+CD $>$ 2CD $>$ Ada+MGS $>$ Ada+GS.

  \subsubsection{Sensitivity}

    We also test the sensitivity of our proposed algorithms with respect to different settings of seed probability functions. To this end, we introduce a second setting with different portion of seed probability functions. The previous setting in Section \ref{Experimental_setup} is denoted as Setting 1 (S1). In Setting 2 (S2), 65\% users are assigned with $p_u(c_u)=2c_u-c_u^2$, and 20\% users with $p_u(c_u)=c_u$, and 15\% users with $p_u(c_u)=c_u^2$. The algorithms are run again in Setting 2.

    Table \ref{table_sensitivity} reports the influence spread of the proposed algorithms with $\alpha=1.0$. As can be seen, the influence spread of the four algorithms decreases in setting 2. This phenomenon indicates that users are harder to satisfy under Setting 2. In terms of the ability to cope with the change of seed probability functions, 2CD shows even better performance than Ada+CD and comparable performance with Ada+GS. The possible reason is that 2CD makes refinement over a large amount of users in $N(S)$, while Ada+CD and Ada+GS only allocate discounts to neighbors of one agent each time. Thus, it is easier for 2CD to find an alternative user to seed when the seed probability function of a user becomes not favorable. However, Ada+MGS shows the best performance in coping with the change of settings. The explanation is that the enumeration process helps find good action combinations in Setting 2.

  \subsubsection{Impact of the Friendship Paradox Phenomenon}

    In this part, we evaluate the impact of the FP phenomenon. The algorithms are evaluated under two settings, i.e., $|X|=100 \text{ and } 1000$. To show the effect of the FP phenomenon, we compare the influence spread between one-stage algorithms and two-stage algorithms. In the non-adaptive case, CD and 2CD are selected for experiment. In the adaptive case, the only one-stage algorithm Ada is tested and Ada+MGS is selected due to its impressive performance.



    The influence spread is shown in Table \ref{table_FP}. Regardless of the size of $X$, two-stage algorithms show larger influence spread than one-stage algorithms by exploiting the influential neighborhood. We next focus on the influence spread of algorithms in the larger $|X|$. For the two one-stage algorithms, when $|X|=1000$, the influence spreads both become larger, since the number of influential users is likely to be larger in a larger $X$. However, the performance of two-stage algorithms is different. When $|X|=1000$, the influence spread of 2CD becomes smaller, while Ada+MGS shows even better performance. We further compare the increase of influence spread brought by two-stage algorithms. In the two relatively smaller datasets, i.e., Wiki-Vote and Ca-Condmat, the increases of 2CD and Ada+MGS both become smaller when $|X|=1000$. However, in Dblp and LiveJournal, the increase of Ada+MGS is larger when $|X|=1000$, while the increase of 2CD is still smaller. This observation indicates that Ada+MGS has a better ability to utilize the friendship paradox phenomenon.

\section{Conclusion and Future Work}\label{conclusion}
    This paper studies the influence maximization problem with limited initially accessible users. To overcome the access limitation, we propose a new two-stage seeding model with the FP phenomenon embedded, where neighbors are further seeded. Based on this model, we solve the limited influence maximization problem under both non-adaptive and adaptive cases. In the non-adaptive case, we examine the properties of this problem and establish a two-stage coordinate descent framework to determine the discount allocation in two stages. In the adaptive case, we first consider the discrete-continuous setting and design the adaptive greedy algorithm with theoretical guarantee. Then, in the discrete-discrete setting, the allocation in stage 2 is considered to be discrete. Accordingly, two algorithms are proposed based on greedy selection. Finally, extensive experiments are carried out on real-world datasets to evaluate the performance of the proposed algorithms.
        Moreover, our work is only a primary study into the two-stage IM, since the derivation of diffusion probabilities is not considered. While by serving as a subroutine, our study would continue to benefit the design of effective online algorithms which consider the learning of diffusion parameters.


    In the future, we would like to devote to finding better allocations than the convergent one to improve the coordinate descent algorithm. In the adaptive case, it is implicitly assumed that we have enough time to observe the whole diffusion process, which may take lots of time and thus impractical. It is worthwhile to study the problem when only part of the diffusion is observed in each round. In footnote \ref{footnote_fix_budget}, although we proposed solutions to the case of unified budget, it is still necessary to treat it as an individual problem and more effective algorithms could be designed.


%
\footnotesize
\bibliographystyle{IEEEtranN}
\bibliography{reference}
\normalsize
%
%

\begin{IEEEbiography}[{\includegraphics[width=1in,height=1.25in,clip,keepaspectratio]{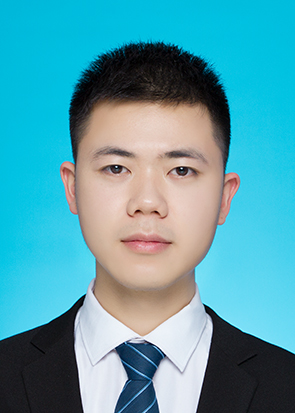}}]{Chen Feng}
received his B.E. degree in Communication Engineering from Tianjin University, China, in 2016. He is currently pursuing the Ph.D. degree in Electronic Engineering at Shanghai Jiao Tong University, Shanghai, China. His current research interests are in the area of social networks and information diffusion.
\end{IEEEbiography}


\begin{IEEEbiography}[{\includegraphics[width=1in,height=1.25in,clip,keepaspectratio]{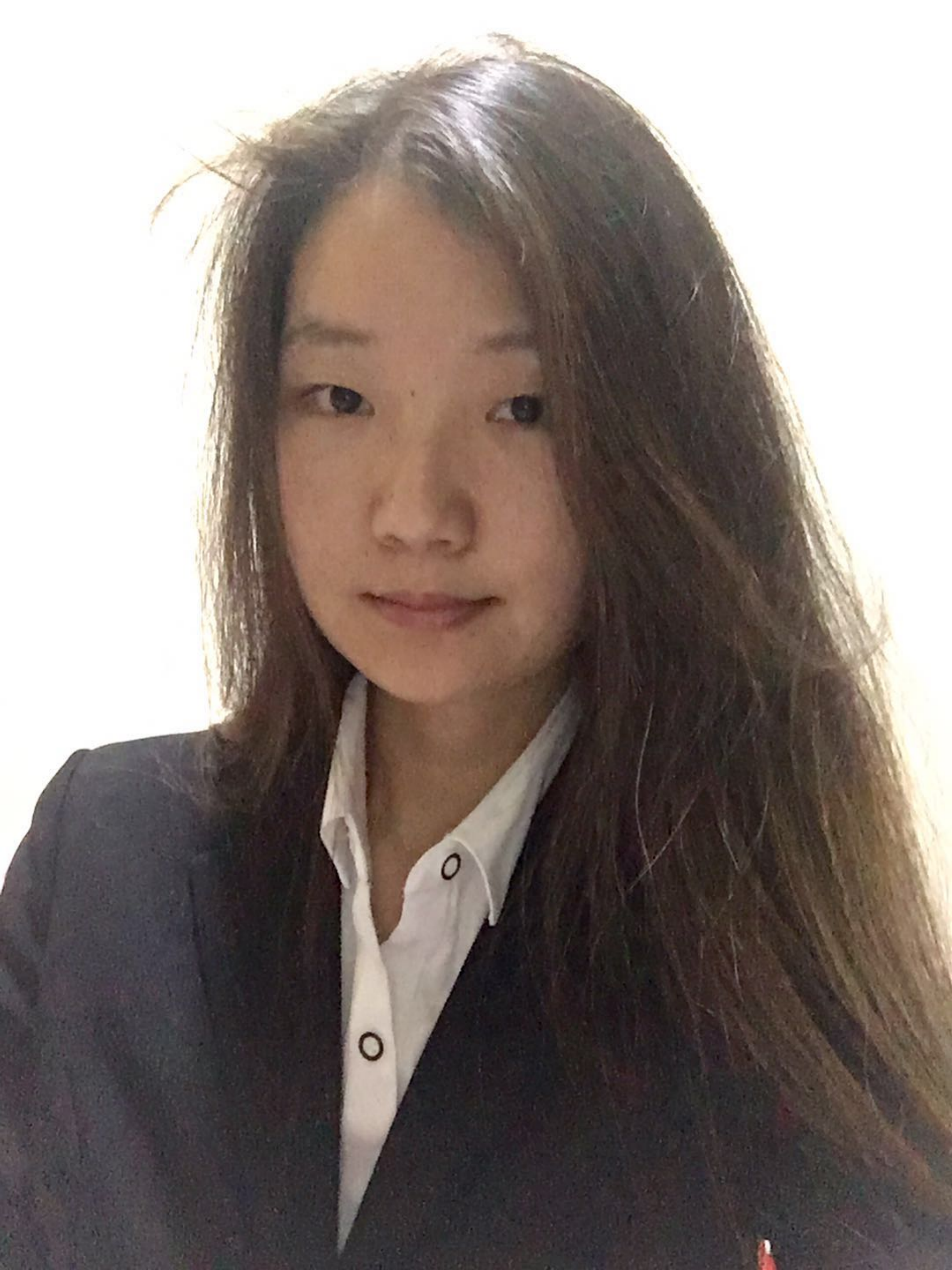}}]{Luoyi Fu}
received her B.E. degree in Electronic Engineering from Shanghai Jiao Tong University, China, in 2009 and Ph.D. degree in Computer Science and Engineering in the same university in 2015. She is currently an Associate Professor in Department of Computer Science and Engineering in Shanghai Jiao Tong University. Her research of interests are in the area of social networking and big data, scaling laws analysis in wireless networks, connectivity analysis and random graphs.
\end{IEEEbiography}

\begin{IEEEbiography}[{\includegraphics[width=1in,height=1.25in,clip,keepaspectratio]{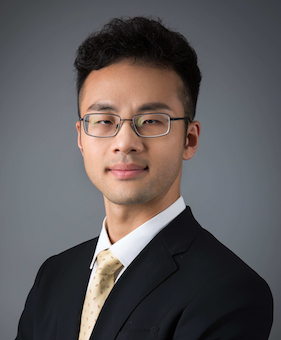}}]{Bo Jiang}
received the B.S. degree in Electronic Engineering from Tsinghua University, Beijing, China in 2006, the M.S. degree in Electrical and Computer Engineering from the University of Massachusetts Amherst in 2008 and the Ph.D. degree in Computer Science from the same university. He is currently an Associate Professor in the John Hopcroft Center in Shanghai Jiao Tong University. His research includes network modeling and analysis, wireless networks, and network science.
\end{IEEEbiography}

\begin{IEEEbiography}[{\includegraphics[width=1in,height=1.25in,clip,keepaspectratio]{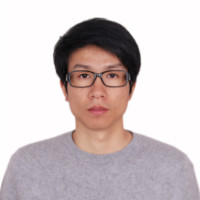}}]{Haisong Zhang}
received his B.E. degree in Computer Science and Technology from Xidian University, China, in 2011 and the M.E. Degree in Software Engineering from the same university in 2014. He is currently a Senior Researcher in Tencent AI Lab. His research interests include natural language understanding and dialogue system etc..
\end{IEEEbiography}


\begin{IEEEbiography}[{\includegraphics[width=1in,height=1.25in,clip,keepaspectratio]{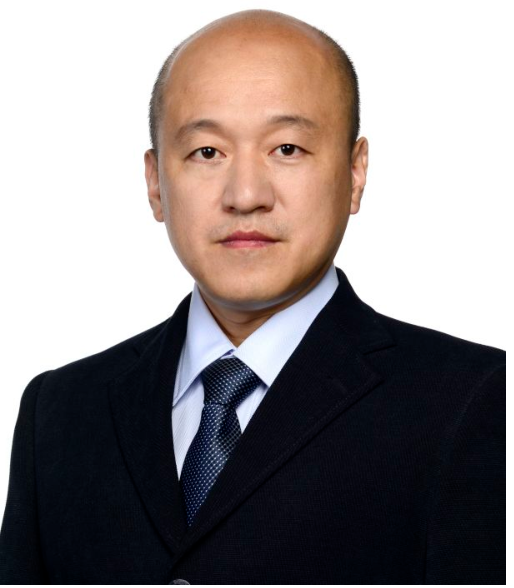}}]{Xinbing Wang}
 received the B.S. degree (with hons.) in automation from Shanghai Jiao Tong University, China, in 1998, the M.S. degree in computer science and technology from Tsinghua University, China, in 2001, and the Ph.D. degree with a major in electrical and computer engineering and minor in mathematics from North Carolina State University, in 2006. Currently, he is a Professor in the Department of Electronic Engineering, and Department of Computer Science, Shanghai Jiao Tong University, China. Dr. Wang has been an Associate Editor for IEEE/ACM TRANSACTIONS ON NETWORKING, IEEE TRANSACTIONS ON MOBILE COMPUTING, and ACM Transactions on Sensor Networks. He has also been the Technical Program Committees of several conferences including ACM MobiCom 2012,2014, ACM MobiHoc 2012-2017, IEEE INFOCOM 2009-2017.
\end{IEEEbiography}

\begin{IEEEbiography}[{\includegraphics[width=1in,height=1.25in,clip,keepaspectratio]{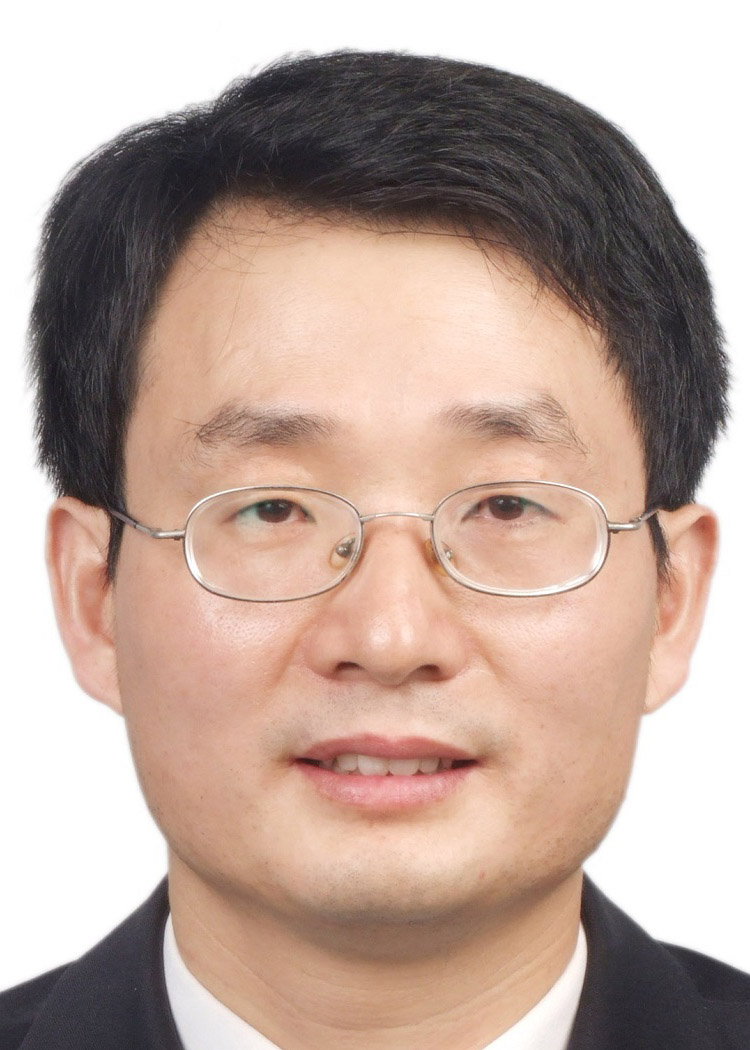}}]{Feilong Tang}
received the Ph.D. degree in computer science and technology from Shanghai Jiao Tong University, in 2005. He was a Japan Society for the Promotion of Science (JSPS) Postdoctoral Research Fellow. Currently, he works with the School of Software, Shanghai Jiao Tong University. His research interests include cognitive and sensor networks, protocol design for communication networks, and pervasive and cloud computing. He has published more than 100 papers in journals and international conferences and works as a PI of many projects such as National Natural Science Foundation of China (NSFC) and National High-Tech R\&D Program (863 Program) of China. 
\end{IEEEbiography}

\begin{IEEEbiography}[{\includegraphics[width=1in,height=1.25in,clip,keepaspectratio]{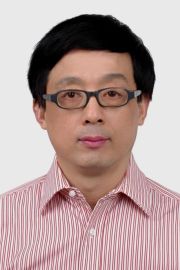}}]{Guihai Chen}
received the B.S. degree from Nanjing University, the M.E. degree from Southeast University, and the Ph.D. degree from The University of Hong Kong. He visited the Kyushu Institute of Technology, in 1998, as a Research Fellow, and the University of Queensland, in 2000, as a Visiting Professor. From 2001 to 2003, he was a Visiting Professor with Wayne State University. He is currently a Distinguished Professor and a Deputy Chair with the Department of Computer Science, Shanghai Jiao Tong University. His research interests include wireless sensor networks, peer-to-peer computing, and performance evaluation. He has served on technical program committees of numerous international conferences.
\end{IEEEbiography}

\newpage

\section*{\Large{Supplimentary Material}}
Due to space limitation, some illustrations, proofs and experimental results are omitted in the main paper, and we provide them in this supplemental file for completeness.

\setcounter{section}{0}

\section{Seeding Examples}
  In this section, for readers' comprehension, we present two seeding examples to illustrate the non-adaptive seeding process and the adaptive process respectively. For ease of illustration, we assume $B_1=B_2=1$ and the discount rate in the adaptive case is $D=\{0.5,1.0\}$.

\subsection{Non-adaptive Seeding Example}
\begin{figure}[h]
\centering
\vspace{-0.3cm}
\includegraphics[width=0.5\textwidth]{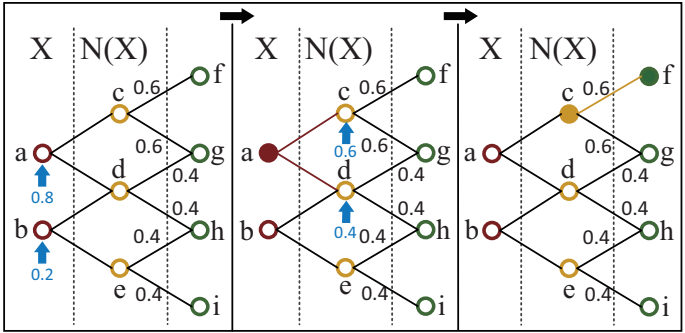}\vspace{-0.3cm}
\caption{An example of the non-adaptive seeding process.}\vspace{-0.3cm}\label{Fig_non-ada_seeding_eg}
\end{figure}

  In the recruitment stage, user $a$ seems to be more profitable, since $a$ is able to reach the influential user $c$. Thus, we allocate discount 0.8 to user $a$ and 0.2 to user $b$ in stage 1. Budget $B_1$ is used up. Suppose user $a$ becomes an agent while $b$ does not. Then, we reach users $c$ and $d$ via user $a$. In the trigger stage, we distribute budget $B_2=1$ to newly reachable users by allocating discount 0.6 to user $c$ and 0.4 to user $d$. User $c$ becomes the seed and then influence diffusion starts from it. Finally, user $f$ gets influenced.

\subsection{Adaptive Seeding Example}

\begin{figure}[h]
\centering
\vspace{-0.3cm}
\includegraphics[width=0.5\textwidth]{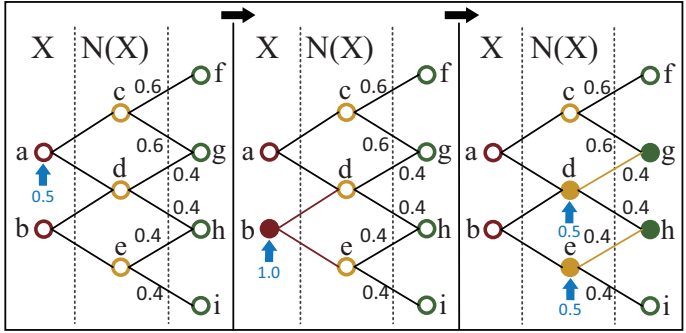}\vspace{-0.3cm}
\caption{An example of the adaptive seeding process.}\vspace{-0.3cm}\label{Fig_ada_seeding_eg}
\end{figure}

  In the adaptive case, initially accessible users are seeded sequentially. In the first round, we adopt the action $(a,0.5)$, i.e., seeding user $a$ with discount 0.5. Unfortunately, $a$ refuses the discount. Then, we move to the second round, where the remaining budget $B_1$ is still 1. We adopt the action $(b,1.0)$ and user $b$ accepts the discount. Then, $b$'s neighbors $d$ and $e$ become reachable. Both $d$ and $e$ are provided with discount 0.5 from budget $B_2$. Suppose $d$ and $e$ both become the seed and the influence diffusion starts from them. Finally, users $g$ and $h$ get influenced.

\section{Missing Proofs}
  In this section, we present the omitted proofs for Theorems and Lemmas.

\subsection{Proof of Lemma 1}
We prove lemma 1 by the probabilistic method. Let $d(\cdot)$ denote the degree of a node and $\bar{d}(\cdot)$ denote the average degree of nodes in a set. The expected average degree of nodes in $X$ is
     \begin{equation*}
       E[\bar{d}(X)]=E\left[\frac{\sum_{v\in X}d(v)}{|X|}\right]=\frac{1}{|X|}\sum_{v\in X}E[d(v)]=E[d(v)].
     \end{equation*}

     Recall that nodes in $V$ are randomly selected into $X$ with probability $p\rightarrow 0$. For node $v$ in $V$, the probability that it is in $N(X)$ is $(1-p)[1-(1-p)^{d(v)}]$, i.e., $v$ is not in $X$, but at least one of its neighbors is in $X$. Since $p\rightarrow 0$, the probability can be approximated by $p(1-p)d(v)$. Thus, the size of the neighborhood of $X$ can be denoted as $|N(X)|=\sum_{v\in V}p(1-p)d(v)$. The sum of degrees of nodes in $N(X)$ is $\sum_{v\in V}p(1-p)d^2(v)$. The expected average degree of nodes in $N(X)$ is
     \begin{equation*}
       E[\bar{d}(N(X))]=\!E\!\left[\frac{\sum_{v\in V}p(1-p)d^2(v)}{\sum_{v\in V}p(1-p)d(v)}\right]\!\!=E\!\!\left[\frac{\sum_{v\in V}d^2(v)}{\sum_{v\in V}d(v)}\right]
     \end{equation*}
     Applying Cauchy-Schwartz inequality, we have $\sum_{v\in V}d^2(v)\geq \frac{1}{n}[\sum_{v\in V}d(v)]^2$. Thus,
     $E[\bar{d}(N(X))]\geq E\left[\frac{1}{n}\sum_{v\in V}d(v)\right]=E[d(v)]$.
     Therefore, $E[\bar{d}(N(X))]\geq E[\bar{d}(X)]$. This completes the proof.

\subsection{Proof of Lemma 2}
  We first show the NP-hardness of the discrete two-stage non-adaptive influence maximization by reduction from the NP-complete \textit{Set Cover} problem. Consider an arbitrary instance of Set Cover. Let $A=\{A_1, A_2, \cdots, A_m\}$ be a family of subsets of a ground set $U=\{u_1,u_2, \cdots, u_n\}$, satisfying $\bigcup_{i=1}^{m}A_i=U$. The problem is whether there exists a subfamily $C\subseteq A$, $|C|=k$, whose union is $U$.

  (1) We proceed to construct a discrete non-adaptive influence maximization problem corresponding to the Set Cover problem. Define a graph $G$ where there is only one node $x$ in $X$. For each subset $A_i$, there is a node $i$ corresponding to it, with a directed edge from $x$ to $i$. For each element $u_j$, there is a node $j$ corresponding to it. If $u_j\in A_i$, then there is an edge between $i$ and $j$. All edges are associated with probability 1. The budget is $k+1$. In the two-stage setting, one budget must be spent on initially reachable user $x$ to reach its neighbors. Thus, the Set Cover problem is equivalent to selecting $k$ nodes in $N(x)$ to influence the $n$ nodes corresponding to $U$.

  (2) If there is a solution to the discrete maximization problem, then we can solve the Set Cover problem by selecting corresponding $k$ subsets in $A$ to cover the $n$ elements in $U$.

  (3) It is obvious that the reduction from Set Cover can be completed in polynomial time by traversing $A$ and $U$.

  Thus, we prove that the discrete case is NP-hard. We can show that our continuous non-adaptive influence maximization is NP-hard by reduction from the discrete case, which is trivial and thus omitted. Intuitively, the discrete IM is only a special case of our continuous setting.

\subsection{Proof of Lemma 3}
  If $T_1=T_2$, Lemma 3 holds definitely. We focus on $T_1\subsetneqq T_2$.

  To begin with, we prove a simple case where $T_2=T_1\cup \{v\}$. Assume the optimal discount allocation in $T_1$ is $C_2^{*}$. Since all users in $T_1$ are also in $T_2$ and the budget is the same, $C_2^{*}$ is a feasible allocation in $T_2$ where the discount of user $v$ is 0 and the discounts of other users are the same as $C_2^{*}$. Then, we have $Q(C_2^*; T_2)=\max Q(C_2; T_1)$. Due to the optimality of $\max Q(C_2; T_2)$, we have $\max Q(C_2; T_2)\geq Q(C_2^*; T_2)= \max Q(C_2; T_1)$, i.e. $\max Q(C_2; T_1\cup \{v\})\geq \max Q(C_2; T_1)$. By the transitivity of $\geq$, we know that $\max Q(C_2; T_2)\geq \max Q(C_2; T_1)$ holds for all $T_1\subsetneqq T_2$. Thus, the proof of Lemma 3 is completed.

\subsection{Proof of Theorem 1}
(1) We start from a simple case, where only one element is different between $C_2$ and $C_2^\prime$. Then, the result can be extended to general cases by the transitivity of $\geq$.

Assume that only the discount of user $u\in N(S)$ is different. $u$ gets discount $c_u$ in $C_2$ while $c_u^\prime$ in $C_2^\prime$ and $c_u\geq c_u^\prime$. By the definition of $Q(C_2; N(S))$, it can be written as
\begin{equation*}
\begin{split}
Q(C_2; N(S))&=\sum_{T \subseteq N(S)\setminus \{u\}}\!\!\!\!\!\!P_{r}(T;C_2,N(S)\!\!\setminus \!\!\{u\})\cdot\\
           &\Big\{[1-p_{u}(c_u)]I(T)+p_{u}(c_u)I(T\cup \{u\}) \Big\}.
\end{split}
\end{equation*}
After a simple calculation, we have
\begin{align*}
\begin{split}
Q(C_2; N(S)&)-Q(C_2^\prime; N(S))=\\
&\!\!\!\!\!\!\sum_{T \subseteq N(S)\setminus \{u\}}\!\!\!\!\!\!P_{r}(T;C_2,N(S)\!\!\setminus \!\!\{u\})\cdot\\
                              &\Big \{[p_{u}(c_u)-p_{u}(c_u^\prime)][I(T\cup \{u\})-I(T)\Big \}.
\end{split}
\end{align*}
Since $I(\cdot)$ is nondecreasing under the independent cascade model and $p_{u}(\cdot)$ is nondecreasing as well by assumption, we have $Q(C_2; N(S))\geq Q(C_2^\prime; N(S))$.

(2) Following similar technique, we prove the result when there is only one element different between $C_1$ and $C_1^\prime$. Assume the different element is $u$ and $u$ gets discount $c_u$ in $C_1$ while $c_u^\prime$ in $C_1^\prime$ with $c_u\geq c_u^\prime$. By the definition of $f(C_1;X)$, $f(C_1;X)$ can be written as
\begin{equation*}
\begin{split}
f(C_1;X)&=\sum_{S \subseteq X\setminus \{u\}}P_{r}(S;C_1,X\setminus \{u\})\\
           &\Big\{p_{u}(c_u)\max Q(C_2; N(S\cup \{u\}))\\
            &+ [1-p_{u}(c_u)]\max Q(C_2; N(S))\Big\}.
\end{split}
\end{equation*}
After a simple calculation, we have
\begin{equation*}
\begin{split}
f(C_1;X)-&f(C_1^\prime;X)=\sum_{S \subseteq X\setminus \{u\}}P_{r}(S;C_1,X\setminus \{u\})\\
                               &\Big \{ [p_{u}(c_u)-p_{u}(c_u^\prime)][\max Q(C_2; N(S\cup \{u\}))\\
                               &-\max Q(C_2; N(S))]\Big \}.
\end{split}
\end{equation*}
   Note that the budget in stage 2 is the same under $C_1$ and $C_1^\prime$. According to Lemma 3, $\max Q(C_2; N(S\cup \{u\}))-\max Q(C_2; N(S))\geq 0$. According to the monotonicity of $p_u(\cdot)$ and $c_u\geq c_u^\prime$, we have $f(C_1;X)\geq f(C_1^\prime;X)$.

Combining case (1) and case (2), we complete the proof.

\subsection{Proof of Lemma 4}
  It is easy to verify the adaptive monotonicity since under any realization, since the influence spread will not decrease when more users are seeded.

  To prove the adaptive submodularity, according to its definition, it is equivalent to prove that $\Delta(y|\psi)\geq \Delta(y|\psi^\prime)$ holds for any $\psi \subseteq \psi^\prime$ and $y \in Y\setminus dom(\psi)$. Note that $\psi$ is a process with sequence recording the actions adopted. $\psi \subseteq \psi^\prime$ means $\psi$ is a subprocess of $\psi^\prime$, i.e., $\psi$ is a history of $\psi^\prime$ and $\psi^\prime$ went through all what $\psi$ has experienced. Therefore, the seeding result of nodes and states of edges observed by $\psi$ are the same in $\psi^\prime$.

  We first introduce some notations. The diffusion realization $\phi$ (resp. $\phi^\prime$) is a function of the states of edges, which are denoted as a series of random variables $X=\{X_{ij}, (i,j)\in E\}$ (resp. $X^\prime=\{X_{ij}^\prime, (i,j)\in E\}$).




  We attempt to define a coupled distribution $\rho((\lambda,\phi),(\lambda^\prime,\phi))$ over two pairs of realizations $(\lambda,\phi)\sim \psi$ and $(\lambda^\prime,\phi^\prime)\sim \psi^\prime$. Recall the definition of $(\lambda,\phi)\sim \psi$ that $\lambda$ is consistent with the partial seeding realization observed by $\psi$ and $\phi$ is consistent with the states of edges explored under $\psi$. Since $(\lambda,\phi)\sim \psi$, $(\lambda^\prime,\phi^\prime)\sim \psi^\prime$, and $\psi \subseteq \psi^\prime$, the states of nodes and edges observed by $\psi$ are the same in $(\lambda,\phi)$ and $(\lambda^\prime,\phi^\prime)$. Then, the diffusion brought by action $y$ (i.e., $\Delta(y|\psi)$) is only dependent on the states of unknown edges. Thus, we will reduce the domain of $\rho$ to $\phi$ and $\phi^\prime$. We define the coupled distribution $\rho$ in terms of a joint distribution $\hat{\rho}$ on $X\times X^\prime$, where $\phi=\phi(X)$ and $\phi^\prime=\phi^\prime(X^\prime)$ are the diffusion realizations induced by the two distinct sets of random edge states respectively. Recall that the domain of $\rho$ is reduced to $\phi$ and $\phi^\prime$. Hence, $\rho((\lambda,\phi(X)),(\lambda^\prime,\phi(X^\prime)))=\hat{\rho}(X,X^\prime)$.

  We say the seeding process $\psi$ observes an edge if it is explored and the state is revealed. For any edge $(i,j)$ observed by $\psi$ (resp. $\psi^\prime$), its state $X_{ij}$ (resp. $X_{ij}^\prime$) is deterministic. Recall that the states of edges observed by $\psi$ are the same in $\phi$ and $\phi^\prime$, since $\psi \subseteq \psi^\prime$. We will construct $\hat{\rho}$ so that the states of all edges unobserved by both $\psi$ and $\psi^\prime$ are the same in $X$ and $X^\prime$, i.e., $X_{ij}=X_{ij}^\prime$, otherwise $\hat{\rho}(X, X^\prime)=0$. The above constraints allow us to select $X_{ij}$ whose edges are unobserved by $\psi$. We select such variables independently. Hence for all $(X, X^\prime)$ satisfying the above constraints, we have

\begin{align*}
\hat{\rho}(X, X^\prime)=\!\!\prod\limits_{(i,j) \text{ unobserved by } \psi}\!\!p_{ij}^{X_{ij}}(1-p_{ij})^{1-X_{ij}},
\end{align*}
  otherwise $\hat{\rho}(X, X^\prime)=0$.

  We next try to prove that the following formula holds for any $((\lambda,\phi),(\lambda^\prime,\phi^\prime))\in support(\rho)$,
\begin{align}\label{sta1}
&\hat{\sigma}(dom(\psi^\prime)\cup \{y\},(\lambda^\prime,\phi^\prime))-\hat{\sigma}(dom(\psi^\prime),(\lambda^\prime,\phi^\prime))\leq \\
&\hat{\sigma}(dom(\psi)\cup \{y\},(\lambda,\phi))-\hat{\sigma}(dom(\psi),(\lambda,\phi)).\nonumber
\end{align}

  Let set $B$ denote $\sigma(dom(\psi),(\lambda,\phi))$, $D$ denote $\sigma(dom(\psi)\cup \{y\},(\lambda,\phi))$, $B^\prime$ denote $\sigma(dom(\psi^\prime),(\lambda^\prime,\phi^\prime))$ and $D^\prime$ denote $\sigma(dom(\psi^\prime)\cup \{y\},(\lambda^\prime,\phi^\prime))$. We will first show that $B\subseteq B^\prime$. For any node $i\in B$, there exists a path from some node $j\in dom(\psi)$ to it. Each edge in this path is observed to be live. Since $\psi \subseteq \psi^\prime$ and $(\lambda,\phi)\sim \psi$, $(\lambda^\prime,\phi^\prime)\sim \psi^\prime$, the edge observed to be live in $\psi$ must be live as well in $\psi^\prime$, and $j$ must also be a seed in $\psi^\prime$. Then, there is also a path from $i$ to $j$ under $(\lambda^\prime,\phi^\prime)$. Thus, $B\subseteq B^\prime$.

  We proceed to prove formula (\ref{sta1}). Since $\psi \subseteq \psi^\prime$, we have $dom(\psi)\subseteq dom(\psi^\prime)$ and thus $N(\!\!\!\bigcup\limits_{p\in dom(\psi)}\!\!\!{v(p)})\subseteq N(\!\!\!\bigcup\limits_{p\in dom(\psi^\prime)}\!\!\!v(p))$. Therefore, nodes newly reached by $v(y)$ under $\psi$ are part of those under $\psi^\prime$, that is, $N(v(y))\setminus N(\!\!\!\bigcup\limits_{p\in dom(\psi^\prime)}\!\!\!v(p))\subseteq N(v(y))\setminus N(\!\!\!\bigcup\limits_{p\in dom(\psi)}\!\!\!v(p))$. Note that the budget allocated to newly reached nodes under $\psi$ and $\psi^\prime$ is the same, because the budget drawn from $B_2$ only depends on the intrinsic property of $v(y)$ itself. Furthermore,
 $((\lambda,\phi),(\lambda^\prime,\phi^\prime))\in support(\rho)$, the states of unobserved edges are the same. Thus, according to Lemma 3 and $B\subseteq B^\prime$, we can see that $D\setminus B$ is a superset of $D^\prime\setminus B^\prime$. In addition, $\hat{\sigma}=|\sigma|$, $B\subseteq D$ and $B^\prime \subseteq D^\prime$, hence formula (\ref{sta1}) holds.

  For $((\lambda,\phi),(\lambda^\prime,\phi^\prime))\notin support(\rho)$, $\rho((\lambda,\phi),(\lambda^\prime,\phi^\prime))=0$. Then, summing over $((\lambda,\phi),(\lambda^\prime,\phi^\prime))$ in $support(\rho)$ and not in $support(\rho)$, we have
\begin{equation}\label{expectation}
\begin{aligned}
&\sum\limits_{((\lambda,\phi),(\lambda^\prime,\phi^\prime))}&&\!\!\!\!\!\!\!\!\!\rho((\lambda,\phi),(\lambda^\prime,\phi^\prime))(\hat{\sigma}(dom(\psi^\prime)\cup \{y\},(\lambda^\prime,\phi^\prime))\\
&  &&\!\!\!\!\!\!\!\!\!-\hat{\sigma}(dom(\psi^\prime),(\lambda^\prime,\phi^\prime)))\leq \\
&\sum\limits_{((\lambda,\phi),(\lambda^\prime,\phi^\prime))}&&\!\!\!\!\!\!\!\!\!\rho((\lambda,\phi),(\lambda^\prime,\phi^\prime))(\hat{\sigma}(dom(\psi)\cup \{y\},(\lambda,\phi))\\
&  &&\!\!\!\!\!\!\!\!\!-\hat{\sigma}(dom(\psi),(\lambda,\phi))).
\end{aligned}
\end{equation}
  Note that $p((\lambda, \phi)|\psi)=\sum\limits_{(\lambda^\prime,\phi^\prime)}\rho((\lambda,\phi),(\lambda^\prime,\phi^\prime))$ and $p((\lambda^\prime, \phi^\prime)|\psi^\prime)=\sum\limits_{(\lambda,\phi)}\rho((\lambda,\phi),(\lambda^\prime,\phi^\prime))$. We sum over $(\lambda,\phi)$ in the left side of formula (\ref{expectation}) and $(\lambda^\prime,\phi^\prime)$ in the right side. Combining the definition of $\Delta(y|\psi)$ and $\Delta(y|\psi^\prime)$, we have $\Delta(y|\psi)\geq \Delta(y|\psi^\prime)$.

\subsection{Proof of Lemma 5}
We would like to prove this lemma by induction. Let $S_m$ denote the first $m$ seeds selected by $\pi^{\text{greedy}}$ and $R_m$ denote the first $m$ seeds selected by $\pi_{\text{relaxed}}^{\text{greedy}}$.

(i) Let us consider the basic case, $m=1$, i.e. the first seed. It is easy to see that $\pi_{\text{relaxed}}^{\text{greedy}}$ will choose user $u$ that maximizes $\frac{ E[\hat{\sigma}(y,(\Lambda, \Phi))] }{ d_{\min}(u) }$, where $y=(u,d_{\min}(u))$.

As for $\pi^{\text{greedy}}$, we probe
$y^*=\argmax_{y\in Y}\frac{\Delta(y|\psi_p)}{d(y)}=\argmax_{y\in Y}\frac{ E[\hat{\sigma}(y,(\Lambda, \Phi))] }{d(y)}$.
Each time, the selected action is refused or accepted. We accordingly delete the action and move to the next round or get a seed.

 The action space $Y$ can be divided into a union of disjoint action subsets $Y:=\bigcup\limits_{v\in X}Y_v$, where $Y_v$ is the set of actions about user $v$, i.e., $Y_v=\{y|v(y)=v\}$. We next show that for each action subset $Y_v$, there is no need to consider actions whose discounts are not $d_{\min}(v)$. Due to the greedy policy $\pi^{\text{greedy}}$, user $v$ will be probed with $Y_v$ from the smallest discount. For actions with discount less than $d_{\min}(v)$, $v$ will reject it. When the discount becomes $d_{\min}(v)$, $v$ becomes a seed and remaining polices are abandoned. Thus, it is equivalent to select actions from $Y^*=\{(v,d_{\min}(v))|v\in X\}$. Then, the selection becomes selecting an action from $Y^*$ that maximizes $\frac{ E[\hat{\sigma}(y,(\Lambda, \Phi))] }{d_{\min}(v(y))}$, which is the same as $\pi_{\text{relaxed}}^{\text{greedy}}$. Thus, the two algorithms will yield the same first seed.

(ii) Assume that $S_m$ and $R_m$ are the same when $m=k$, we proceed to the case $m=k+1$. Given partial seeding process $\psi_p$, the seed selected by $\pi_{\text{relaxed}}^{\text{greedy}}$ is the user $u$ that maximizes $\frac{\Delta((u,d_{\min}(u))|\psi_p)}{d_{\min}(u)}$. In terms of $\pi^{\text{greedy}}$, for each user $u$, $\Delta(y|\psi_p)$ is the same for any action $y$ about $u$, i.e. $v(y)=u$, since $u$ is assumed to be the seed when calculating $\Delta(y|\psi_p)$ and the budget allocated in stage 2 only depends on user $u$ itself. Following similar arguments in (i), we derive the $(k+1)$-th seed selected by maximizing $\frac{ E[\hat{\sigma}(y,(\Lambda, \Phi))] }{d_{\min}(v(y))}$. Thus, the $(k+1)$-th seed is the same in two algorithms.

Combining (i) and (ii), we complete the proof of Lemma 5.

\subsection{Proof of Theorem 2}
We next analyze the budget used in stage 1 in the relaxed setting. Let us consider the selection of the last action, if the remaining budget is smaller than the minimum desired discounts $d_v$ of all the remaining users, then no one will accept the discount and the remaining budget can not be used. In the extreme case, all the remaining users desire discount 1. Thus, the budget used in stage 1 is at least $B_1-1$. As noted in Section 5, with proper initial allocation and enough iterations, the local optimum of the coordinate descent algorithm in stage 2 could be eliminated. Let $\pi_{\text{relaxed}}^{\text{greedy}}$ denote the optimal action of the relaxed setting. According to Theorem A.10 in [12], since $\hat{\sigma}$ is adaptive submodular and optimal value is obtained in stage 2, we have
\begin{center}
$\hat{\sigma}(\pi_{\text{relaxed}}^{\text{greedy}})\geq (1-e^{-\frac{B_1-1}{B_1}})\hat{\sigma}(\pi_{\text{relaxed}}^{\text{OPT}})$.
\end{center}
By the definition of $\hat{\sigma}(\pi_{\text{relaxed}}^{\text{OPT}})$, we can see that
\begin{center}
$\hat{\sigma}(\pi_{\text{relaxed}}^{\text{OPT}})\geq \hat{\sigma}(\pi^{\text{OPT}})$.
\end{center}
 Moreover, from Lemma 5, we have
\begin{center}
 $\hat{\sigma}(\pi_{\text{relaxed}}^{\text{greedy}})=\hat{\sigma}(\pi^{\text{greedy}})$.
\end{center}
  Thus,
\begin{center}
 $\hat{\sigma}(\pi^{\text{greedy}})\geq (1-e^{-\frac{B_1-1}{B_1}})\hat{\sigma}(\pi^{\text{OPT}})$.
\end{center}

\subsection{Proof of Lemma 6}
It has been proven in [1] that the classical influence maximization under independent cascade model is NP-hard. We would like to show that the classical IM can be reduced to our action selection problem. Given an arbitrary instance of the classical IM problem with budget $k$, the goal is to influence the whole network by initially selecting $k$ nodes. This is only a special case of our problem, where there is only one discount rate $D=\{1\}$ and the action space is $Z=R\times D$. If $L^*\subseteq Z$ is the optimal solution of this action selection problem. Then, we can solve the classical influence maximization problem by selecting corresponding nodes in $L^*$. It is easy to see that the reduction can be performed within polynomial time. Based on the above analysis, we prove the NP-hardness of the optimal action selection in stage 2.

\subsection{Proof of Lemma 7}
Suppose the newly reached users are $R$, and accordingly the action space is $Z=R\times D$. We first prove the submodularity of $Q(L; R|(\lambda, \phi))$, under some realization $(\lambda, \phi)$. $\forall E_1, E_2\subseteq Z$ and $E_1\subseteq E_2$, $\forall z\in Z\setminus E_2$, we aim to prove that
\begin{align}\label{sta2}
&Q(E_1\cup\{z\};R|(\lambda, \phi))-Q(E_1;R|(\lambda, \phi))\geq \\
&Q(E_2\cup\{z\};R|(\lambda, \phi))-Q(E_2;R|(\lambda, \phi)).\nonumber
\end{align}
Under the same seeding and diffusion realization $(\lambda, \phi)$, $E_1\subseteq E_2$ implies that the users influenced by actions $E_2$ are also influenced by actions $E_1$, but there exist users influenced by $E_2$ but not influenced by $E_1$. Therefore, the set of users influenced by $z$ but not influenced by $E_1$ is a superset of that under $E_2$. This conclusion directly leads to formula (\ref{sta2}). Since the left side of formula (\ref{sta2}) represents the number of users influenced by $z$ but not by $E_1$, and the right side represents the number of users influenced by $z$ but not by $E_2$, Thus formula (\ref{sta2}) holds. Since a non-negative combination of submodular functions is still submodular, we derive that $Q(L;R)$ is submodular.

We next prove that under some realization $(\lambda, \phi)$, $Q(L;R|(\lambda, \phi))$ is monotone nondecreasing. $\forall E_1\subseteq E_2\subseteq Z$, users who become seeds under $E_1$ must be seeds under $E_2$. However, users seeded by $E_2\setminus E_1$ may become seeds. The states of edges are the same for $E_1$ and $E_2$ under the same realization $\phi$. Thus, $E_2$ can achieve at least the same diffusion as $E_1$, i.e. $Q(E_2;R|(\lambda, \phi))\geq Q(E_1;R|(\lambda, \phi))$. The non-negative combination of monotone functions are still monotone. Hence, $Q(E_2;R)\geq Q(E_1;R)$.

\subsection{Proof of Lemma 8}
Under the discrete-discrete setting, the seeding process in stage 1 is the same as the discrete-continuous setting, while the seeding process in stage 2 aims at choosing an optimal subset of actions from $Z$. Therefore, we only need to modify the proof of Lemma 4 about stage 2. Then, we attempt to prove the inequality $|D^\prime|-|B^\prime|\leq |D|-|B|$. Following the same argument in the proof of Lemma 4, we have $B\subseteq B^\prime$. Suppose the actions adopted in $v(y)$'s neighbors under $\psi^\prime$ are $L^\prime$. Since $\psi \subseteq \psi^\prime$, by the proof of Lemma 4, we have
\begin{equation}\label{newly_reached_users_contain}
N(v(y))\setminus N(\!\!\!\bigcup\limits_{p\in dom(\psi^\prime)}\!\!\!v(p))\subseteq N(v(y))\setminus N(\!\!\!\bigcup\limits_{p\in dom(\psi)}\!\!\!v(p)).
\end{equation}
Moreover, the budget brought by $v(y)$ in stage 2 is the same under $\psi$ and $\psi^\prime$. Thus, we can carry out the same set of actions $L^\prime$ under $\psi$, triggering the same diffusion from newly reached users. By formula (\ref{newly_reached_users_contain}), we see that $v(y)$ reaches more new users under $\psi$ than $\psi^\prime$. Thus, the action space under $\psi^\prime$ is contained in that of $\psi$, which allows the possibility to achieve a larger diffusion under $\psi$. Recall $B\subseteq B^\prime$, we have $D^\prime\setminus B^\prime\subseteq D\setminus B$. Moreover, $\hat{\sigma}=|\sigma|$, $B\subseteq D$ and $B^\prime \subseteq D^\prime$. Thus, we obtain $|D^\prime|-|B^\prime|\leq |D|-|B|$. The rest of the proof is the same as the proof of Lemma 4.

\subsection{Proof of Theorem 3}
Recall that $\hat{\sigma}(\cdot |(\lambda, \phi))$ is adaptive submodular with respect to realization distribution $p((\lambda, \phi))$. From the general result of Theorem A.10 in [27], we see that in each round if we obtain a $Q(L; R)$ of $\alpha$ approximation ratio of the optimal solution, then the greedy policy $\pi^{\text{greedy}}_{\text{discrete}}$ achieves a $1-e^{-\alpha\frac{B_1-1}{B_1}}$ approximation of the optimal policy $\pi^{\text{OPT}}_{\text{discrete}}$. According to Theorem 1 in [29], since $Q(L; R)$ is monotone and submodular, the approximation ratio of the greedy algorithm in stage 2 is $\frac{1}{2}(1-e^{-1})$, i.e., $\alpha=\frac{1}{2}(1-e^{-1})$. Thus, the approximation ratio of $\pi^{\text{greedy}}_{\text{discrete}}$ is $1-e^{-\frac{B_1-1}{2B_1}(1-\frac{1}{e})}$.

\section{Estimation of $I(S)$}

      Influence estimation is frequently demanded in IM algorithms to decide the discount allocation. However, it is proven to be $\#$P-hard [S1] [S2] and becomes an obstacle of influence maximization. Thanks to the efforts of Tang \emph{et al.} [7] [21], a polling based framework called IMM is proposed for efficient influence estimation, with which the time complexity of greedy algorithm is only $O((k+l)(n+m)\log n/\epsilon^2)$. Due to its favorable performance, we apply the framework to estimate influences when designing allocations, just like many other works, e.g., [8] [9] [S3]. To help readers comprehend the implementation of our experiment, we would like to briefly introduce the framework below.

        \textbf{Definition S1.} \textit{\textbf{(Reverse Reachable Set) [7] [21]} Let $v$ be a node in $V$. A reverse reachable (RR) set for $v$ is generated by first sampling a graph $g$ from $G$, and then taking the set of nodes in $g$ that can reach $v$. A random RR set is an RR set for a node selected uniformly at random from $V$.       }


        \textbf{Lemma S1.} ([S4]) \textit{For any seed set $S$ and any node $v$, the probability that a diffusion process from $S$ can activate $v$ equals the probability that $S$ overlaps an RR set for $v$.}

        Given network $G(V,E)$ and propagation probabilities of each edge, we first derive the transpose graph of $G$ defined as $G^T(V,E^T)$, i.e., edge $(v,u)\in G^T$ iff. $(u,v)\in G$. Note that the propagation probability of edge $(v,u)$ remains to be $p_{uv}$. The crux of estimating $I(S)$ lies in generating $\theta$ random RR sets defined in Definition S1. Briefly speaking, to generate an RR set, we select a node (e.g., $v$) uniformly at random from $V$ and stimulate a propagation from $v$. The intuition is that by reverse propagation, we could find out which set of users could potentially influence $v$. From Lemma S1, it could be inferred that in expectation the influence spread $I(S)$ is equal to the fraction of RR sets covered by $S$ times the number of users $|V|$ [7] [21]. Denoting the number of RR sets covered by seed set $S$ as $F_R(S)$, we have $E[I(S)]=E[F_R(S)|V|/\theta]$.

        As can be seen, to estimate the influence spread, the only parameter in need of specification is $\theta$, i.e., the number of RR sets to be generated, which is usually in $O(n\log n)$ [8]. In our experiment, for datasets same with previous works, we check their $\theta$ according to the constraints in [7]. If their setting is appropriate, we adopt the same $\theta$ as them. Specifically, ``wiki-Vote'' is adopted in [8] where $\theta=0.25M$; ``com-Dblp'' is applied in [8] and [S5], where both works set $\theta$ to be 20M; and ``soc-Livejournal'' is also tested in [S5], where $\theta=40M$. For the dataset ``ca-CondMat'' adopted by us, we find the parameter following the process in [7].

        To illustrate the setting of $\theta$, we reproduce relevant content in [7] as follows.

        According to Theorem 1 in [7], to obtain an approximation of error $\epsilon$ with probability $1-1/n^l$, $\theta$ should be no less than
        \begin{equation*}
          \theta=\frac{2n[(1-1/e)\cdot \alpha +\beta]^2}{\text{\text{OPT}}\epsilon^2},
        \end{equation*}
        where OPT is the maximum expected influence spread,
        \begin{equation*}
          \alpha=\sqrt{l\log n+\log 2} \text{, and}
        \end{equation*}
        \begin{equation*}
         \beta=\sqrt{(1-1/e)\cdot[\log C_n^k+l\log n+\log 2]}.
        \end{equation*}

        As we know, it is almost impossible to obtain OPT. Thus, we substitute OPT with a lower bound of the influence spread. Accordingly, $\theta$ is larger than it is under OPT. Since $\theta$ is expected to be large, the estimation is even more precise after substitution. With this idea, we check the setting of previous works which adopt identical datasets and find their values of $\theta$ are reasonable. Regarding the dataset ``ca-CondMat'', we consider the most demanding case $k=50$ and set $\epsilon$ to be 0.05 for rigor. The resultant $\theta$ is 1.4M. For safety, we set $\theta$ to be 2M as shown in the main paper.

        The above basic technique could apply to the adaptive case well, where we only need to estimate the influence of a deterministic seed set (in each round the seeded user will certainly become the seed or not, accepting or refusing the discount). However, for the non-adaptive case, given an allocation, since no observation is made, the state of users is probabilistic and accordingly the influence spread is related to $p_u(c_u)$. When estimating the influence, we need to apply Theorem 9 in [8]. 
        Specifically, when computing $\hat{f}(C_1;X)$ via Equation (10), we need to estimate the value of $\max \hat{Q}(C_2;\cdot)$. By running Alg. 1, we could optimize $\hat{Q}(C_2;\cdot)$ and obtain the corresponding allocation (assumed to be $C$). Then, by Theorem 9 in [8], we know that the expected influence is $n\cdot [\sum_{h\in \text{RR sets}}1-\prod_{u\in h}(1-p_u(c_u))]/\theta$. Since $p_u(c_u)$ is available, the value of the expression could be derived. Accordingly, the value of $\hat{f}(C_1;X)$ could be derived by Eq. (10).

        We would also like to mention that the IMM method is only to estimate the influence when designing allocations. However, when experimentally testing the performance of algorithms, the influence spread is obtained by running 20K times Monte Carlo simulations, as indicated in Sec. 7.1.

\section{More Experimental Results}
In this section, we report the influence spread (Fig. \ref{Influence_spread_in_Non-ada}, Fig. \ref{Influence_spread_in_Ada}) and running time (Fig. \ref{Running_Time_in_Non-ada}, Fig. \ref{Run_time_in_Ada}) under $\alpha=0.8$.

  \begin{figure*}[h]
    \centering
    \subfigure[Wiki-Vote]
        {
            \begin{minipage}[h]{0.24\textwidth}
            \centerline{\includegraphics[width=1\textwidth]{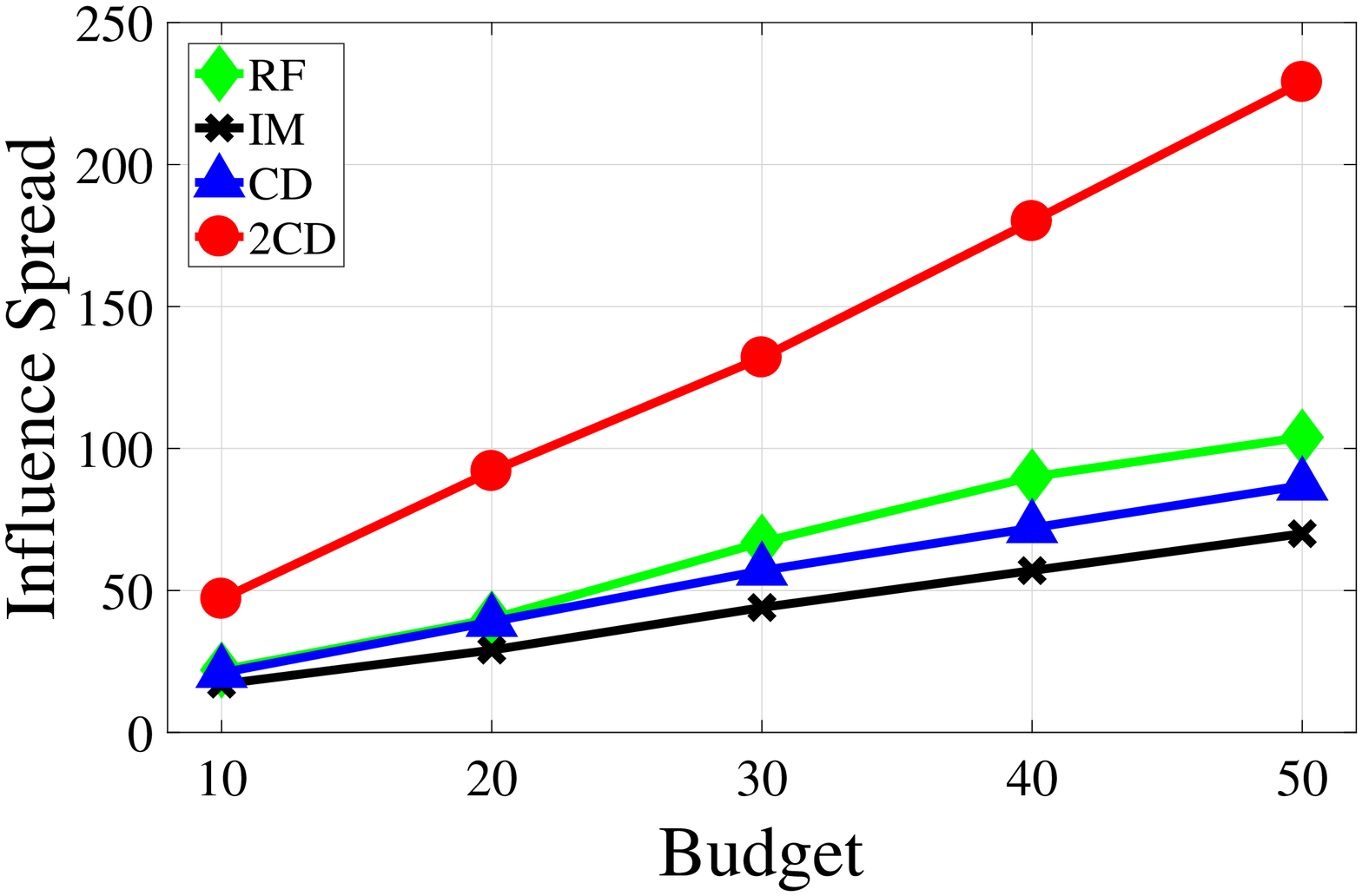}\label{wiki-vote_Non-ada_alpha=08}}
            \end{minipage}
        }
    \hspace{-3mm}
    \subfigure[Ca-CondMat]
        {
            \begin{minipage}[h]{0.24\textwidth}
            \centerline{\includegraphics[width=1\textwidth]{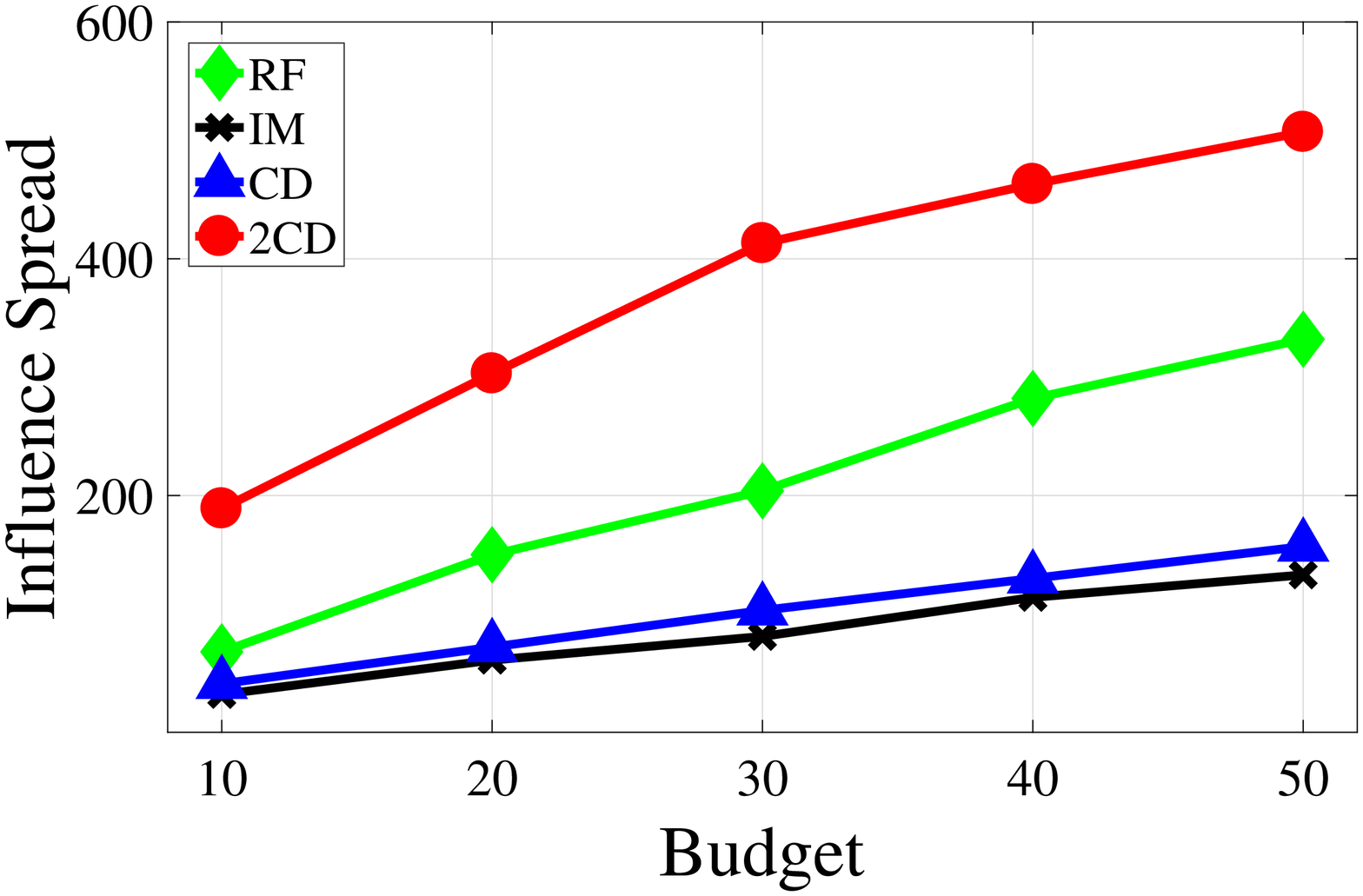}\label{Condmat_Non-ada_alpha=08}}
            \end{minipage}
            \hspace{-1mm}
        }
    \hspace{-3mm}
      \subfigure[com-Dblp]
        {
            \begin{minipage}[h]{0.24\textwidth}
            \centerline{\includegraphics[width=1\textwidth]{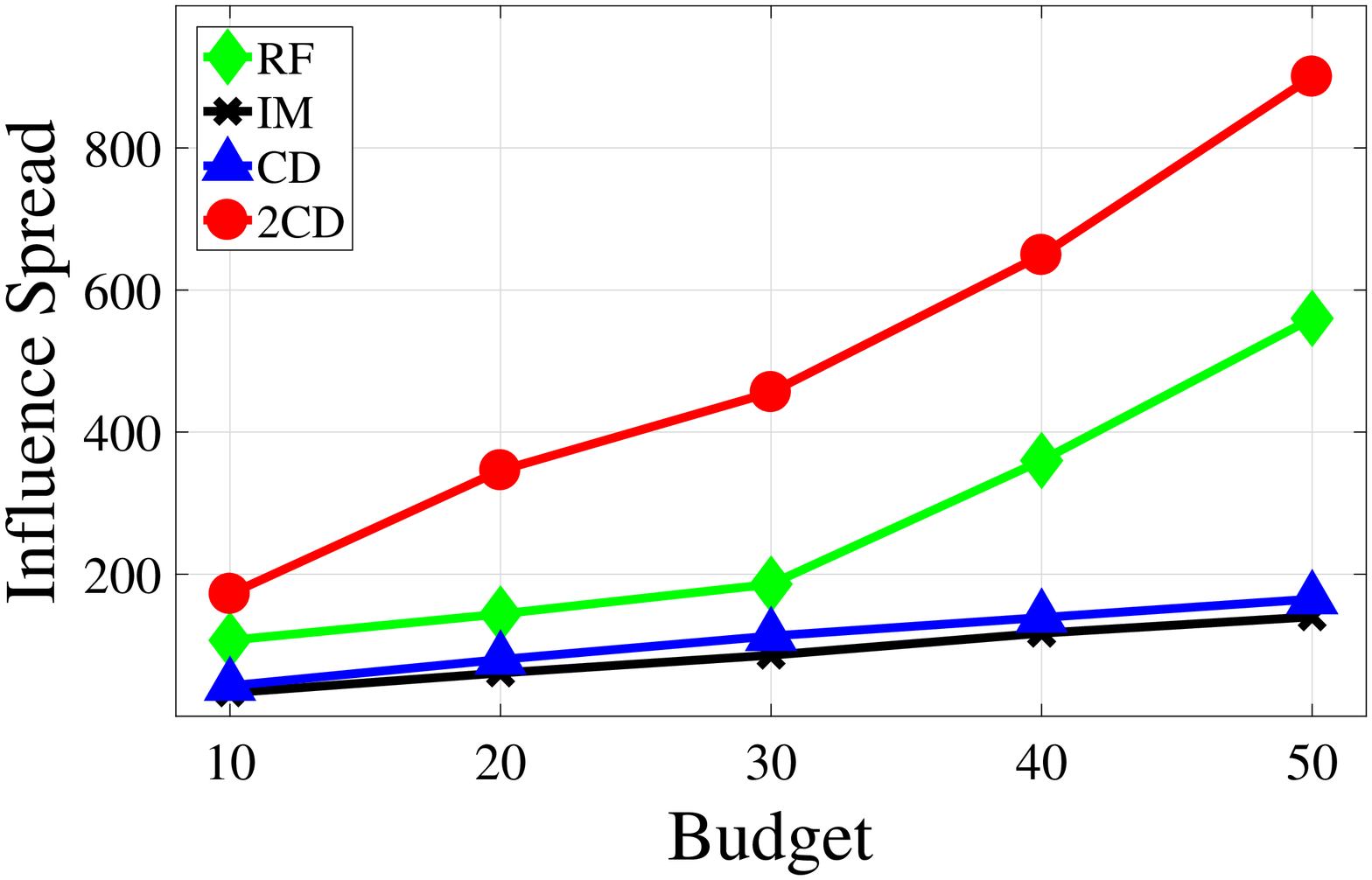}\label{Dblp_Non-ada_alpha=08}}
            \end{minipage}

        }
    \hspace{-3mm}
      \subfigure[soc-Livejournal]
        {
            \begin{minipage}[h]{0.24\textwidth}
            \centerline{\includegraphics[width=1\textwidth]{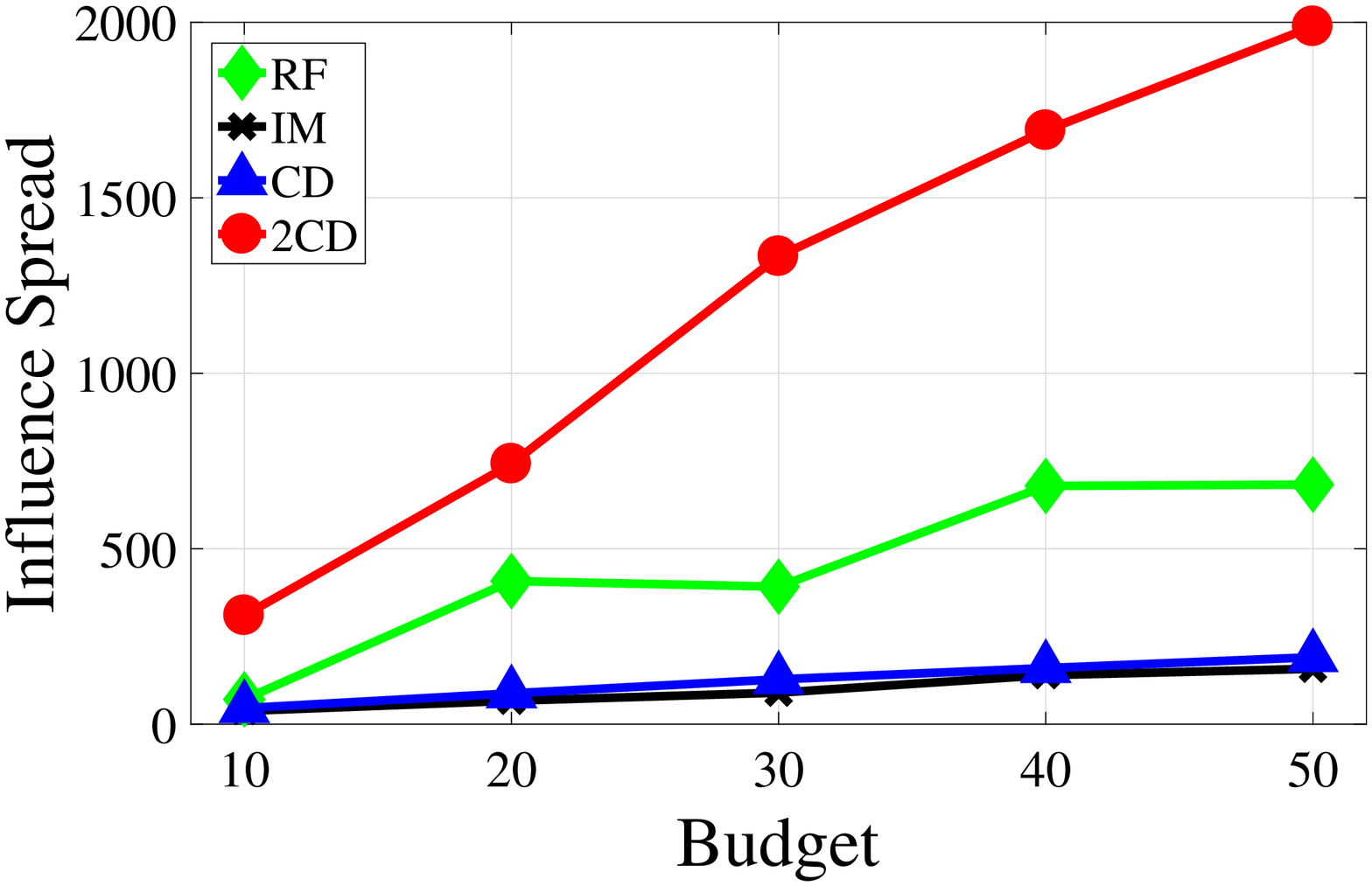}\label{livejournal_Non-ada_alpha=08}}
            \end{minipage}
            \hspace{-1mm}
        }
      \vspace{-1.5mm}
      \caption{Influence Spread in the Non-adaptive Case ($\alpha$=0.8).}\label{Influence_spread_in_Non-ada}
      \vspace{-2mm}
  \end{figure*}

  \begin{figure*}[h]
    \centering
    \subfigure[Wiki-Vote]
        {
            \begin{minipage}[h]{0.24\textwidth}
            \centerline{\includegraphics[width=1\textwidth]{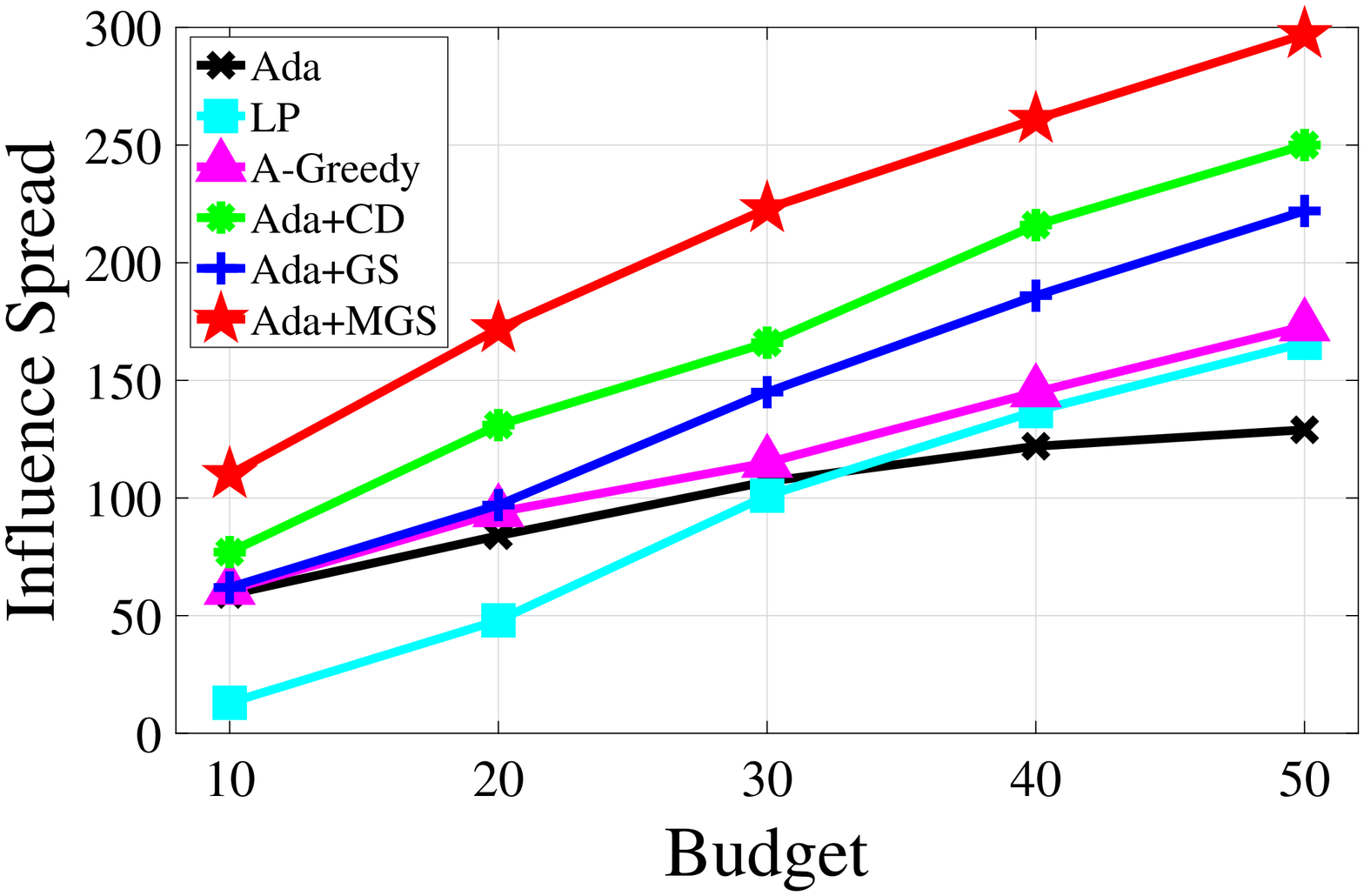}\label{wiki-vote_Ada_alpha=08}}
            \end{minipage}
        }
    \hspace{-3mm}
    \subfigure[Ca-CondMat]
        {
            \begin{minipage}[h]{0.24\textwidth}
            \centerline{\includegraphics[width=1\textwidth]{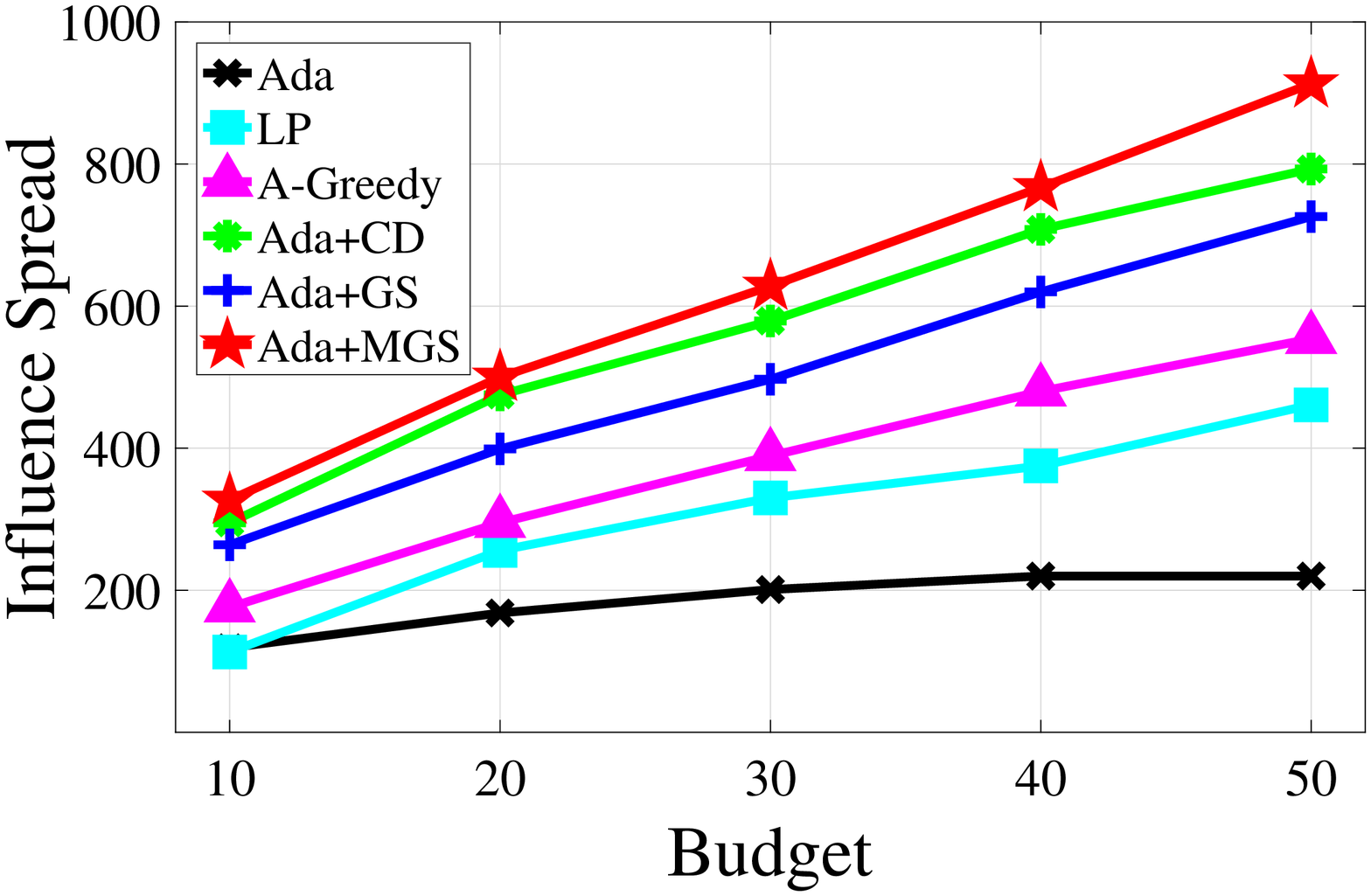}\label{Condmat_Ada_alpha=08}}
            \end{minipage}
        }
    \hspace{-3mm}
      \subfigure[com-Dblp]
      {
            \begin{minipage}[h]{0.24\textwidth}
            \centerline{\includegraphics[width=1\textwidth]{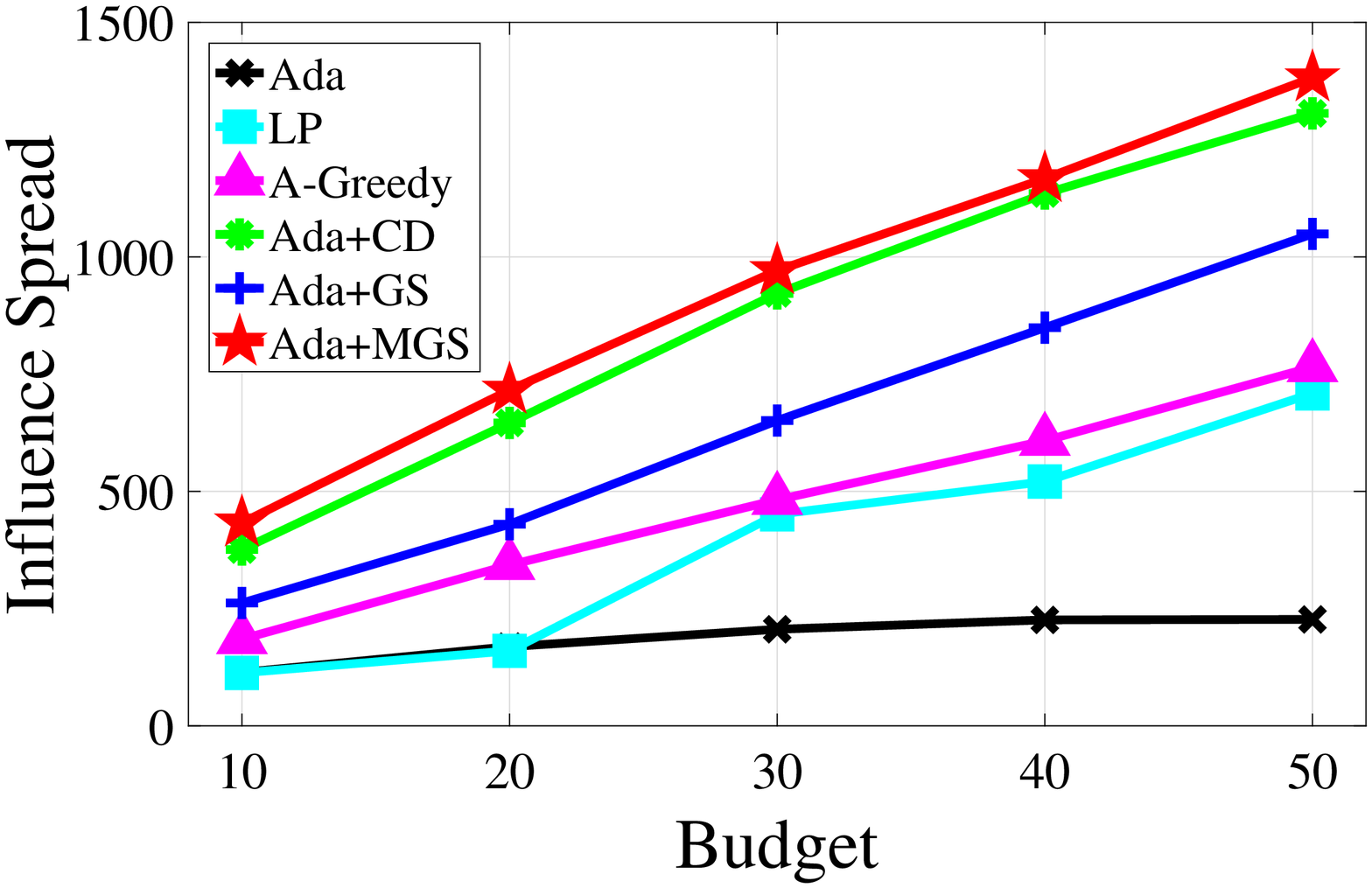}\label{Dblp_Ada_alpha=06}}
            \end{minipage}
        }
      \hspace{-3mm}
      \subfigure[soc-Livejournal]
      {
            \begin{minipage}[h]{0.24\textwidth}
            \centerline{\includegraphics[width=1\textwidth]{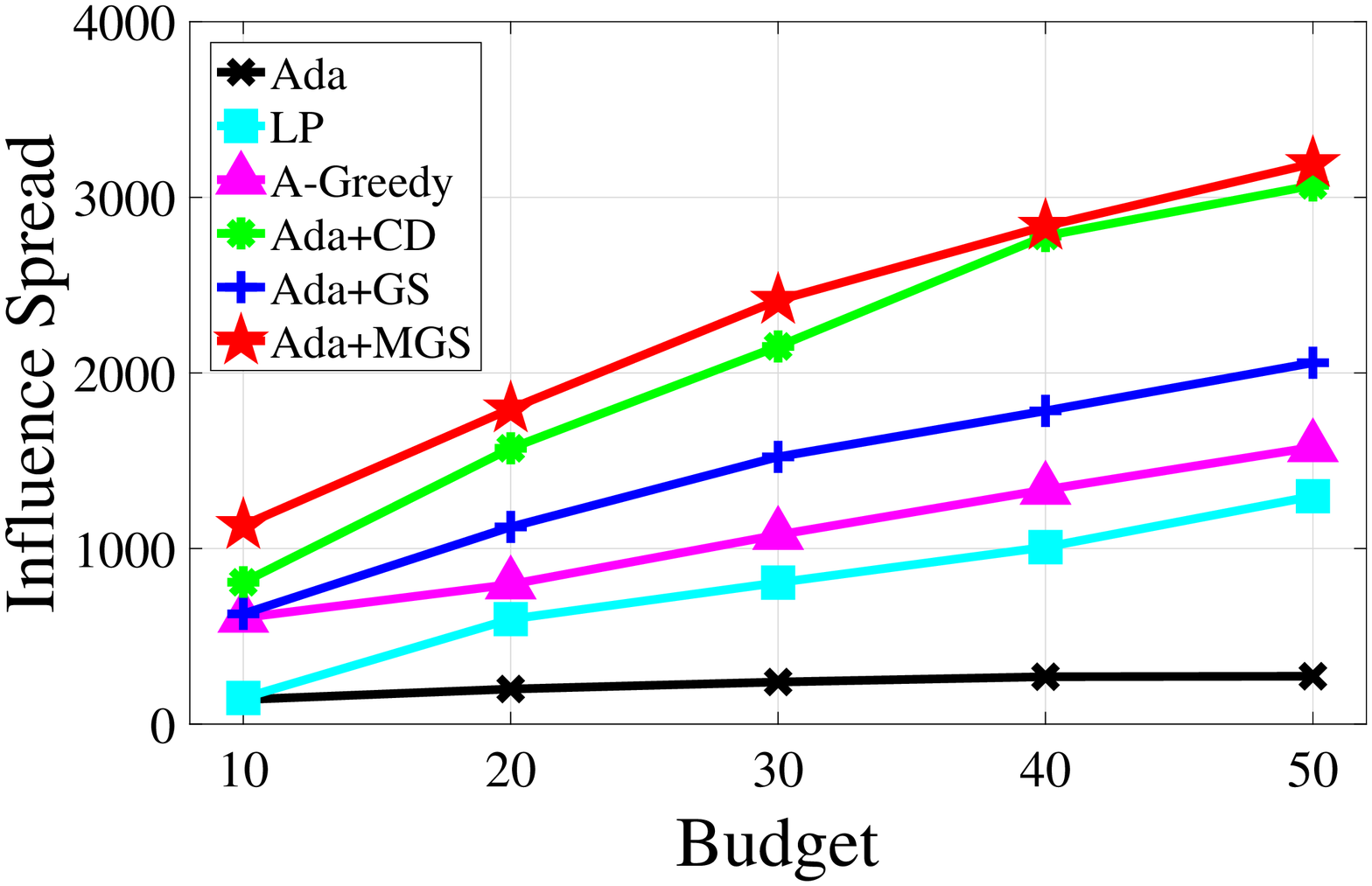}\label{livejournal_Ada_alpha=08}}
            \end{minipage}
        }
     \vspace{-1.5mm}
      \caption{Influence Spread in the Adaptive Case ($\alpha$=0.8).}\label{Influence_spread_in_Ada}
      \vspace{-2mm}
  \end{figure*}

  \begin{figure*}[h]
    \centering
    \subfigure[Wiki-Vote]
        {
            \begin{minipage}[h]{0.24\textwidth}
            \centerline{\includegraphics[width=1\textwidth]{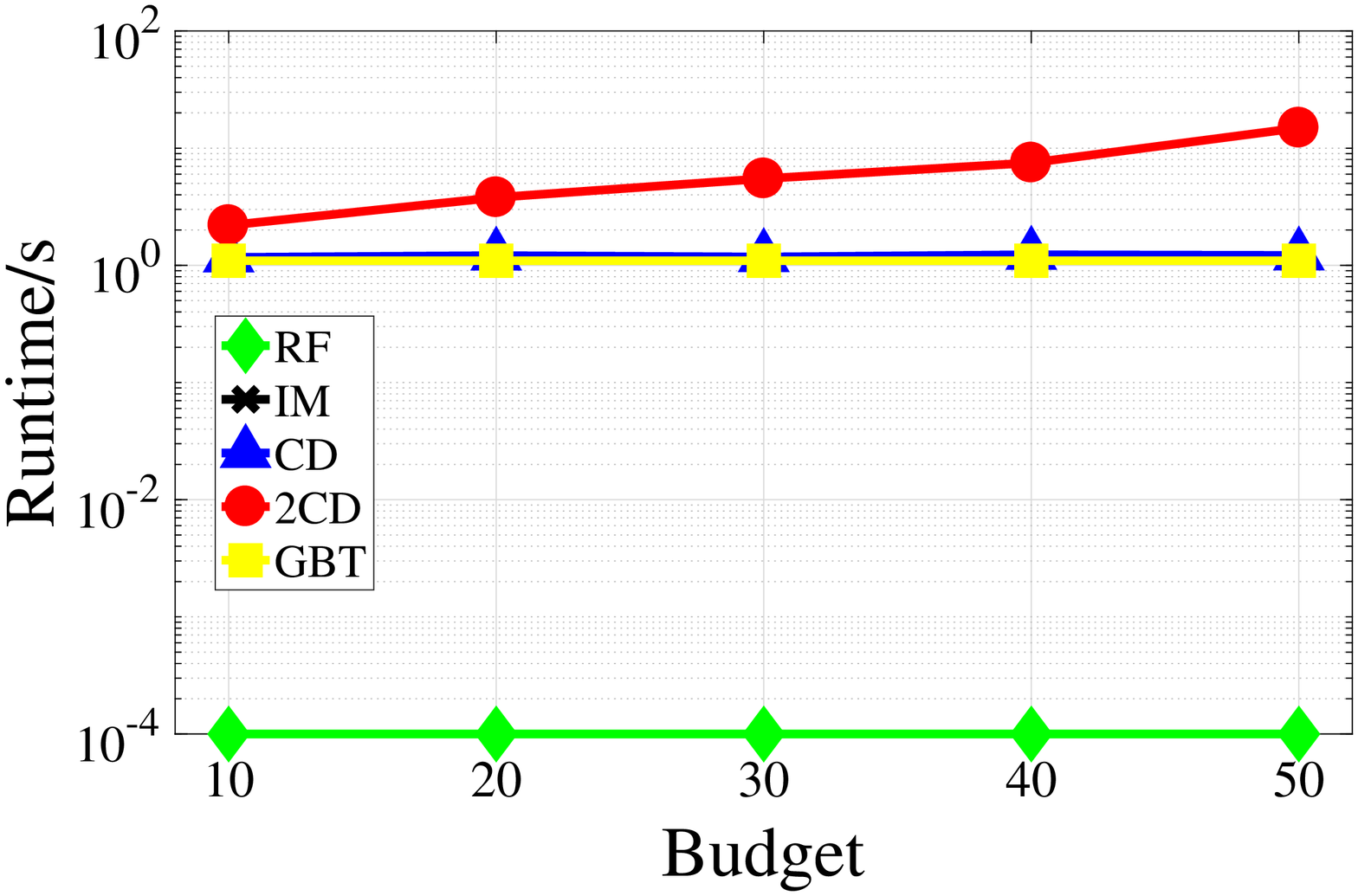}\label{wiki-vote_Non-ada_alpha-Runtime=08}}
            \end{minipage}
        }
    \hspace{-3mm}
    \subfigure[Ca-CondMat]
        {
            \begin{minipage}[h]{0.24\textwidth}
            \centerline{\includegraphics[width=1\textwidth]{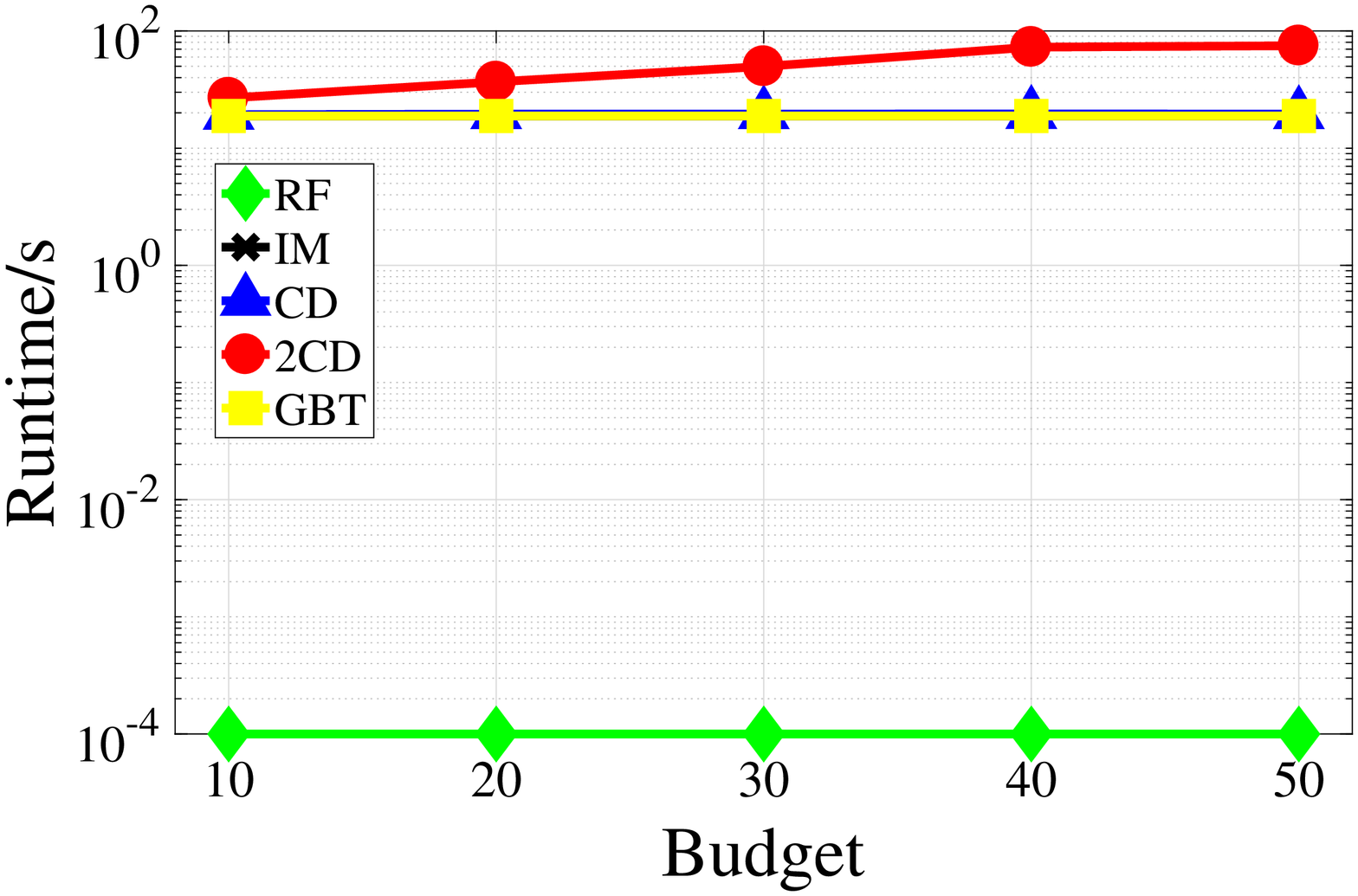}\label{Condmat_Non-ada_alpha-Runtime=08}}
            \end{minipage}
        }
    \hspace{-3mm}
      \subfigure[com-Dblp]
      {
            \begin{minipage}[h]{0.24\textwidth}
            \centerline{\includegraphics[width=1\textwidth]{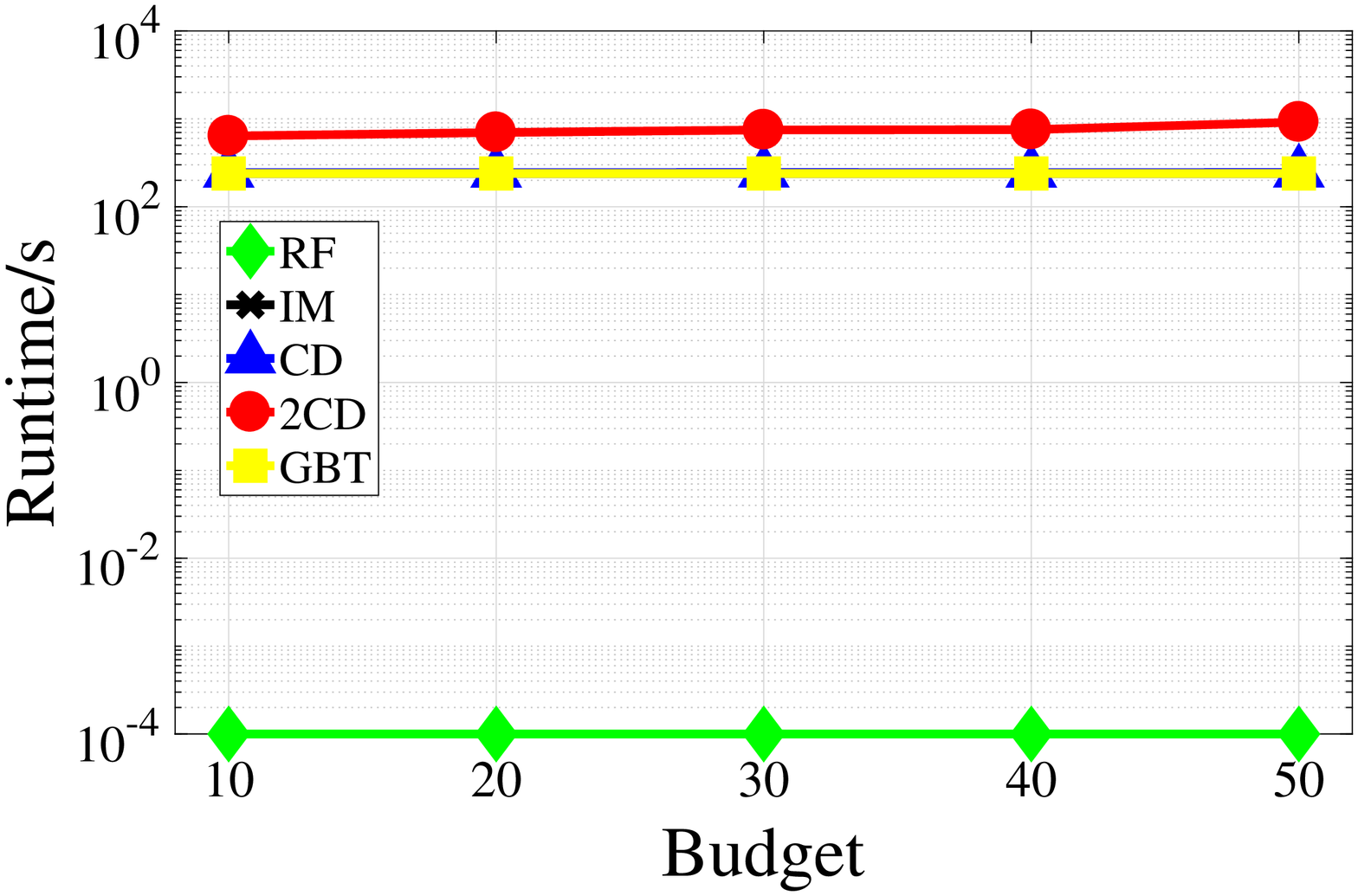}\label{Dblp_Non-ada_alpha-Runtime=08}}
            \end{minipage}
        }
      \hspace{-3mm}
      \subfigure[soc-Livejournal]
      {
            \begin{minipage}[h]{0.24\textwidth}
            \centerline{\includegraphics[width=1\textwidth]{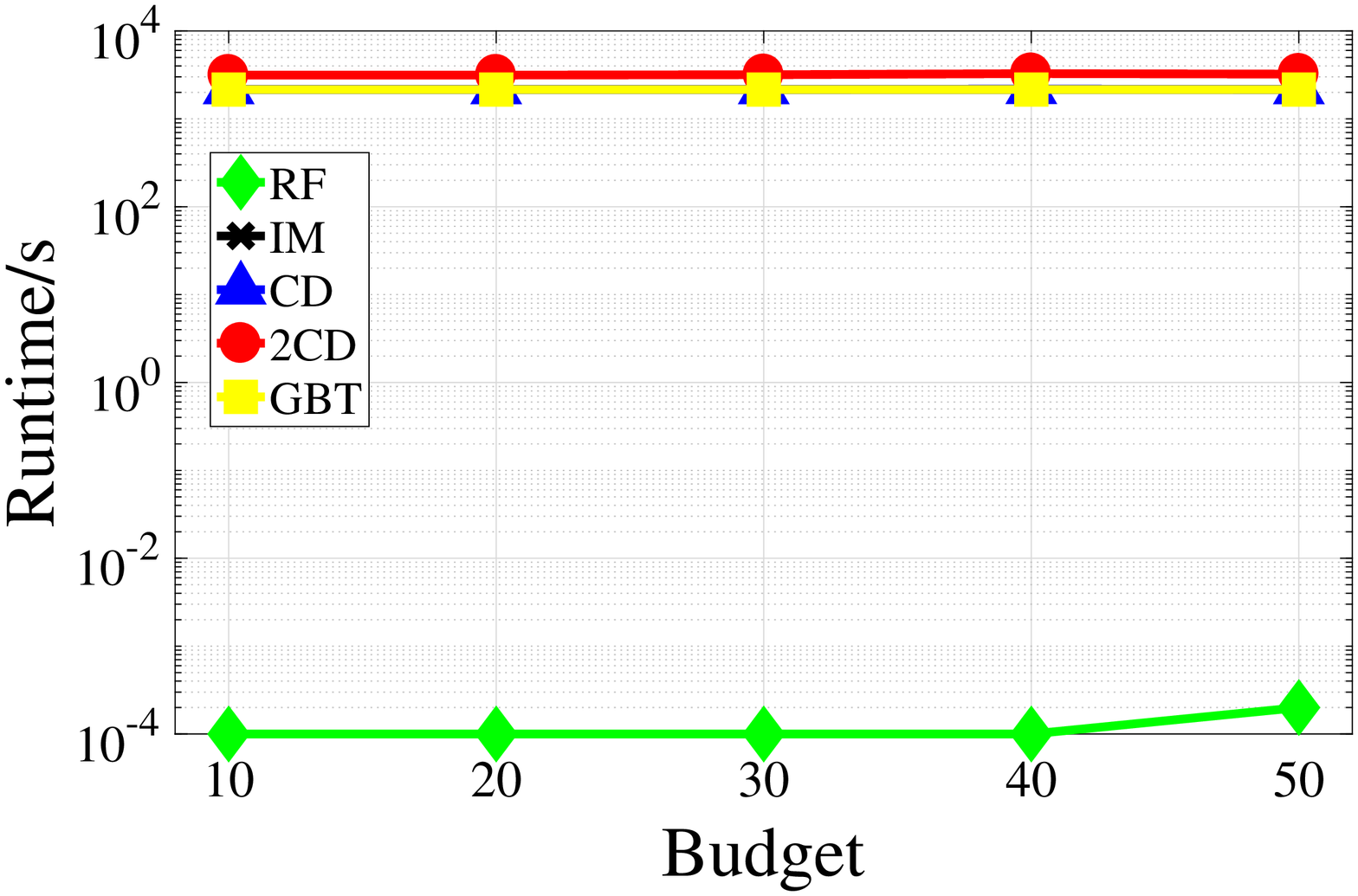}\label{livejournal_Non-ada_alpha-Runtime=08}}
            \end{minipage}
        }
     \vspace{-1.5mm}
      \caption{Running Time in the Non-adaptive Case ($\alpha$=0.8).}\label{Running_Time_in_Non-ada}
      \vspace{-2mm}
  \end{figure*}

  \begin{figure*}[h]
    \centering
    \subfigure[Wiki-Vote]
        {
            \begin{minipage}[h]{0.24\textwidth}
            \centerline{\includegraphics[width=1\textwidth]{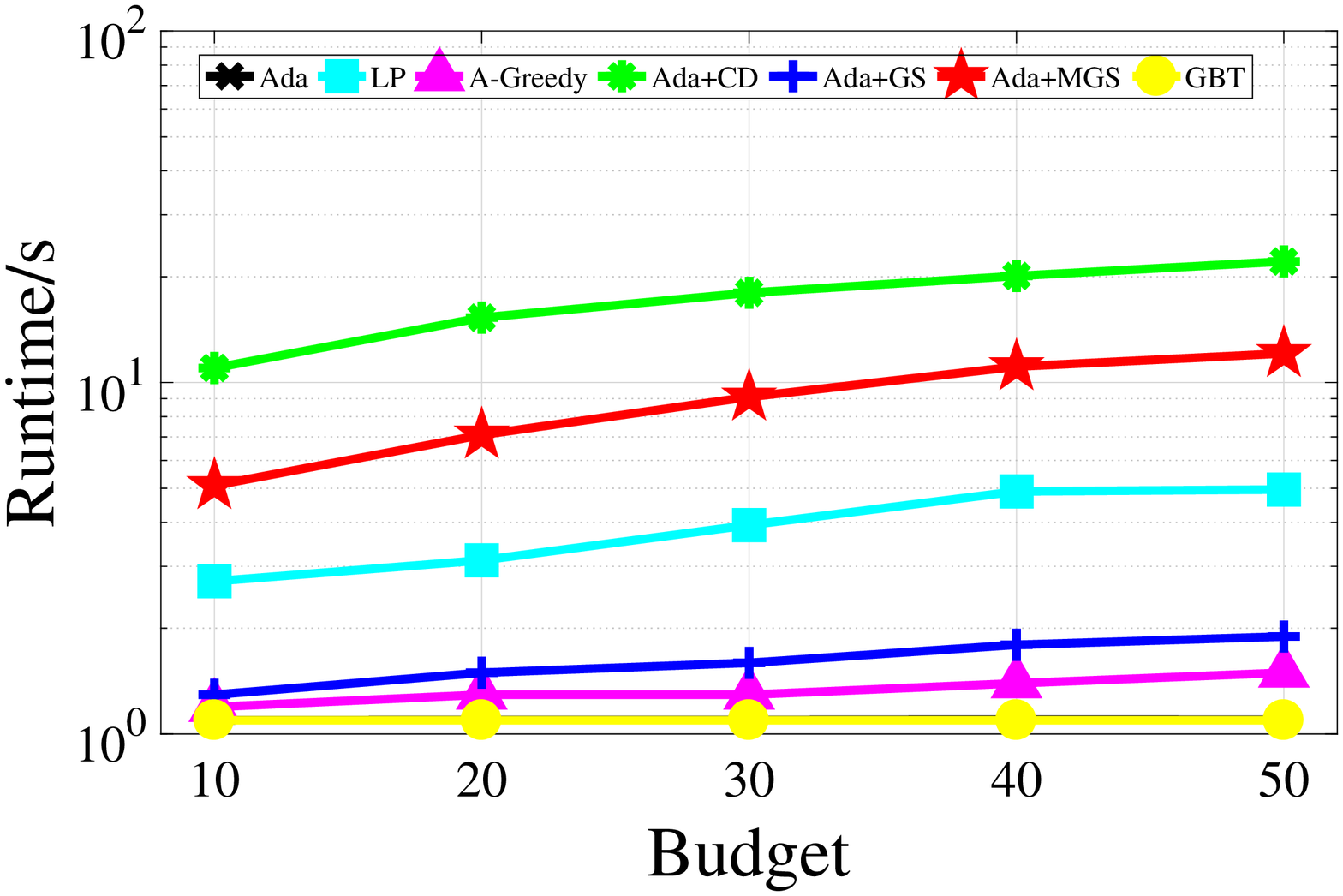}\label{wiki-vote_Ada_alpha=08-Runtime}}
            \end{minipage}
        }
    \hspace{-3mm}
    \subfigure[Ca-CondMat]
        {
            \begin{minipage}[h]{0.24\textwidth}
            \centerline{\includegraphics[width=1\textwidth]{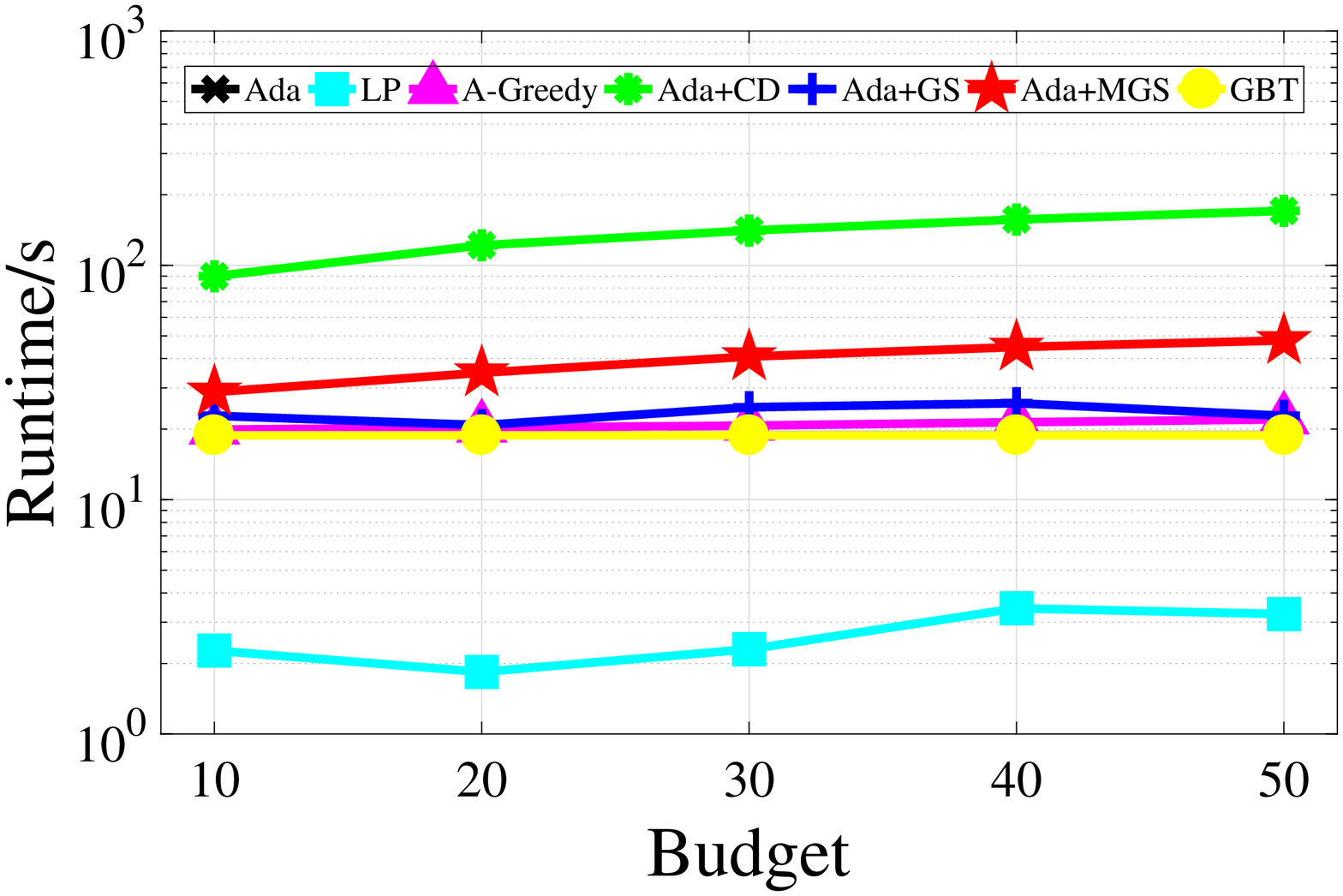}\label{Condmat_Ada_alpha=08-Runtime}}
            \end{minipage}
        }
    \hspace{-3mm}
      \subfigure[com-Dblp]
      {
            \begin{minipage}[h]{0.24\textwidth}
            \centerline{\includegraphics[width=1\textwidth]{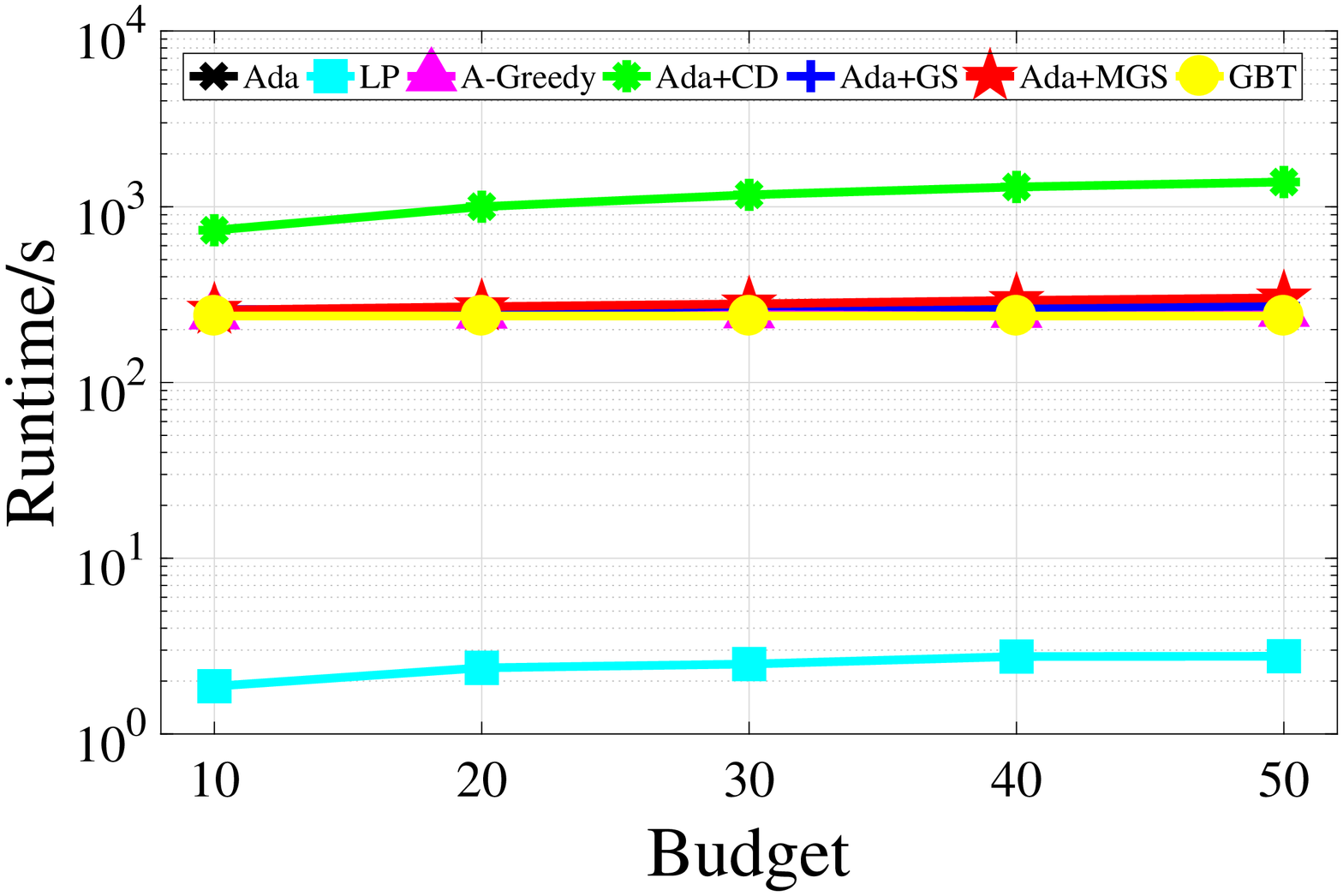}\label{Dblp_Ada_alpha=08-Runtime}}
            \end{minipage}
        }
      \hspace{-3mm}
      \subfigure[soc-Livejournal]
      {
            \begin{minipage}[h]{0.24\textwidth}
            \centerline{\includegraphics[width=1\textwidth]{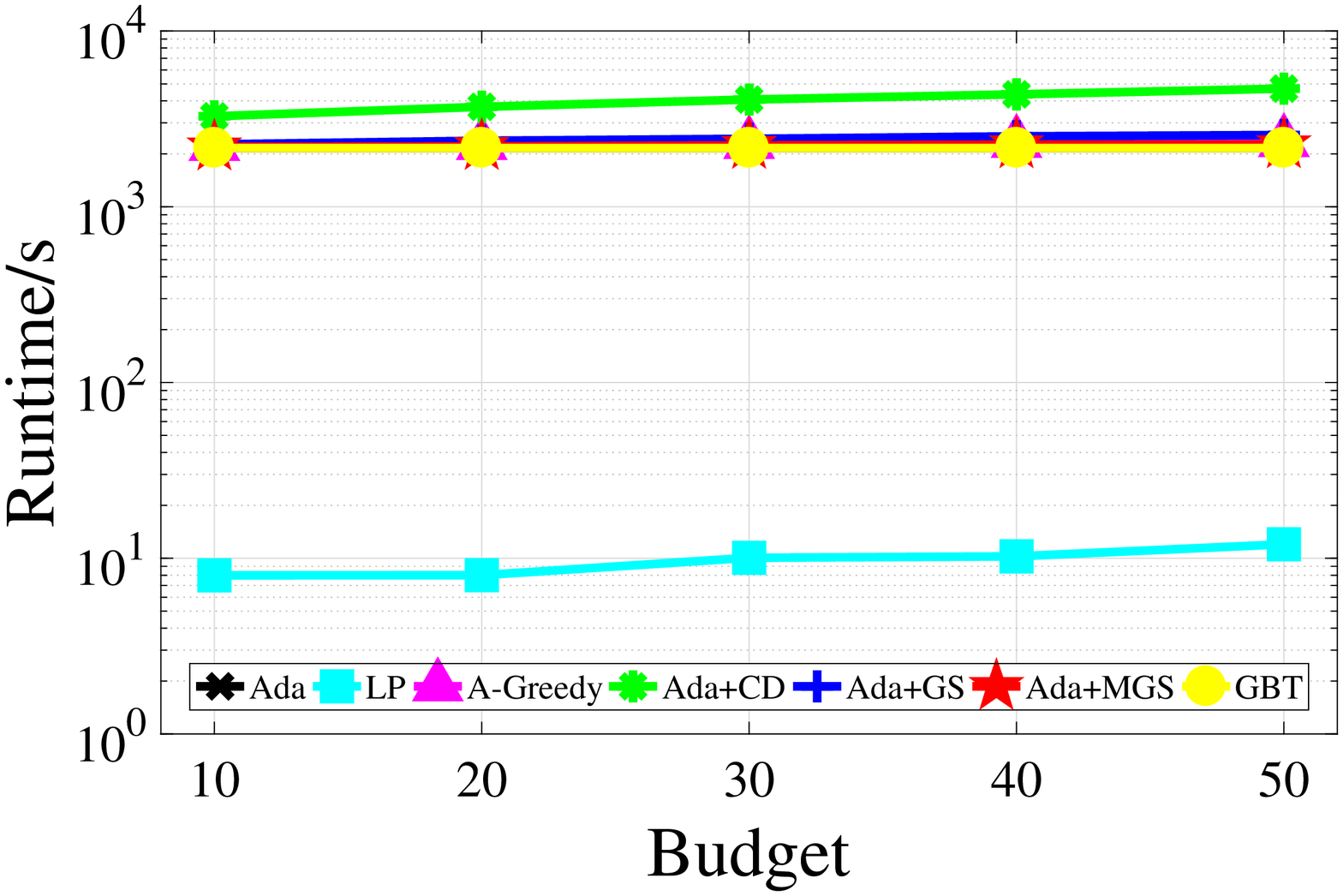}\label{livejournal_Ada_alpha=08-Runtime}}
            \end{minipage}
        }
      \vspace{-1.5mm}
      \caption{Running Time in the Adaptive Case ($\alpha$=0.8).}\label{Run_time_in_Ada}
      \vspace{-2mm}
  \end{figure*}

\section{References}
\footnotesize
        \quad\quad[S1] W. Chen, C. Wang, and Y. Wang, Scalable influence maximization for prevalent viral marketing in large-scale social networks. in \textit{Proc. Int. Conf. Knowl. Discov. Data Min., KDD}, ACM, 2010, pp. 1029--1038.

        [S2] W. Chen, Y. Yuan, and L. Zhang. Scalable influence maximization in social networks under the linear threshold model. in \textit{Proc. Int. Conf. Data Mining, ICDM}, IEEE, 2010, pp. 88--97.

        [S3] K. Huang, J. Tang, X. Xiao, A. Sun, and A. Lim. Efficient Approximation Algorithms for Adaptive Target Profit Maximization, in \textit{Proc. Int. Conf. Data Eng. ICDE}, IEEE, 2020, to appear.

        [S4] C. Borgs, M. Brautbar, J. Chayes, and B. Lucier. Maximizing social influence in nearly optimal time. in \textit{Proc. Annu. {ACM-SIAM} Symp. Discrete Algorithms}, 2014, pp. 946--957.

        [S5] J. Yuan and S. Tang. Adaptive Discount Allocation in Social Networks. \url{https://www.dropbox.com/s/tlm5ldll0md6lu2/TRmobihoc.pdf?dl=0}. (2017).
\normalsize
\end{document}